\documentclass[11pt]{article}
\usepackage[utf8]{inputenc}
\usepackage{hyperref}

\usepackage{amsmath}
\usepackage{amssymb}
\usepackage{amsthm}
\usepackage{fullpage}
\usepackage{verbatim}
\usepackage{cleveref}

\usepackage{times}
\usepackage{tikz}
\usetikzlibrary{fadings, patterns, decorations}

\newtheorem{theorem}{Theorem}
\newtheorem{lemma}[theorem]{Lemma}
\newtheorem{lem}[theorem]{Lemma}
\newtheorem{defn}[theorem]{Definition}

\usepackage{tikz}
\usetikzlibrary{arrows.meta}
\usetikzlibrary{patterns}
\usetikzlibrary{decorations.pathreplacing, calligraphy}

\renewcommand{\epsilon}{\varepsilon}

\usepackage[colorinlistoftodos,textsize=tiny,textwidth=2cm,color=green!50!gray]{todonotes} 

\title{A $(2+\varepsilon)$-approximation algorithm for the general scheduling problem in quasipolynomial time}
 \author{Alexander Armbruster \and Lars Rohwedder \and Andreas Wiese}
\date{}
\author{Alexander Armbruster\thanks{Technical University of Munich (\href{mailto:alexander.armbruster@tum.de}{alexander.armbruster@tum.de}). Funded by the Deutsche Forschungsgemeinschaft (DFG, German Research Foundation) – Project Number 551896423.}
	\and Lars Rohwedder\thanks{University of Southern Denmark (\href{mailto:rohwedder@sdu.dk}{rohwedder@sdu.dk}). Supported by Dutch Research Council
		(NWO) project ``The Twilight Zone of Efficiency: Optimality of Quasi-Polynomial Time Algorithms" [grant number OCEN.W.21.268]}    
	\and Andreas Wiese\thanks{Technical University of Munich (\href{mailto:andreas.wiese@tum.de}{andreas.wiese@tum.de}). Funded by the Deutsche Forschungsgemeinschaft (DFG, German Research Foundation) – Project Number 551896423.}}

\global\long\def\lef{\mathrm{left}}%
\global\long\def\mi{\mathrm{mid}}%
\global\long\def\ri{\mathrm{right}}%
\global\long\def\OPT{\mathrm{OPT}}%
\global\long\def\APX{\mathrm{APX}}%
\global\long\def\filled{\mathrm{filled}}%
\global\long\def\N{\mathbb{N}}%
\global\long\def\R{\mathcal{R}}%
\global\long\def\rows{\mathcal{W}}%
\global\long\def\rsub{\mathrm{right}}%
\global\long\def\lsub{\mathrm{left}}%
\global\long\def\W{\mathcal{W}}%
\global\long\def\ce{\mathrm{center}}%
\global\long\def\add{\mathrm{add}}%
\global\long\def\Q{\mathcal{Q}}%
\global\long\def\spn{\mathrm{span}}%
\global\long\def\cost{\mathrm{cost}}%
\global\long\def\RR{\mathbb{R}}%
\global\long\def\Z{\mathbb{Z}}%
\global\long\def\L{\mathcal{L}}%

\global\long\def\cross{\mathrm{cross}}%
\global\long\def\done{\mathrm{done}}%
\global\long\def\G{\mathcal{G}}%

\DeclareMathOperator{\projy}{proj_y}

\global\long\def\cost{\mathrm{cost}}%

\definecolor{fillcolor}{RGB}{235,235,235}
\begin{tikzfadingfrompicture}[name=myfading]   
	\clip (0,0) rectangle (2,2);   
	\shade [top color=transparent!50, bottom color=transparent!50] (0,0) rectangle (2,1.5);   
	\shade [top color=transparent!100, bottom color=transparent!50] (0,1.5) rectangle (2,1.85);                 
	\shade [top color=transparent!100, bottom color=transparent!100] (0,1.85) rectangle (2,2);                                  
\end{tikzfadingfrompicture}

\begin{document}

\maketitle

\begin{abstract}
We study the general scheduling problem (GSP) which generalizes and
unifies several well-studied preemptive single-machine scheduling problems, such as weighted flow
time, weighted sum of completion time, and minimizing the total weight
of tardy jobs. We are given a set of jobs with their processing times
and release times and seek to compute a (possibly preemptive) schedule
for them on one machine. Each job incurs a cost that depends
on its completion time in the computed schedule, as given
by a separate job-dependent cost function for each job, and our objective is
to minimize the total resulting cost of all jobs. The best known result for
GSP is a polynomial time $O(\log\log P)$-approximation algorithm
[Bansal and Pruhs, FOCS 2010, SICOMP 2014].

We give a quasi-polynomial time $(2+\epsilon)$-approximation algorithm
for GSP, assuming that the jobs' processing times are quasi-polynomially
bounded integers. For the special case of the weighted tardiness objective,
we even obtain an improved approximation ratio of $1+\epsilon$. For this
case, no better result had been known than the mentioned $O(\log\log P)$-approximation
for the general case of GSP.
	Our algorithms use a reduction to an auxiliary geometric covering problem.
	In contrast to a related reduction for the special case of weighted flow time
	[Rohwedder, Wiese, STOC~2021][Armbruster, Rohwedder, Wiese, STOC~2023]
	for GSP it seems no longer possible to establish a tree-like structure for the rectangles to guide an algorithm that solves this geometric problem.
	Despite the lack of structure due to the problem itself,
	we show that an optimal solution can be transformed into a near-optimal solution that has certain structural properties. Due to those we can guess a substantial part of the solution quickly and partition the remaining problem in an intricate way, such that we can independently solve each part recursively.

\end{abstract}

\thispagestyle{empty}
\newpage
\setcounter{page}{1}

\section{Introduction}

We consider the following general and fundamental scheduling problem:
we are given one machine and a set of jobs $J$ where each job $j\in J$
is characterized by a processing time $p_{j}\in\N$ and a release
time $r_{j}\in\N$. We seek to compute a (possibly preemptive) schedule
for the jobs in $J$, i.e., each job $j\in J$ is processed for $p_{j}$
time units in total such that $j$ is not processed before time $r_{j}$
and the machine works on at most one job at a time. Ideally, we would like to
finish each job as early as possible. To quantify this,
it is natural to associate a cost to each job $j\in J$ in the computed
schedule which depends on the completion time of $j$; the later $j$
completes, the higher should be the cost for $j$.

In practical settings, the given jobs might be highly heterogeneous:
some jobs might be very important or urgent and, hence, need to be
finished quickly. They may even have a hard deadline before which
they must finish in \emph{any} feasible schedule. Other jobs might
be less critical and even large delays may be tolerable for them at
(almost) no cost. Therefore, it makes sense to allow each job $j$
to have a separate function that defines its cost, depending
on its completion time. For each job $j$ we denote this function by
$\cost_{j}:[r_j, \infty)\rightarrow\RR_{\geq 0} \cup \{\infty\}$
and we assume it (only) to be non-decreasing; given a schedule we denote by $C_{j}$ the completion time of $j$ which then
yields a cost of $\cost_{j}(C_{j})$.
We assume that we have access the jobs' cost functions via suitable oracles
 (similarly as in \cite{DBLP:journals/siamcomp/BansalP14}, see Section~\ref{sec:algorithmic-framework} for details).
Overall, we seek to
minimize the total resulting cost of our jobs, i.e., we want to minimize
$\sum_{j\in J}\cost_{j}(C_{j})$. This defines the General Scheduling
Problem~(GSP) as introduced by Bansal and Pruhs~\cite{DBLP:journals/siamcomp/BansalP14}, who
gave a polynomial time
$O(\log\log P)$-approximation algorithm for it. Here, $P$ denotes the ratio
between the largest and the smallest job processing times in the input.
It is open whether GSP admits a constant factor approximation algorithm, or
even a PTAS. For the approach used in~\cite{DBLP:journals/siamcomp/BansalP14}
it seems unlikely that it can yield a constant factor approximation, as argued in \cite{Batra0K18-journal}.
There has been no progress for the general case of GSP since the work by Bansal and Pruhs~\cite{DBLP:journals/siamcomp/BansalP14}.

Instead, most research in the context of GSP has focused on special cases of the problem in which the jobs' cost
functions are very structured. For example, in the weighted sum of
completion times objective each job $j$ has a weight $w_{j}>0$
and we seek to minimize $\sum_{j\in J}w_{j}C_{j}$. Hence, $\cost_{j}(t)=w_{j}\cdot t$
for each job $j$. This setting admits a polynomial time approximation
scheme (PTAS) as shown by Afrati et~al.~\cite{afrati1999approximation}. Another well-studied
special case is the weighted flow time objective~\cite{Batra0K18-journal, feige2019polynomial, RohwedderW21, armbruster2023ptas}.
For each job $j$ its flow time in a computed schedule is the time
between its release and its completion, i.e., $C_{j}-r_{j}$, and
we seek to minimize $\sum_{j\in J}w_{j}(C_{j}-r_{j})$.
The difference
to the weighted sum of completion times objective is ``only'' the
fixed term $\sum_{j\in J}w_{j}r_{j}$ and, therefore, the optimal solutions are the
same for both objectives. However, approximation ratios of non-optimal solutions are not preserved and
constructing approximation algorithms for the weighted flow time objective is much more challenging. For example, it had been a long-standing
open question to find even a constant factor approximation algorithm
with polynomial running time (quasi-polynomial time algorithms for
bounded input data had been known earlier though~\cite{chekuri2002approximation}). In a breakthrough
result Batra, Kumar, and Garg \cite{Batra0K18-journal} presented an $O(1)$-approximation
algorithm with pseudopolynomial running time which was subsequently
improved to polynomial time by Feige et al.~\cite{feige2019polynomial}. Finally, the approximation
ratio was improved to $1+\epsilon$ by Armbruster, Rohwedder,
and Wiese~\cite{armbruster2023ptas}.

Intuitively, the mentioned algorithms for weighted flow time work with discretized candidate completion times for the jobs. They crucially
use that slight changes in a job's flow time change its cost only by a small factor and that this change is proportional for any two jobs with (almost) the same release times.
These properties extend to the objective of minimizing the weighted $\ell_p$-norm of the jobs' flow times for $p\ge 1$~\cite{Batra0K18-journal, ArmbrusterRW24}. However, they do not extend to arbitrary instances of GSP,
since there the cost function $\cost_j$ of a job $j$ can increase abruptly by large factors when
its flow time only increases marginally.
Already the weighted tardiness objective suffers from this issue since a job's cost can stay zero for some time (which may be different for two jobs with the same release time)
and only then start to increase. Therefore, it forms an interesting case to study; moreover, it is well-motivated in its own right.
Formally, in this objective we are given a due date $d_{j}\in \N_0$ for each job $j$ which is intuitively
a soft deadline for $j$. Since we might not be able to finish each
job $j$ before its due date, for each job $j$ we consider the time
between the completion of $j$ and $d_{j}$, i.e., $C_{j}-d_{j}$,
and seek to minimize $\sum_{j\in J}w_{j}\max\{C_{j}-d_{j},0\}$. No improvements over
the general result of GSP are known for weighted tardiness.

There is another limitation of the mentioned results for special cases of GSP: they assume that
\emph{all} jobs' cost functions have the \emph{same} structure.
However, they no longer work for heterogeneous sets of jobs, e.g., when
each job $j$ either has a hard deadline (the cost is zero if $j$ is completed before the deadline and~$\infty$ otherwise)
or its cost in the computed schedule equals its weighted flow time.
Even though both cases individually can be approximated well or even solved exactly in polynomial time, this is
not clear for the combination of both. This motivates searching for algorithms for the general case of GSP which can thus handle
such settings.

\subsection{Our contribution}
In this paper, we present a quasi-polynomial time $(2+\epsilon)$-approximation
algorithm for (the general case of) GSP, assuming that the jobs' processing times are quasi-polynomially bounded integers. We denote by $p_{\max} := \max_{j\in J}p_j$ the maximum job processing time of a given job.
\begin{theorem}\label{thm:main}
	For each $\epsilon>0$ there is a $(2 + \epsilon)$-approximation algorithm for the general scheduling problem with a running time of~$2^{\mathrm{poly}((1/\epsilon)^{1/\epsilon}\log(n + p_{\max}))}$.
\end{theorem}
For the special case of weighted tardiness, our approximation ratio improves to $1+\epsilon$.
\begin{theorem}\label{thm:main-tardiness}
	For each $\epsilon>0$ there is a $(1 + \epsilon)$-approximation algorithm for the weighted tardiness problem with a running time of~$2^{\mathrm{poly}((1/\epsilon)^{1/\epsilon}\log(n + p_{\max}))}$.
\end{theorem}

Since weighted tardiness is strongly $\mathsf{NP}$-hard,
this is the best possible approximation ratio in quasi-polynomial running time,
unless $\mathsf{NP}\subseteq\mathrm{DTIME}(2^{\mathrm{poly}(\log n)})$.
Note that our result also implies that the problem cannot be $\mathsf{APX}$-hard (with bounded input data), unless $\mathsf{NP}\subseteq\mathrm{DTIME}(2^{\mathrm{poly}(\log n)})$.

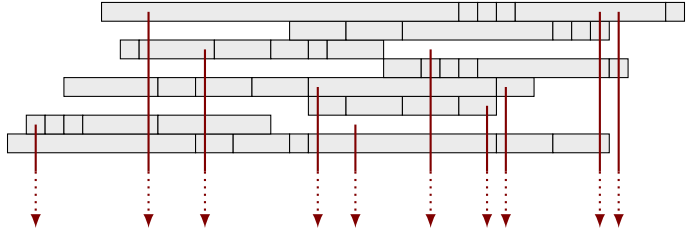
\begin{figure}
	\centering
	\begin{tikzpicture}[scale=0.25] 
		\draw[fill=fillcolor] (0,1) rectangle (10,2);
\draw[fill=fillcolor] (10,1) rectangle (12,2);
\draw[fill=fillcolor] (12,1) rectangle (15,2);
\draw[fill=fillcolor] (15,1) rectangle (16,2);
\draw[fill=fillcolor] (16,1) rectangle (26,2);
\draw[fill=fillcolor] (26,1) rectangle (29,2);
\draw[fill=fillcolor] (29,1) rectangle (32,2);
\draw[fill=fillcolor] (1,2) rectangle (2,3);
\draw[fill=fillcolor] (2,2) rectangle (3,3);
\draw[fill=fillcolor] (3,2) rectangle (4,3);
\draw[fill=fillcolor] (4,2) rectangle (8,3);
\draw[fill=fillcolor] (8,2) rectangle (14,3);
\draw[fill=fillcolor] (16,3) rectangle (18,4);
\draw[fill=fillcolor] (18,3) rectangle (21,4);
\draw[fill=fillcolor] (21,3) rectangle (24,4);
\draw[fill=fillcolor] (24,3) rectangle (26,4);
\draw[fill=fillcolor] (3,4) rectangle (8,5);
\draw[fill=fillcolor] (8,4) rectangle (10,5);
\draw[fill=fillcolor] (10,4) rectangle (13,5);
\draw[fill=fillcolor] (13,4) rectangle (16,5);
\draw[fill=fillcolor] (16,4) rectangle (26,5);
\draw[fill=fillcolor] (26,4) rectangle (28,5);
\draw[fill=fillcolor] (20,5) rectangle (22,6);
\draw[fill=fillcolor] (22,5) rectangle (23,6);
\draw[fill=fillcolor] (23,5) rectangle (24,6);
\draw[fill=fillcolor] (24,5) rectangle (25,6);
\draw[fill=fillcolor] (25,5) rectangle (32,6);
\draw[fill=fillcolor] (32,5) rectangle (33,6);
\draw[fill=fillcolor] (6,6) rectangle (7,7);
\draw[fill=fillcolor] (7,6) rectangle (11,7);
\draw[fill=fillcolor] (11,6) rectangle (14,7);
\draw[fill=fillcolor] (14,6) rectangle (16,7);
\draw[fill=fillcolor] (16,6) rectangle (17,7);
\draw[fill=fillcolor] (17,6) rectangle (20,7);
\draw[fill=fillcolor] (15,7) rectangle (18,8);
\draw[fill=fillcolor] (18,7) rectangle (21,8);
\draw[fill=fillcolor] (21,7) rectangle (29,8);
\draw[fill=fillcolor] (29,7) rectangle (30,8);
\draw[fill=fillcolor] (30,7) rectangle (31,8);
\draw[fill=fillcolor] (31,7) rectangle (32,8);
\draw[fill=fillcolor] (5,8) rectangle (24,9);
\draw[fill=fillcolor] (24,8) rectangle (25,9);
\draw[fill=fillcolor] (25,8) rectangle (26,9);
\draw[fill=fillcolor] (26,8) rectangle (27,9);
\draw[fill=fillcolor] (27,8) rectangle (35,9);
\draw[fill=fillcolor] (35,8) rectangle (36,9);

\draw[red!50!black, thick] (1.5, 0) -- (1.5, 2.5);
\draw[-latex, red!50!black, dotted, thick] (1.5, 0) -- (1.5, -3);

\draw[red!50!black, thick] (7.5, 0) -- (7.5, 8.5);
\draw[-latex, red!50!black, dotted, thick] (7.5, 0) -- (7.5, -3);

\draw[red!50!black, thick] (10.5, 0) -- (10.5, 6.5);
\draw[-latex, red!50!black, dotted, thick] (10.5, 0) -- (10.5, -3);

\draw[red!50!black, thick] (16.5, 0) -- (16.5, 4.5);
\draw[-latex, red!50!black, dotted, thick] (16.5, 0) -- (16.5, -3);

\draw[red!50!black, thick] (18.5, 0) -- (18.5, 2.5);
\draw[-latex, red!50!black, dotted, thick] (18.5, 0) -- (18.5, -3);

\draw[red!50!black, thick] (22.5, 0) -- (22.5, 6.5);
\draw[-latex, red!50!black, dotted, thick] (22.5, 0) -- (22.5, -3);

\draw[red!50!black, thick] (25.5, 0) -- (25.5, 3.5);
\draw[-latex, red!50!black, dotted, thick] (25.5, 0) -- (25.5, -3);

\draw[red!50!black, thick] (26.5, 0) -- (26.5, 4.5);
\draw[-latex, red!50!black, dotted, thick] (26.5, 0) -- (26.5, -3);

\draw[red!50!black, thick] (31.5, 0) -- (31.5, 8.5);
\draw[-latex, red!50!black, dotted, thick] (31.5, 0) -- (31.5, -3);

\draw[red!50!black, thick] (32.5, 0) -- (32.5, 8.5);
\draw[-latex, red!50!black, dotted, thick] (32.5, 0) -- (32.5, -3);
	\end{tikzpicture}
	\caption{Rectangles and rays in an instance of the rectangle covering problem to which we reduce GSP.}
	\label{fig:RCPinstance}
\end{figure}

Inspired by the PTAS for weighted flow time~\cite{armbruster2023ptas}, we reduce GSP and weighted tardiness
to a geometric covering problem in which we need to select a subset
of some given axis-parallel rectangles in the plane in order to cover
the demand of a given set of rays (see Figure~\ref{fig:RCPinstance}). All given rays
are oriented vertically downwards. Each rectangle $R$ has a cost
and a value that it contributes to satisfy the demands of the rays
that it intersects with (in case $R$ is selected). The rectangles
are pairwise non-overlapping and organized in rows such that from
each row we must select a prefix of its rectangles (or none of them).
Our objective is to satisfy the demand of each ray while minimizing
the total cost of the selected rectangles.
Our reduction loses only a factor of $1+\epsilon$ in the objective function value.

Due to the properties of the weighted flow time objective,
in the algorithm for that case in \cite{armbruster2023ptas} it was possible to slightly round the constructed rectangles so that they
have a tree-like structure. This structure then guided a dynamic program~(DP) for the geometric covering problem.
However, in GSP the job's cost function may not have these properties and
it seems impossible to ensure a similar structure for the rectangles.
Despite this, we construct a quasi-polynomial time algorithm that
computes a $(2+\epsilon)$-approximation for this geometric covering
problem which, due to our reduction, yields a $(2+\epsilon)$-approximation for GSP as well.
Our strategy is to recursively split the plane vertically in two parts, i.e., into a left and a right subproblem. All rays are vertical and, hence, with
respect to them the two subproblems are independent. However, there can be rows whose rectangles intersect with both subproblems (see Figure~\ref{fig:RowTypes}).
\begin{figure}
	\centering
	\begin{tikzpicture}[scale=0.25] 
		\draw[fill=fillcolor] (0,1) rectangle (10,2);
\draw[fill=fillcolor] (10,1) rectangle (12,2);
\draw[fill=fillcolor] (12,1) rectangle (15,2);
\draw[fill=fillcolor] (15,1) rectangle (16,2);
\draw[fill=fillcolor] (16,1) rectangle (26,2);
\draw[fill=fillcolor] (26,1) rectangle (29,2);
\draw[fill=fillcolor] (29,1) rectangle (35,2);
\draw[fill=fillcolor] (1,2) rectangle (38,3);
\draw[fill=fillcolor] (5,3) rectangle (7,4);
\draw[fill=fillcolor] (7,3) rectangle (9,4);
\draw[fill=fillcolor] (9,3) rectangle (10,4);
\draw[fill=fillcolor] (3,4) rectangle (8,5);
\draw[fill=fillcolor] (8,4) rectangle (10,5);
\draw[fill=fillcolor] (10,4) rectangle (13,5);
\draw[fill=fillcolor] (13,4) rectangle (16,5);
\draw[fill=fillcolor] (16,4) rectangle (26,5);
\draw[fill=fillcolor] (26,4) rectangle (28,5);
\draw[fill=fillcolor] (23,5) rectangle (24,6);
\draw[fill=fillcolor] (24,5) rectangle (25,6);
\draw[fill=fillcolor] (25,5) rectangle (32,6);
\draw[fill=fillcolor] (32,5) rectangle (33,6);
\draw[fill=fillcolor] (33,5) rectangle (38,6);
\draw[fill=fillcolor] (6,6) rectangle (7,7);
\draw[fill=fillcolor] (7,6) rectangle (11,7);
\draw[fill=fillcolor] (11,6) rectangle (14,7);
\draw[fill=fillcolor] (14,6) rectangle (16,7);
\draw[fill=fillcolor] (16,6) rectangle (17,7);
\draw[fill=fillcolor] (17,6) rectangle (21,7);
\draw[fill=fillcolor] (21,6) rectangle (22,7);
\draw[fill=fillcolor] (22,6) rectangle (32,7);
\draw[fill=fillcolor] (32,6) rectangle (35,7);
\draw[fill=fillcolor] (29,7) rectangle (30,8);
\draw[fill=fillcolor] (30,7) rectangle (31,8);
\draw[fill=fillcolor] (31,7) rectangle (32,8);
\draw[fill=fillcolor] (12,8) rectangle (24,9);
\draw[fill=fillcolor] (24,8) rectangle (25,9);
\draw[fill=fillcolor] (25,8) rectangle (26,9);
\draw[fill=fillcolor] (26,8) rectangle (27,9);
\draw[fill=fillcolor] (27,8) rectangle (33,9);
\filldraw [pattern=north west lines, pattern color=gray!50!white, path fading=myfading, draw=black] (5,0) rectangle (35,12); 

\draw[color=blue, thick, dotted] (20,-3) rectangle (20,12);

\draw[decorate,decoration={brace,amplitude=7pt}]
(38.5,2.9) -- (38.5,1.1);
\node[right] at (39.5,2) {spanning rows};
\draw[decorate,decoration={brace,amplitude=7pt}]
(38.5,4.9) -- (38.5,3.1);
\node[right] at (39.5,4) {left-sticking-in rows};
\draw[decorate,decoration={brace,amplitude=7pt}]
(38.5,6.9) -- (38.5,5.1);
\node[right] at (39.5,6) {right-sticking-in rows};
\draw[decorate,decoration={brace,amplitude=7pt}]
(38.5,8.9) -- (38.5,7.1);
\node[right] at (39.5,8) {centered rows};

\draw[red!50!black, thick] (7.5, -1) -- (7.5, 8.5);
\draw[-latex, red!50!black, dotted, thick] (7.5, 0) -- (7.5, -3);

\draw[red!50!black, thick] (10.5, -1) -- (10.5, 6.5);
\draw[-latex, red!50!black, dotted, thick] (10.5, 0) -- (10.5, -3);

\draw[red!50!black, thick] (16.5, -1) -- (16.5, 4.5);
\draw[-latex, red!50!black, dotted, thick] (16.5, 0) -- (16.5, -3);

\draw[red!50!black, thick] (18.5, -1) -- (18.5, 2.5);
\draw[-latex, red!50!black, dotted, thick] (18.5, 0) -- (18.5, -3);

\draw[red!50!black, thick] (22.5, -1) -- (22.5, 6.5);
\draw[-latex, red!50!black, dotted, thick] (22.5, 0) -- (22.5, -3);

\draw[red!50!black, thick] (25.5, -1) -- (25.5, 3.5);
\draw[-latex, red!50!black, dotted, thick] (25.5, 0) -- (25.5, -3);

\draw[red!50!black, thick] (26.5, -1) -- (26.5, 4.5);
\draw[-latex, red!50!black, dotted, thick] (26.5, 0) -- (26.5, -3);

\draw[red!50!black, thick] (31.5, -1) -- (31.5, 8.5);
\draw[-latex, red!50!black, dotted, thick] (31.5, 0) -- (31.5, -3);

	\end{tikzpicture}
	\caption{The shaded region indicates a given subproblem. It is split into two parts by the blue dotted line. We lose a factor of $2$ for the centered rows crossing this line in this step, but not for the other rows.}
	\label{fig:RowTypes}
\end{figure}
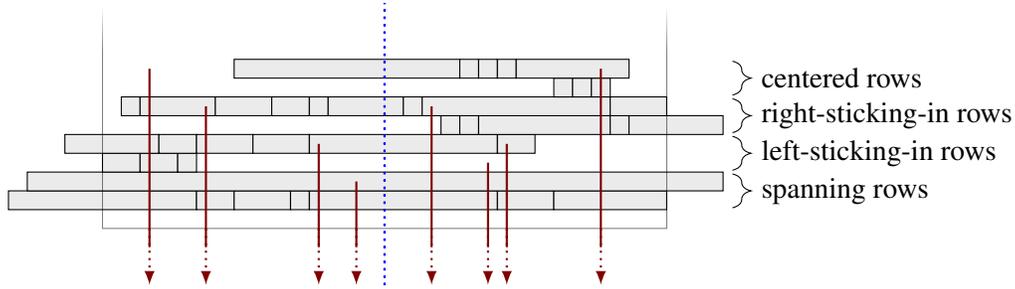
To compensate for this, we can sacrifice a factor of 2 for the cost of such rows, and give independent
copies of its rectangles to each subproblem. More precisely, we can restrict each
copy to rectangles contained in the respective subproblem and cut
one rectangle (which intersects both subproblems) into two pieces if necessary.
One core difficulty is that from each rectangle row we need to select a prefix of its rectangles.
Therefore, even if the right subproblem needs to select only one single rectangle from the row, it needs to pay for
\emph{all} rectangles in that row, including those that are entirely contained in the left subproblem.
Moreover, when we recurse further we cannot sacrifice another factor of $2$ in
each recursion level; doing so would result in a superconstant approximation
ratio. In fact, even if each row contained only one single rectangle it would not be clear
how obtain even a constant approximation ratio.

One key observation is that once we sacrificed a factor of $2$ for the cost of a row,
in each subsequent recursive subproblem
the rectangles of the row start on the left of the subproblem, or end on the right of the subproblem, or even both.
We call such rows \emph{left-sticking-in} rows, \emph{right-sticking-in} rows, and \emph{spanning} rows, see again Figure~\ref{fig:RowTypes}.
We treat these rows differently than the other rows; in this way avoid that
in each recursion level we lose another factor of 2 in our approximation ratio for them.
Intuitively, for each recursive subproblem we consider an optimal solution for this subproblem.
We add certain rectangles to it and in this way, we extend it to a structured near-optimal solution. As a result, for some rays their respective demand is satisfied to a significantly larger extent than necessary, which creates some slack. Also,
the solution's structure allows us to guess a potentially very large number of rectangles directly.

If a row with right-sticking-in rectangles is entirely contained in the
area of the right subproblem, it is clear that its rectangles should be passed on to the right subproblem.
However, a difficult case is when there is a right-sticking-in row from which some rectangles intersect the area
of the left subproblem and some other rectangles intersect the area of the right subproblem, see e.g., the first (upper) right-sticking-in row in Figure~\ref{fig:RowTypes}.
To handle this, we argue that we can pass all rectangles of some of these rows to the left subproblem and all  rectangles of the remaining rows to the right
subproblem. However, if then a row is passed to the \emph{left} subproblem, then we might need its rectangles also to (partially) cover
the demand of rays in the \emph{right} subproblem. To ensure this we create some artificial rays for the left subproblem which are, maybe counterintuitively, placed outside the area corresponding to that subproblem. A crucial argument is that it is sufficient to create a polylogarithmic number of these artificial rays due to the mentioned slack. Hence, we can guess them in quasi-polynomial time. We use a similar argumentation for the left-sticking-in rays. However, for those we have the additional difficulty that whenever we select one of its rectangles, we need to pay for all rectangles on the left of it, and such rectangles might not even
``belong'' to the current subproblem, i.e., they might not even intersect its area. To remedy this, on a high level we argue that we do not lose the mentioned factor of 2 for all rectangles in a row at once in one subproblem. Instead, intuitively we argue that for each rectangle in a row there might be \emph{some} subproblem in which we lose this factor, and for different rectangles these might be different subproblems. In this way, we obtain an approximation ratio of $2+\epsilon$ overall.

Finally, for the special case of the weighted tardiness objective, we argue that
we can adjust our reduction to the geometric covering problem such
that the rectangles' $x$-coordinates are slightly rounded.
This is related to the obtained properties in the corresponding
reduction for weighted flow time~\cite{RohwedderW21, armbruster2023ptas}. However, our obtained properties
are much weaker since weighted tardiness is a more general
cost function. Nevertheless, we show that even these weaker properties
are sufficient to ensure that we do \emph{not} lose a factor of 2 in our approximation ratio as above.
On a high level, our rounding ensures that the rectangles' $x$-coordinates are discretized such
that in each recursive subproblem we need to consider only a polylogarithmic number of different \emph{types} for the rows that do not trivially belong to the left or the right subproblem.
Even more, for each row type we can directly guess up to a factor of
$1+\epsilon$ in quasi-polynomial time which of its rectangles are contained in an optimal
solution for the subproblem. This makes the resulting algorithm much simpler and, at the same time,
improves the approximation ratio to $1+\epsilon$.

\subsection{Other related work}
For GSP, when all release times are equal, Bansal and Pruhs \cite{DBLP:journals/siamcomp/BansalP14} showed that their approach leads to a polynomial time
constant approximation and a QPTAS was presented by Antoniadis et al.~\cite{DBLP:conf/icalp/AntoniadisHMVW17}.
Moseley~\cite{Moseley19} extended the approach to multiple machines (when job migration is allowed) and gave a polynomial time $O(\log\log nP)$-approximation algorithm.
Bansal and Batra \cite{bansal2021non} improved this approximation factor to $O(1)$.
For the objective of weighted sum of completion times (and release times) there is even a PTAS for multiple identical  machines~\cite{afrati1999approximation} and there are extensions to uniformly related and unrelated machines~\cite{afrati1999approximation, ChekuriK01, afrati2000ptas}.
A significant amount of research has focused on online algorithms for weighted flow time. In particular, it is possible to
achieve competitive ratios with only a logarithmic dependence on $P$ and other parameters~\cite{azar2018improved, bansal2009weighted}.
Another well-studied special case of GSP is
the problem of minimizing the total weight of tardy jobs where for each job $j$ we have that
$\cost_j(t) = 0$ if $t\le d_j$ and $\cost_j(t) = w_j$ if $t > d_j$, which generalizes the knapsack problem. Lawler~\cite{lawler1969functional} gave a pseudopolynomial exact algorithm for this setting. When all release times are equal, Lawler and Moore~\cite{lawler1969functional} gave a PTAS. There have been several recent results that focus on optimizing the (pseudopolynomial) running
time for variants of minimizing weight of tardy jobs~\cite{BringmannFHSW20,Klein0R23,FischerW24}.

\section{Algorithmic framework for GSP}\label{sec:algorithmic-framework}

Given an instance of GSP, we reduce it to an instance of the rectangle
covering problem (RCP) which we define in the following. This reduction
loses only a factor of $1+\epsilon$ in the objective. Then, we present
a $(2+\epsilon)$-approximation algorithm for RCP which yields a $(2+\epsilon)$-approximation
algorithm for GSP.

Formally, RCP is defined as follows. The input consists of a set of
axis-parallel rectangles $\R$ and a set of rays $\mathcal{L}$ in
the plane with the following properties:
\begin{enumerate}
\item Each rectangle $R\in\R$ is of the form $R=[a,b)\times[j,j+1)$ for
some values $a,b,j\in\N_{0}$ and we define $\lef(R):=a$ and $\ri(R):=b$; it has a
given cost $c(R)>0$ and a given value $p(R)>0$.
\item The rectangles in $\R$ are pairwise disjoint.
\item \label{prop:costValue}The rectangles in $\R$ are partitioned into a set of \emph{rows}
$\rows$ such that two rectangles $R=[a,b)\times[j,j+1)$, $R'=[a',b')\times[j',j'+1)$
are in the same row if and only if $j=j'$.
\item For each row $W\in\rows$ all rectangles in
$W$ have the same value, i.e., for any $R,R'\in W$ we have that
$p(R)=p(R')$, and
the rectangles in $W$ are consecutive along the $x$-axis, i.e.,
for each rectangle $R\in W$ there is another
rectangle $R'\in W$ with $\ri(R)=\lef(R')$ or there is no rectangle
$R''\in W$ with $\ri(R)<\lef(R'')$.
\item Each ray in $\mathcal{L}$ is of the form $\{t\}\times(-\infty, s]=:L(s,t)$
for some $s,t\in\N_{0}$ and has a given demand $d(L(s,t))\ge0$. For a clearer visualization, we draw such
a ray as $\{t+1/2\}\times(-\infty, s+1/2]$ since then the drawn ray intersects
with the same rectangles as the actual ray but it is not aligned with
the rectangle boundaries.
\end{enumerate}
The goal is to compute a set of rectangles $\R'\subseteq\R$ with
the following properties:
\begin{itemize}
\item For each ray $L(s,t)\in\L$ the rectangles in $\R'$ intersecting
$L(s,t)$ cover the whole demand of $L(s,t)$, i.e., $\sum_{R\in\R':R\cap L(s,t)\neq\emptyset}p(R)\geq d(L(s,t))$.
\item From each row $W\in\rows$ the set $\R'$ contains a consecutive (possibly
empty) set of rectangles starting with the leftmost rectangle, i.e.,
for each rectangle $R\in W\cap\R'$ the set $\R'$ contains also each
rectangle $R'\in W$ with $\ri(R')\leq\lef(R)$.
\end{itemize}
Our objective is to minimize the total cost of the rectangles in $\R'$,
i.e., to minimize $\sum_{R\in\R'}c(R)$. In our reduction, we will
ensure for each row that the cost of all its rectangles together is by at most a bounded
factor larger than the cost of any single one of its rectangles.
Therefore, in a given instance of RCP
we denote by $K$ the maximum ratio of these costs in the same row,
i.e., $K:=\max_{W\in\rows}\frac{\sum_{R\in W}c(R)}{\min_{R\in W}c(R)}$. Note that this implies that
each row can have at most $K$ rectangles.
For a given instance of RCP we denote by $n$ the number of input
bits in binary encoding and, analogously to GSP, we define $p_{\max}:=\max_{R}p(R)$.

We can reduce GSP to RCP while losing only a factor of $1+\varepsilon$
in the approximation ratio. To this end, we will prove the following
lemma in Section~\ref{sec:Reduction}. When measuring the running
time to solve a given instance of GSP, we denote by $n$ the number
of bits needed to encode the processing time $p_{j}$ and the release
time $r_{j}$ of each job $j\in J$. For the jobs' cost functions,
we assume (similarly as in \cite{DBLP:journals/siamcomp/BansalP14})
that we are given access to an oracle that returns in constant time
for each combination of a job $j\in J$ and a value $q>0$ the earliest
time $t\in\mathbb{R}$ such that $\cost_{j}(t)\ge q$.

\begin{lemma}\label{lem:reduction-to-RCP}
Given an $\alpha$-approximation
algorithm for RCP with a running time of $f(n,p_{\max},K)$ there is an $\alpha(1+\varepsilon)$-approximation algorithm for GSP with a running time of $f((n\cdot p_{\max})^{O(1)}, p_{\max}, (1/\epsilon)^{O(1/\epsilon^3)})$.
\end{lemma}

Due to this reduction, it suffices to construct an algorithm for RCP.
We will present a $(2+\epsilon)$-approximation algorithm for RCP
in Section~\ref{sec:RCPalgorithm} which will prove the following lemma.

\begin{lemma}\label{lem:approx-RCP} There is a $(2+\varepsilon)$-approximation
algorithm for RCP with a running time of $2^{(1/\epsilon\cdot K\log n\log p_{\max})^{O(1)}}$. \end{lemma}

Finally, Lemmas \ref{lem:reduction-to-RCP} and \ref{lem:approx-RCP}
yield Theorem~\ref{thm:main}. We will present our $(1+\epsilon)$-approximation algorithm
for weighted tardiness (corresponding to Theorem~\ref{thm:main-tardiness}) in Section~\ref{apx:weighted-tardiness}.

\section{Reduction}\label{sec:Reduction}
In this section we prove \Cref{lem:reduction-to-RCP}. Recall that our goal is to reduce GSP to RCP, losing only a factor of $1 + O(\epsilon)$. This yields a factor of $1+\epsilon$ by standard rescaling. Throughout the section, we allow our running time to be polynomial in $n$, the size of
the GSP instance and $p_{\max}$.

A standard result in scheduling, see e.g.~\cite[Chapter~3]{lenstra2020elements}
is the following condition on when a set of jobs can be scheduled within their
deadlines on a single machine.
\begin{lemma}\label{lem:edf-condition}
Suppose that for each job $j\in J$ we are given a deadline $d_j\geq r_j$. Then there
exists a possibly preemptive schedule on one machine where each job
	is completed by $d_j$ if and only
	if
	\begin{equation}\label{eq:edf-condition}
		\sum_{j : s \le r_j < d_j \le t} p_j \le t - s \quad \forall \forall s,t\in \RR:\min_{j}r_j\leq s<t\leq \max_j d_j
	\end{equation}
	Such a schedule can be found by scheduling the jobs using earliest-deadline-first (EDF).
	If all input numbers are integral
	then EDF produces a schedule that completes and preempts jobs only at integral times.
\end{lemma}
The lemma above implies that to solve GSP, all we need to do is to find job completion times that minimize the cost and
satisfy the condition. Then we can find a schedule using EDF on the completion times.
Let $T = 2^{k+1}$, where $k\in\mathbb N$ with $2^k \le \max_{j\in J} r_j + \sum_{j\in J} p_j < 2^{k+1}$.
Then $T$ is an upper bound on
the last completion time in an optimal schedule produced by EDF.
We may assume without loss of generality that the smallest release time is zero and the difference between any two consecutive release times $r_j < r_{j'}$ is
at most $\sum_j p_j$, since otherwise all jobs released before $r_{j'}$ will in an EDF solution
be finished before $r_{j'}$ and therefore the instance would split into two independent instances.
Thus, we can bound $T$ by $O(n \sum_j p_j) \le O(n^2 p_{\max})$.

Based on the considerations above, the following time-indexed integer linear program (ILP1) exactly models GSP:
\begin{align*}
	\min \sum_{j\in J} &\sum_{t=r_j}^{T-1} (\cost_j(t+1) - \cost_j(t)) \cdot x_{j, t} \hspace{-3cm} \\
	\sum_{j : s \le r_j < t} p_j \cdot x_{j, t} &\ge \sum_{j : s \le r_j < t} p_j - (t - s) &&\forall s, t\in [T]: s < t \\
	x_{j, t-1} &\ge x_{j, t} &&\forall j\in J, t\in \{r_j+2,r_j+3,\dotsc,T\} \\
	x_{j, t} &\in \{0, 1\} &&\forall j\in J, t\in \{r_j+1,r_j+2,\dotsc,T\}
\end{align*}
Here $x_{j, t} = 1$ if and only if $C_j > t$.
Note that the size of ILP1 is not necessarily polynomial in the size of the GSP instance $n$, but
polynomial in $n + p_{\max}$.
\begin{lemma}\label{lem:reduction1}
	Given a solution for GSP with some cost $c$, in time $\mathrm{poly}(n, p_{\max})$ we can compute a solution
	for ILP1 with cost at most $c$.
	Conversely, given a solution of cost $c$ for ILP1
	 in time $\mathrm{poly}(n, p_{\max})$ we can compute a solution for GSP of cost at most $c$.
\end{lemma}
\begin{proof}
	Let $C_j$, $j\in J$, be the completion times in a solution to GSP.
	We derive a solution to ILP1 by setting $x_{j, t} = 1$ for all $t < C_j$ and $x_{j, t} = 0$ for
	all $t \ge C_j$. Then the cost of the solution to ILP1 is
	\begin{equation*}
		\sum_{j\in J} \sum_{t = r_j}^{C_j - 1} \cost_j(t+1) - \cost_j(t) = \sum_{j\in J} \cost_j(C_j)
	\end{equation*}
	Let $s, t\in [T]$ with $s < t$.
	Since $C_j$, $j\in J$, are feasible, we have $\sum_{j : s\le r_j < C_j \le t} p_j \le t - s$.
	Let $S = \{j\in J : s\le r_j < t\}$.
	Then,
	\begin{equation*}
		t - s \ge \sum_{j: s\le r_j < C_j \le t} p_j = \sum_{j\in S} p_j - \sum_{j\in S : C_j > t} p_j
		= \sum_{j\in S} p_j - \sum_{j\in S} p_j \cdot x_{j, t} .
	\end{equation*}
	By rearranging we obtain that the first constraint of ILP1 is satisfied. The second constraint
	is satisfied by definition of $x_{j,t}$.
	
	For the other direction, let $x_{j, t}$ be a solution to ILP1.
	We define $C_j = t$, where $t$ is maximal with $x_{j, t-1} = 1$.
	Because of the second constraint we then have $x_{j, t} = 1$ if and only if $t < C_j$.
	Thus,
	\begin{equation*}
		\sum_{j\in J} \cost_j(C_j) = \sum_{j\in J}\sum_{t = r_j}^{C_j-1} \cost_j(t+1) -  \cost_j(t)
		= \sum_{j\in J}\sum_{t = r_j}^{T - 1} (\cost_j(t+1) -  \cost_j(t)) \cdot x_{j, t} .
	\end{equation*}
	Let $s < t$. As above, define $S = \{j\in J : s\le r_j < t\}$.
	Then
	\begin{equation*}
		\sum_{j : s \le r_j < C_j \le t} p_j = \sum_{j \in S} p_j - \sum_{j\in S : C_j > t} p_j
		= \sum_{j \in S} p_j - \sum_{j\in S} p_j \cdot x_{j, t} \le t - s .
	\end{equation*}
	Thus, EDF applied to the values $C_j$, $j\in J$, produces a solution to GSP with at most the cost of
	the solution for ILP1.
\end{proof}
We will now rewrite ILP1 in two steps. Both steps introduce a small approximation error, but
this ultimately leads to an integer program that is equivalent to RCP.
For each job $j$ we define a sequence of \emph{milestones} $m_0(j),m_1(j),m_2(j),\dotsc \in \mathbb N$,
which are points in time that roughly indicate 
{that the cost function of $j$ increased by a factor of $1+\epsilon$ compared to the previous milestone. 
The properties are formalized in the lemma below.
\begin{lemma}\label{lem:milestones}
	For each job $j$ we can in time $\mathrm{poly}(n, p_{\max})$ construct a sequence $m_0(j),m_1(j),m_2(j),\dotsc,m_{f_j}(j)$
	where
	\begin{enumerate}
		\item $\cost_j(m_{i+1}(j)) \le (1 + \epsilon) \cdot \cost_j(m_{i}(j) + 1)$ for all $0\le i < f_j$,
		\item $\cost_j(m_{i+1}(j) + 1) > (1 + \epsilon / 4) \cdot \cost_j(m_{i}(j) + 1)$ for all $0\le i < f_j - 1$,
		\item $m_0(j) = r_j$,
		\item $m_{f_j}(j) = T$.
	\end{enumerate}
\end{lemma}
\begin{proof}
	The construction is straight-forward:
	For $j\in J$ set $m_0(j) = r_j$ and
	for $i > 0$ let $m_i(j) \le T$ be maximal with
	\begin{equation*}
		\cost_j(m_i(j)) \le (1 + \epsilon) \cdot \cost_j(m_{i-1}(j) + 1) .
	\end{equation*}
	Let $f_j$ be the first index such that $m_{f_j}(j) = T$.
	Note that the construction satisfies (2) even with $\epsilon$ instead of $\epsilon/4$\footnote{We keep this weaker formulation of the lemma in order to later give an alternative approach in the weighted tardiness objective.}.
\end{proof}
The idea of integer program ILP2, which is given below, is that it suffices to restrict ILP1 to the variables $x_{j, m_i(j)}$ (which we denote by $y_{j, i}$ in ILP2) and set all other variables implicitly based on those.
\begin{align*}
	\min \sum_{j\in J} &\sum_{i=0}^{f_j-1} (\cost_j(m_{i+1}(j)) - \cost_j(m_{i}(j))) \cdot y_{j, i} \hspace{-3cm} \\
	\sum_{\substack{j, i : \ s \le r_j < t,\\ m_{i}(j) \le t < m_{i+1}(j)}} p_j \cdot y_{j, i} &\ge \sum_{j : \ s \le r_j < t} p_j - (t - s) &&\forall s, t\in [T]: s < t \\
	y_{j, i-1} &\ge y_{j, i} &&\forall j\in J, i \in \{0,1,\dotsc,f_j\} \\
	y_{j, i} &\in \{0, 1\} &&\forall j\in J, i \in \{0,1,\dotsc,f_j\}
\end{align*}
\begin{lemma}\label{lem:reduction2}
	Let $(m_0(j), m_1(j),\dotsc, m_{f_j}(j))_{j\in J}$ be milestones satisfying the properties of \Cref{lem:milestones}.
	Then given a solution of some cost $c$ for ILP1, one can in time $\mathrm{poly}(n, p_{\max})$ compute
	a solution of cost at most $(1 + \epsilon)c$
	for ILP2. Conversely, given a solution of some cost $c$ for ILP2 one can compute in time $\mathrm{poly}(n, p_{\max})$ 
	a solution of the same cost for ILP1.
\end{lemma}
\begin{proof}
	Let $x_{j,t}$ be a solution to ILP1. Define $y_{j, i} = x_{j, m_i(j)}$ for all $i,j$.
	Denote by $K_j$ the minimal time such that $y_{j, K_j} = 0$.
	Then
	\begin{align*}
		\sum_{j\in J} \sum_{i=0}^{f_j - 1} (\cost_j(m_{i+1}(j)) - \cost_j(m_{i}(j))) \cdot y_{j, i}
		&\le \sum_{j\in J} \cost_j(m_{K_j}(j)) \\
		&\le (1 + \epsilon) \sum_{j\in J} \cost_j(m_{K_j-1}(j)+1) \\
		&= (1 + \epsilon) \sum_{j\in J} \sum_{t=0}^{m_{K_j - 1}(j)} \cost_j(t+1) - \cost_j(t) \\
		&\le (1 + \epsilon) \sum_{j\in J} \sum_{t=0}^{T-1} (\cost_j(t+1) - \cost_j(t)) \cdot x_{j, t} .
	\end{align*}
	Moreover,
	\begin{equation*}
		\sum_{\substack{j, i : \ s \le r_j \le m_0(j) < t,\\ m_{i}(j) \le t < m_{i+1}(j)}} p_j \cdot y_{j, i}
		= \sum_{\substack{j, i : \ s \le r_j < t,\\ m_{i}(j) \le t < m_{i+1}(j)}} p_j \cdot x_{j, m_i(j)}
		\ge \sum_{j : \ s \le r_j < t} p_j \cdot x_{j, t} \ge \sum_{j:\ s \le r_j < t} p_j - (t - s) .
	\end{equation*}
	The prefix constraint holds trivially.
	
	Now consider the other direction. Let $y_{j, i}$ be a solution to ILP2.
	For all $t$ let $x_{j, t} = y_{j, i}$ where $i\in \{0,1,\dotsc,f_j\}$ with $m_i(j) \le t < m_{i+1}(j)$.
	Let $K_j$ be minimal with $y_{j, K_j} = 0$.
	By definition of $x_{j, t}$, we have that $x_{j, t} = 1$ if and only if $t < m_{K_j}(j)$.
	Thus,
	\begin{align*}
		\sum_{j\in J} \sum_{t=0}^{T-1} (\cost_j(t+1) - \cost_j(t)) \cdot x_{j, t}
		&= \sum_{j\in J} \cost_j(m_{K_j}(j)) \\
		&= \sum_{j\in J} \sum_{i=0}^{f_j - 1} (\cost_j(m_{i+1}(j)) - \cost_j(m_{i}(j))) \cdot y_{j, i}
	\end{align*}
	Furthermore,
	\begin{equation*}
		\sum_{j : \ s \le r_j < t} p_j \cdot x_{j, t} =
		\sum_{\substack{j, i : \ s \le r_j < t,\\ m_{i}(j) \le t < m_{i+1}(j)}} p_j \cdot y_{j, i}
		\ge \sum_{j:\ s \le r_j < t} p_j - (t - s) .
	\end{equation*}
	Again, the other constraint holds trivially.
\end{proof}
Note that the prefix constraint in ILP2, i.e., $y_{j,i-1} \ge y_{j,i}$ for all $ j\in J$ and all $i \in \{0,1,\dotsc,f_j\}$, leads to a sequence of variables that for a fixed $j$ show up in the same constraints and thus cannot be chosen independently.
In the next transformation, it is our goal to slice this sequence into small, constant length, sequences.
We define ILP3($(\tau^j_k)_{j\in J, k\in\mathbb N}$) for given increasing sequences $\tau^j_1,\tau^j_2,\tau^j_3,\dotsc \in \mathbb N$ with $\tau^j_1 = 1$.
We later construct sequences carefully to make sure that the cost of the solution does not increase
significantly. If the sequences are clear from the context, we simply write ILP3.
\begin{align*}
	\min \sum_{j\in J} \sum_{k : \tau^j_k \le f_j} &\bigg[ \cost_j(m_{\tau^j_k}(j)) \cdot z_{j, \tau^j_k}
	+ \sum_{i = \tau^j_k}^{\tau^j_{k+1} - 1} (\cost_j(m_{i+1}(j)) - \cost_j(m_{i}(j))) \cdot z_{j, i} \bigg] \hspace{-10cm} \\
	\sum_{\substack{j, i : s \le r_j, \\ m_{i}(j) \le t < m_{i+1}(j)}} p_j \cdot z_{j, i} &\ge \sum_{j : s \le r_j < t} p_j - (t - s) &&\forall s, t\in [T]: s < t \\
	z_{j, i-1} &\ge z_{j, i} &&\forall j\in J, k\in \N, \text{ and } i\in\{\tau^j_k+1,\dotsc,\tau^j_{k+1} - 1\}\cap  \{0, \dots, f_j\} \\
	z_{j, i} &\in \{0, 1\} &&\forall j\in J, i\in \{0, \dots, f_j\}
\end{align*}
\begin{lemma}\label{lem:reduction3}
	For any increasing sequences $\tau^j_k$ with $\tau^j_1 = 1$
	and any solution of some cost $c$ for ILP3($(\tau^j_k)_{j\in J, k\in\mathbb N}$) one can in time $\mathrm{poly}(n, p_{\max})$ 
	compute a solution for ILP2 of cost at most $c$.
\end{lemma}
\begin{proof}
	Consider a solution $z$ for ILP3.
	Define a solution $y$ for ILP2 as follows: for each $j\in J$, let $\ell$ be maximal with
	$z_{j,\ell} = 1$. Then set $y_{j, i} = 1$ if $i\le \ell$ and $y_{j, i} = 0$ otherwise.
	Since the prefix constraint $y_{j,i-1}\ge y_{j,i}$ is satisfied by definition and
	since we have $y_{j,i}\ge z_{j,i}$ for all $i, j$, it is immediate that $y$ is feasible for ILP2.
	Furthermore, the cost of $y$ in ILP2 is at most the cost of $z$ in ILP3:
	we consider separately the cost induced by the variables for each job $j$.
	Let again $\ell$ be maximal with $z_{j,\ell} = 1$ and let $k$ be such that $\ell\in\{\tau^j_k,s^j_k+1,\dotsc,\tau^j_{k+1}-1\}$.
	The cost of variables $(y_{j,i})_i$ is exactly $\cost_j(m_{\ell+1}(j))$ and the cost
	of variables $(z_{j,i})_i$ is at least
	\begin{equation*}
		\cost_j(m_{\tau^j_k}(j)) + \sum_{i=\tau^j_{k}}^{\ell} (\cost_j(m_{i+1}(j)) - \cost_j(m_{i}(j)))
		= \cost_j(m_{\ell+1}(j)) . \qedhere
	\end{equation*}
\end{proof}
Let $S\in \{1,2,\dotsc,(1/\epsilon)^3\}$ be a parameter. We may think
of $S$ as an offset selected uniformly at random.
Define the sequence $(\tau_k(S))_{k\in\N}$ with $\tau_1(S) = 1$ and $\tau_k(S) = S + (k-1)\cdot (1/\epsilon)^3$
for $k \ge 2$. We refer to ILP3($(\tau_k(S))_{k\in\mathbb N}$) as the ILP that uses $(\tau_k(S))_{k\in\mathbb N}$ for each job.
\begin{lemma}\label{lem:reduction4}
	For any solution of some cost $c$ for ILP2 there is some value $S \in\{1, 2, \dotsc, (1/\epsilon)^3 \}$ such that ILP3($(\tau_k(S))_{k\in\mathbb N}$)
	has a solution of cost at most $(1 + 6\epsilon)c$.
\end{lemma}
\begin{proof}
	Consider a solution $y$ for ILP2. We define solution $z = y$ for ILP3
	and analyze the expected cost with $S\in\{1, 2, \dotsc, (1/\epsilon)^3 \}$ chosen uniformly at random. Note that feasibility follows immediately, since ILP3 contains
	a subset of the constraints of ILP2.
	Regarding the cost, we will analyze the expected cost for variables of each job $j\in J$ separately.
	Let $\ell$ be maximal with $z_{j,\ell} = 1$ and let $k$ be such that $\ell\in\{\tau_k,\tau_k+1,\dotsc,\tau_{k+1}-1\}$.
	The cost of $(y_{j,i})_i$ in ILP2 is exactly $\cost_j(m_{\ell+1}(j))$.
	On the other hand, the cost of the variables $(z_{j,i})_i$ in ILP3 is
	\begin{equation}\label{eq:cost-ilp3}
		\sum_{k' = 0}^{k} \cost_j(m_{\tau_{k'}(S)}(j)) + \cost_j(m_{\ell+1}(j)) .
	\end{equation}
	Note that by \Cref{lem:milestones} we have that $\cost_j(m_i(j) + 1) > (1 + \epsilon/4)\cost_j(m_{i-1}(j) + 1)$  for all $i > 0$. Furthermore, $\tau_{k}(S) \ge \tau_{k-1}(S) + (1/\epsilon)^3$ for $k>1$.
	Thus, \eqref{eq:cost-ilp3} is at most
	\begin{align*}
		&\sum_{k' = 0}^{k-1} \cost_j(m_{\tau_{k'}(S)}(j)) + \cost_j(m_{\tau_{k}(S)}(j)) + \cost_j(m_{\ell+1}(j)) \\
		&\le \sum_{k' = 0}^{k-1} \cost_j(m_{\tau_{k'}(S)}(j) + 1) + \cost_j(m_{\tau_{k}(S)}(j)) + \cost_j(m_{\ell+1}(j)) \\
		&\le \cost_j(m_{\tau_{k-1}(S)}(j) + 1) \sum_{h=0}^{\infty} \frac{1}{(1 + \epsilon/4)^{h(1/\epsilon)^3}} + \cost_j(m_{\tau_{k}}(j)) + \cost_j(m_{\ell+1}(j)) \\
		&\le 3\cost_j(m_{\tau_{k}(S)}(j)) + \cost_j(m_{\ell+1}(j)) \\
		&\le 4 \cost_j(m_{\ell+1}(j)) .
	\end{align*}
	Here we assume that $\epsilon$ is a sufficiently small constant.
	Note that the bound above holds for any choice of $S$. However, it also looses a factor of $4$.
	We will show that with probability $1 - \epsilon$ the above inequality holds even with
	a factor of $1 + 3\epsilon$.
	More precisely, we will show that
	\begin{equation*}
		\mathbb P[\cost_j(m_{\tau_{k}(S)}(j)) > \epsilon \cdot \cost_j(m_{\ell+1}(j))] \le \epsilon .
	\end{equation*}
	Observe that if $k = 1$ then $\cost_j(m_{\tau_{k}(S)}(j)) = 0 \le \epsilon \cdot \cost_j(m_{\ell+1}(j))$.
	Furthermore, if $\ell\in \{\tau_k + (1/\epsilon)^2, \tau_k(S) + (1/\epsilon)^2 + 1, \dotsc, \tau_{k+1}(S) - 1 \}$ then
	we have 
	\begin{equation*}
		\cost_j(m_{\ell+1}) > (1 + \epsilon)^{(1/\epsilon)^2} \cost_j(m_{\tau_k(S)}) > 1/\epsilon \cdot \cost_j(m_{\tau_k(S)}).
	\end{equation*}
	Hence, consider the event that $k \in \{\ell - 1, \ell - 2, \dotsc, \ell - (1/\epsilon)^2\}$ and $k > 1$.
	This can only coincide with at most $(1/\epsilon)^2$ of the $(1/\epsilon)^3$ possible values of $O$.
	Thus, by uniformly random choice, the probability is at most $\epsilon$.
	It follows that the cost of the variables $(z_{j, i})_i$ in ILP3 are in expectation
	at most
	\begin{multline*}
		\cost_j(m_{\ell+1}(j)) + 3 \epsilon \cdot \cost_j(m_{\ell+1}(j)) + \mathbb P[\cost_j(m_{\tau_{k}(S)}(j)) > \epsilon \cdot \cost_j(m_{\ell+1}(j))] \cdot 3 \cdot \cost_j(m_{\ell+1}(j)) \\
		\le (1 + 6\epsilon) \cost_j(m_{\ell+1}(j)) . \qedhere
	\end{multline*}
\end{proof}
As a final step in transforming the ILPs, we want to avoid having large jumps in the cost within a group of variables
$\{z_{j, i} : i\in\{\tau_k,\tau_k+1,\dotsc,\tau_{k+1}-1\}\}$ for any job $j\in J$.
Formally, we call $i$ a \emph{large jump} for job $j$, if $\cost_j(m_{i+1}(j)) > 1/\epsilon \cdot \cost_j(m_{i}(j))$.
We derive the sequence $(\tau_k^j(S))_{k\in\mathbb N}$ from $(\tau_k(S))_{k\in\mathbb N}$ by inserting each large jump (which is not yet part of the sequence)
at the
position that maintains the increasing order within the sequence.
\begin{lemma}\label{lem:reduction5}
	For any solution of some cost $c$ for ILP3($(\tau_k(S))_{k\in\mathbb N}$) there there is a solution to ILP3($(\tau^j_k(S))_{j\in J, k\in\mathbb N}$)
	with cost at most $(1 + \epsilon)c$.
\end{lemma}
\begin{proof}
	Any solution for ILP3($(\tau_k(S))_{k\in\mathbb N}$) remains feasible for ILP3($(\tau^j_k(S))_{j\in J, k\in\mathbb N}$).
	Comparing the costs in ILP3($(\tau_k(S))_{k\in\mathbb N}$) and ILP3($(\tau^j_k(S))_{j\in J, k\in\mathbb N}$), we have that
	the coefficient in the objective for variables $z_{j, i}$ where $i$ is a large jump can increase from
	$\cost_j(m_{i+1}(j)) - \cost_j(m_{i}(j))$ to $\cost_j(m_{i+1}(j))$.
	Since it is a large jump, this increase is only by a factor of at most $1 + \epsilon$.
	Hence, also the total cost of the solution increases at most by a factor of $1 + \epsilon$.
\end{proof}

We are now ready to prove the main lemma of the reduction.
\begin{proof}[Proof of \Cref{lem:reduction-to-RCP}]
	Consider an instance of GSP and let $\OPT$ be the value of the optimal solution. From
	\Cref{lem:reduction1} it follows that there is a solution of cost at most $\OPT$ for ILP1.
	Then, by \Cref{lem:reduction1} there is also a solution of cost at most $(1 + \epsilon)\OPT$ for ILP2.
	Because of \Cref{lem:reduction4} and \Cref{lem:reduction5} there is some $S\in\{1,2,\dotsc,(1/\epsilon)^3\}$ such
	that ILP($(\tau^j_k(S))_{j\in J, k\in\mathbb N}$) has a solution of cost at most $(1 + O(\epsilon))\OPT$.
	We guess this value of $S$.
	
	We will now rewrite ILP3 as an instance of the Rectangle Covering Problem (RCP).
	Assume without loss of generality that the jobs are labelled with $J = \{1,2,\dotsc,n\}$ and are
	ordered by release time, that is, $r_1 \le r_2 \le \cdots \le r_n$.
	For each job $j$ and each $k\in\{0,1,\dotsc,f_k\}$ we introduce one row $\sum_{j'=1}^{j-1}(f_j + 1) + k$.
	
	The rectangles of the row corresponding to $j, k$ are defined as follows:
	there is one rectangle $R(j, i) = [m_{i}(j), m_{i+1}(j)) \times [\sum_{j'=1}^{j-1}(f_j + 1) + k, \sum_{j'=1}^{j-1}(f_j + 1) + k+1)$ for every $i\in\{\tau^j_k(S), \tau^j_k(S)+1,\dotsc,\tau^j_{k+1}(S)-1\}$.
	The intuition is that solutions to ILP3 correspond to solutions for RCP, where a rectangle
	$R(j, i)$ is selected if and only if $z_{j, i} = 1$.
	
	The cost of the rectangle is $c(R(j, i)) = \cost_j(m_{i+1}(j))$ if $i = \tau^j_k(S)$ and $c(R(j, i)) = \cost_j(m_{i+1}(j)) - \cost_j(m_{i}(j))$
	otherwise. Note that this is precisely the coefficient of the variable $z_{j,i}$ in the objective of ILP3.
	The rays mimic the covering constraint of ILP3:
	We set $p(R(j, i)) = p_j$.
	Consider some $s, t\in [T]$ with $s < t$.
	Let $j$ be minimal with $s \le r_j$.
	Then we introduce one ray $L(\sum_{j'=1}^{j-1}(f_j + 1), t) = [t] \times [\sum_{j'=1}^{j-1}(f_j + 1), \infty)$ with demand $d(\sum_{j'=1}^{j-1}(f_j + 1), t) = \sum_{j : s\le r_j < t} p_j - (t - s)$, i.e., the right-hand side of the covering constraint of ILP3.
	The rectangles that intersect $L(c, t)$ are exactly those $R(j, i)$ for which $s \le r_j$ and $t \in [m_i(j),m_{i+1}(j))$.
	Note that this corresponds to those $z_{j,i}$ appearing on the left-hand side of the covering constraint.
	Thus, with the transformation between RCP and ILP3 as noted above, the covering constraint of ILP3 is satisfied if and only if the demands of the rays are satisfied.
	Since rows correspond to the set of rectangles $\{R(j, i) : i\in\{\tau_k,\dotsc,\tau_{k+1}-1\}\}$ for some $j$ and $k$ and RCP enforces prefixes of rows to be selected, this again corresponds one-to-one to the prefix constraint of ILP3.
	
	Our definition of RCP requires that all costs are strictly positive, which we can easily establish by
	the following preprocessing. Consider a rectangle of cost zero. If it is the first rectangle of a row, without loss of generality
	any optimal solution contains it, so we remove it from the instance and decrease the demand of
	all rays that intersect it by its value. If the rectangle is not the first in its row, without
	loss of generality any optimal solution, which selects the preceeding rectangle, will also select 
	this rectangle. Hence, we can merge these two rectangles into one with the cost of the first rectangle.

	In any row $W$ we have $|W| \le 2 (1/\epsilon)^3$ rectangles: this is because the difference
	of any two consecutive $\tau_k(S), \tau_{k+1}(S)$ is by definition at most $2 (1/\epsilon)^3$ and
	therefore this also holds for any two consecutive $\tau^j_k(S), \tau^j_{k+1}(S)$.
	Since there are no large jumps
	within a row, each rectangle must have cost at most $(1 / \epsilon)^{|W|} \cdot c(R_1)$,
	where $R_1$ is the first rectangle of $W$. In order to bound the ratio of costs between
	any pair of rectangles of one row, we modify the cost of any rectangle $R\in W$
	to $c'(R) = c(R) + \epsilon / |W| \cdot c(R_1)$.
	Since any solution that selects $R$ must also select $R_1$, this increases (after applying it to
	all rows) the total cost by only a factor of $1 + \epsilon$. The resulting instance
	satisfies that the ratio between the cost of the entire row to the cost of any rectangles of the same row is at most
	\begin{equation*}
		K \le |W| \frac{(1 / \epsilon)^{|W|} \cdot c(R_1) + \epsilon / |W| \cdot c(R_1)}{\epsilon / |W| \cdot c(R_1)} \le (1/\epsilon)^{O(1/\epsilon^3)} .
	\end{equation*}
	
	To summarize, we obtained an instance of RCP, which has a solution of cost at most $(1 + O(\epsilon))\OPT$ (assuming the correct guess of $S$). It is easy to see that the size of this instance is $\mathrm{poly}(n,p_{\max})$.
	Using the given algorithm for RCP, we approximate the instance obtain a solution of cost $\alpha(1 + O(\epsilon))\OPT$. We then apply \Cref{lem:reduction3} to obtain a solution of
	the same cost for ILP2, \Cref{lem:reduction2} to obtain a solution of the same cost for ILP1, and
	finally \Cref{lem:reduction1} to obtain the solution of the same cost for GSP.
	One can scale $\epsilon$ by a constant to achieve the claimed rate of $\alpha(1 + \epsilon)$.
\end{proof}

\section{Algorithm for RCP}\label{sec:RCPalgorithm}

Suppose that we are given an instance of RCP defined by a set of rectangles
$\R$ and a set of rays $\L$. Let $n$ be the total size of the
input and define $p_{\max}:=\max_{R}p(R)$. First, we argue that it suffices
to construct an algorithm in which the number of different rectangle costs appear in the
exponent of the running time, since we can round those costs to only $O(\log n)$ pairwise different values.

\begin{lem}\label{lem:preprocessing}
Assume that for any $M\ge 1$ there is an $\alpha$-approximation algorithm
with a running time of $2^{(M\cdot K\log n\log p_{\max} \log \max_R\ri(R))^{O(1)}}$ for instances
of RCP with at most $M$ different rectangle costs. Then there
is an $(1+\epsilon)\alpha$-approximation algorithm for arbitrary
instances of RCP with a running time of $2^{(1/\epsilon\cdot K\log n\log p_{\max})^{O(1)}}$.
\end{lem}
\begin{proof}
	First we argue that we may assume $T=O(n)$.
	Suppose there exists an interval $I=[a,b)$ with $a,b\in \N_{0}$ such that 
	\begin{itemize}
		\item there is no ray $L(s,t)\in \L$ with $t\in I$,
		\item there is no rectangle $R\in \R$ with $\lef(R)\in I$ or $\ri(R)\in I$ and
		\item there exists a rectangle $\R\in R$ with $\ri(R)\geq b$.
	\end{itemize}
	Intuitively, we just delete the vertical strip $I\times \RR$. Formally, do the following: We decrease each of the following quantities simultaneously by $b-a$:
	\begin{itemize}
		\item  $t$ for every ray $L(s,t)\in \L$ where $t\geq b$,
		\item $\lef(R)$ for every rectangle $R\in \R$ with $\lef(R)\geq b$, and
		\item $\ri(R)$ for every rectangle $R\in \R$ with $\lef(R)\geq b$.
	\end{itemize} 
	This operation does not change whether a set of rectangles is a prefix of a row nor does it change whether a certain rectangle and a certain ray intersect, so the feasible and optimal solutions do not change.
	We do this operation repeatedly for an interval $I$ of maximal length. 
	This can only happen $|\R|+|\L|$ times. 
	And after that, we have that for any $t$ with $0\leq t \leq \max_{R\in \R}\ri(R)$, there exists a ray $L(s,t)\in \L$, or there is rectangle $R\in \R$ with $\lef(R)=t$ or there exists a rectangle $\R\in R$ with $\ri(R)=t$. So $\max_{R\in \R}\ri(R)=O(|\R|+|\L|)=O(n)$ and thus also $T=O(n)$.
	
	Now we guess the rectangle $R_{\max}$ with maximal cost in the optimal solution $\OPT$.
	We discard all rectangles $R'$ with cost larger than $c(R_{\max})$ and all rectangles that are in the same row as such an $R'$ on the right of it.
	Let $\R_{\text{disc}}$ denote these rectangles.
	Note that the optimal solution cannot select any rectangle from $\R_{\text{disc}}$, as it has to select a prefix of each row.
	Next, let $\APX_\text{cheap}$ denote all rectangles in rows in which no rectangle has a cost of more that $\varepsilon c(R_{\max})/|\R|$.
	We select all rectangles in  $\APX_\text{cheap}$.
	Note that $c(\APX_{\text{cheap}})\leq |\APX_{\text{cheap}}|\cdot \varepsilon c(R_{\max})/|\R| \leq \varepsilon c(\OPT)$.
	
	Let $\R'=\R \setminus (\R_{\text{disc}} \cup \APX_\text{cheap})$.
	For $R\in \R'$, let $\tilde{c}(R)$ be the smallest power of $1+\varepsilon$ larger than $c(R)$ and let $c'(R):=\lceil \frac{|\R|\cdot K}{\varepsilon^2 c(R_{\max})}\tilde{c}(R)\rceil$.
	We call the given algorithm with the instance $I'$ consisting of the rectangles $\R'$ with cost $c'$ and the original demand rays $\mathcal{L}$, but a ray $L(s,t)\in \mathcal{L}$ has demand $d'(s,t)=\max\{0, d(s,t)-\sum_{R\in \APX_{\text{cheap}}:R\cap L(s,t)\neq \emptyset}p(R)\}$.
	Let $\APX$ be the obtained solution. Then we output $\APX \cup \APX_{\text{cheap}}$.
	
	Clearly, $\APX \cup \APX_{\text{cheap}}$ is feasible.
	Next, we analyze the cost.
	The cost of two rectangles in the same row can differ at most by a factor of $K$.
	So for $R\in R'$, we have $c(R)\geq \frac{\varepsilon c(R_{\max})}{K\cdot |\R|}$. Thus $c(R)\leq \frac{\varepsilon^2 c(R_{\max})}{ |\R|\cdot K}c'(R)$ and $c'(R)\leq \frac{ |\R|\cdot K}{\varepsilon^2 c(R_{\max})} \tilde{c}(R)+1\leq  \frac{ |\R|\cdot K}{\varepsilon^2 c(R_{\max})} (\tilde{c}(R)+\varepsilon c(R))\leq  \frac{ |\R|\cdot K}{\varepsilon^2 c(R_{\max})} (1+2\varepsilon)c(R)$.
	As $\OPT\setminus \APX_\text{cheap}$ is a feasible solution for $I'$, we have $c'(\APX)\leq \alpha c'(\OPT\setminus \APX_\text{cheap})\leq \alpha \cdot  c'(\OPT)$. Thus
	$c(\APX \cup \APX_{\text{cheap}})\leq c(\APX)+ c(\APX_{\text{cheap}})\leq (1+2\varepsilon)\alpha \cdot c(\OPT)\leq (1+3\varepsilon)\alpha \cdot c(\OPT)$.
	Rescaling $\varepsilon$ by a factor of $3$ yields the desired $(1+\varepsilon)\alpha$-approximation algorithm.
	
	Finally, we analyze the running time. There are at most $|\R|$ options for $R_{\max}$. In $I'$, let $R, R'$ be two rectangles in the same row with $c'(R)\leq c'(R')$.
	Then as they are in the same row, we have $c(R)\geq 1/K\cdot c(R')$ and therefore $c'(R)\geq 1/((1+2\varepsilon)K)\cdot c'(R')\geq 1/(2K)\cdot c'(R')$.
	Thus the cost of two rectangles in a row differs by at most $2K$ in the instance $I'$.
	As for each $R\in \R'$ we have $c(R_{\max})\geq c(R)\geq \frac{\varepsilon c(R_{\max})}{K\cdot |\R|}$, there are at most $O(\log(K\cdot |\R|/\varepsilon)/\varepsilon)$ different values for $\tilde{c}(R)$ and thus also at most the same number of values for $c'(R)$.
	This yields the desired bound on the running time.
\end{proof}
Assume that we are given an instance of RCP with at most $M$ different rectangle costs for some value $M\ge 1$ and assume
w.l.o.g. that $\min_{R\in\R}\lef(R)=0$. Our algorithm is based
on a recursion in which the input of each recursive call consists
of
\begin{itemize}
\item an instance of RCP defined by a set of rays $\L'$ and a set of rectangles
$\R'$; let $\rows'$ be the partition of $\R'$ into rows,
\item an area $A$ of the form $A=[\lef(A),\ri(A))\times[0,\infty)$ for
two values $\lef(A),\ri(A)\in\N_{0}$ such that $\ri(A)-\lef(A)$
is a power of 2,
\item for each row $W\in\rows'$ there is a rectangle $R\in W$ with $W\cap A\neq\emptyset$,
\item there are at most $(K\cdot M\cdot \log ( p_{\max}\cdot \max_{R}\ri(R))/\varepsilon)^{O(1)}$ rays $L\in\L'$ outside $A$,
	i.e., such that $L\cap A=\emptyset$.
\end{itemize}
Note that there may be rectangles $R\in\R'$ with $R\not\subseteq A$.
The recursive call returns a solution to the RCP instance defined
by $\R'$ and $\L'$. At the end of our algorithm, we output the solution
returned by the (main) recursive call in which $\R':=\R$, $\L':=\L$,
$\lef(A)=0$, and $\ri(A):=T$ where $T$ is the smallest power of 2 that is at
least $\max_{R\in\R}\ri(R)$. Any solution to this subproblem is a
solution to our given instance.

Assume we are given a recursive call as defined above. The base case
of our recursion arises when $\ri(A)-\lef(A)=1$. We can solve such
instances by a simple dynamic program in quasi-polynomial time, using
that there are only $(K\cdot M\cdot \log ( p_{\max}\cdot T)/\varepsilon)^{O(1)}$
rays $L\in\L'$ outside
$A$, i.e., in $\mathcal{L}'_{\mathrm{out}}:=\{L\in\L':L\cap A=\emptyset\}$.
\begin{lem}
\label{lem:BaseCase} If $\ri(A)-\lef(A)=1$, we can solve the subproblem
defined by the recursive call exactly in time $(n\cdot p_{\max})^{O(|\mathcal{L}'_{\mathrm{out}}|)}$.
\end{lem}
\begin{proof}
	We solve this case via a dynamic program.  
	There is a cell for every combination of a row $W\in \rows$, a covered demand $p(s,t)\in \N_0$ with $p(s,t)\leq d(s,t)$ for each $L(s,t)\in \L'_{\operatorname{out}}$ and a covered demand $p(A)\in \N_0$ with $p(A)\leq \max\{d(s,t):L(s,t)\in \mathcal{L}'_{\operatorname{out}}\}$. 
	The subproblem corresponding to a cell $(W, (p(s,t))_{L(s,t)\in \mathcal{L}'_{\operatorname{out}}}, p(A))$ is to compute a subset $S$ of rectangles in rows $W$ and below of minimal cost with the following properties or decide that such a set does not exist:
	\begin{itemize}
		\item The set $S$ contains a prefix from each row.
		\item The total demand covered on each ray $L(s,t)\in \mathcal{L}'_{\operatorname{out}}$ at least $p(s,t)$, i.e. $p(\{R\in S:R\cap L(s,t)\neq \emptyset\})\geq p(s,t)$.
		\item the total value of rectangles intersecting with $A$ is exactly $p(A)$, i.e. $p(\{R\in S:R\cap A\neq \emptyset\})= p(A)$
		\item all rays $L(s,t)\in \mathcal{L}'$ that begin in or below $W$ are covered by $S$.
	\end{itemize}
	Intuitively, we use $p(A)$ to remember the amount covered on any ray that starts above $W$ and intersects $A$. As they all intersect with the same rectangles from rows $W$ and below, it suffices to remember this number.
	The dynamic program can be solved as follows: First of all, the subproblem is infeasible if there exists a demand ray $L(s,t)\in \mathcal{L}'\setminus\L'_{\operatorname{out}}$ with $s=\lef(A)$, that starts in or below row $W$ and fulfills $d(s,t)> p(A)$.
	For a subproblem in the bottom row $W$, we can iterate through all possible prefixes of $W$, check the demand constraints for each prefix and select the minimal cost prefix if there are multiple ones or declare the subproblem infeasible if no feasible prefix exists.
	This yields the optimal solution for such a DP-cell.
	
	So suppose that $W$ is not the bottom row. Then we iterate through all possible prefixes $S_W\subseteq W$ of rectangles in row $W$. Let $W'$ be the row directly below $W$, i.e., such that there is no row between $W$ and $W'$.
	For each prefix, the resulting subproblem corresponds the DP-cell $(W', (p'(s,t))_{L(s,t)\in \mathcal{L}'_{\operatorname{out}}}, p'(A))$ where  $p'(s,t)=\max\{0,p(s,t)-p(R)\}$ if there exists a rectangle $R\in S_W$ intersecting with $L(s,t)$ and $p'(s,t)=p(s,t)$ otherwise,
	and $p'(A)=\max\{0,p(A)-p(R)\}$ if there exists a rectangle $R\in S_W$ intersecting with $A$ and $p'(A)=p(A)$ otherwise.
	Let $S'$ be the solution for the DP-cell $(W', (p'(s,t))_{L(s,t)\in \mathcal{L}'_{\operatorname{out}}}, p'(A))$.
	For each prefix $S_W\subseteq W$, we check whether $S_W\cup S'$ is feasible for the current subproblem and select the minimal cost solution among all feasible ones.
	
	Now we show that the algorithm solves the problem optimally. Feasibility follows directly from construction. To show optimality we use induction. We already argued that the base case is solved optimally. Let $S$ be the optimal solution for a cell which is not a base case and let $S_W=S\cap W$. Then $S\setminus S_W$ is also a feasible solution for the subproblem. Let $S'$ be the solution for the subproblem given by the dynamic program. Then $c(S')\leq c(S \setminus S_W)$. So $c(S_W\cup S')\leq c(S)$ and it remains to show that $S_W\cup S'$ is feasible. It certainly contains only prefixes, covers a demand $p(s,t)$ on each ray $L(s,t)\in \L'_{\operatorname{out}}$, covers all rays that start in or below $W'$ and the value of rectangles intersecting $A$ is also $p(A)$. And for a ray $L(s,t)\in \L'\setminus \L_{\operatorname{out}}'$ that start in or below $W$, but not in or below $W'$, the following happens. The rectangles intersecting $L(s,t)$ (except possible the ones in $W$) are exactly the rectangles in rows in or below $W'$ that intersecting with $A$. Hence, $L(s,t)$ is covered to the same extend by $S'$ and $S\setminus S_W$ and thus also to the same extend by $S'\cup S_W$ and $S$. Thus, $L(s,t)$ is also covered. This shows that the algorithm computes the optimal solution to every subproblem
	
	The output is the minimum cost solution among all values of $p(A)$ of the cell $(W, (p(s,t))_{L(s,t)\in \mathcal{L}'}, p(A))$ where $W$ is the highest row and $p(s,t)=d(s,t)$ for each $L(s,t)\in \L'_{\operatorname{out}}$, which is the optimal solution to the given RCP instance.
	The running time can be bounded as follows. 
	There are at most $n$ values for $W$ and $n\cdot p_{\max}$ options for each $p$, this yields a total number of DP-states of $(n\cdot p_{\max})^{O(|\L'_{\operatorname{out}}|)}$. One DP-state for a row $W$ can be solved in time $O(|W|)=O(n)$, yielding the desired running time.
\end{proof}

Assume now that $\ri(A)-\lef(A)>1$. We classify the rows $\rows'$
into four different types which we will handle differently in the
following, see also Figure~\ref{fig:RowTypes}.
\begin{defn}
A row $W\in\W'$ is
\begin{itemize}
\item \emph{centered} if it is contained in the interior of $A$, i.e., for
each rectangle $R\in W$ we have that $\lef(A)<\lef(R)$ and $\ri(R)<\ri(A)$,
\item \emph{right-sticking-in} if each rectangle $R\in W$ satisfies $\lef(A)<\lef(R)$
and there exists a rectangle $R'\in W$ with $\ri(A)\le\ri(R')$,
\item \emph{left-sticking-in} if each rectangle $R\in W$ satisfies $\ri(R)<\ri(A)$
and there exists a rectangle $R'\in W$ with $\lef(R')\le\lef(A)$,
\item \emph{spanning }if there exists a rectangle $R\in W$ with $\lef(R)\le\lef(A)$
and a rectangle $R'\in W$ with $\ri(A)\le\ri(R')$.
\end{itemize}
\end{defn}

We denote by $\W_{\ce},\W_{\ri},\W_{\lef},$ and $\W_{\spn}$ the
set of centered, right-sticking-in, left-sticking-in, and spanning
rows, and by $\R_{\ce},\R_{\ri},\R_{\lef},$ and $\R_{\spn}$ the
union of their rectangles, respectively. Intuitively, we will compute
a solution to the given subproblem in which we lose a factor of $2+\epsilon$
on the cost of the rectangles in centered rows (compared to their
cost in the optimal solution) and a factor of $1+\epsilon$ on the
cost of rectangles in right-sticking-in and spanning rows.
This has some similarities to the treatment of tasks according to their time windows in \cite{armbruster2025approximability}.
For the rectangles in left-sticking-in rows, this is more complicated and we will lose a factor of $1+\epsilon$
on the cost of rectangles $R$ with $\lef(R)\le\lef(A)$ and a factor
of $2+\epsilon$ on the cost of rectangles $R'$ with $\lef(R')>\lef(A)$.

We define a reference solution $S\subseteq\R'$ which is a solution
to our given subproblem that optimizes an auxiliary cost function
$c_{\APX}$, motivated by the factors we lose for the different rectangle types, defined as
\begin{align*}
	c_{\APX}(S'):=\:&2\cdot c(S'\cap\R_{\ce})+c(S'\cap\R_{\ri})+c(S'\cap\R_{\spn})+c(R\in S'\cap\R_{\lef}:\lef(R)\le\lef(A))\\
	&+2\cdot c(R\in S'\cap\R_{\lef}:\lef(R)>\lef(A))
\end{align*}
 for each $S'\subseteq\R'$. Our algorithm will involve guessing certain
properties. Intuitively, we will argue later that we compute a solution
with small overall cost if we guess these properties corresponding
to $S$.

\subsection{Definition of structured near-optimal solution}

Based on $S$ we define a more structured near-optimal solution $S^{+}$
with $S\subseteq S^{+}$ and $c(S^{+})\le(1+O(\epsilon/\log T))c(S)$.
First, we consider all left-sticking-in rows. Let $\mi(A):=(\ri(A)-\lef(A))/2$.
We consider each pair $(c',p')\in\N_{0}^{2}$ such that there is a
row $W\in\rows_{\lef}$ for which
\begin{itemize}
\item its leftmost rectangle $R\in W$ has cost $c(R)=c'$,
\item for each rectangle $R'\in W$ its value $p(R')$ satisfies $p(R')\in[(1+\epsilon)^{p'},(1+\epsilon)^{p'+1})$,
and
\item there is a rectangle $R''\in W$ for which $\lef(R'')<\mi(A)\le\ri(R'')$.
\end{itemize}
For each such pair $(c',p')$ let $\W_{\lef}^{c',p'}$ denote the
set of all rows with the above properties and let $n_{\lef}^{c',p'}$
denote the number of rows in $\W_{\lef}^{c',p'}$ from which $S$
selects \emph{at least one} rectangle intersecting $A$. Also, let $\R_{\lef}^{c',p'}:=\bigcup_{W \in \W_{\lef}^{c',p'}}W$ denote the set of all rectangles in these rows.
Note that
the total cost of these rectangles in $S$ is at least $c'\cdot n_{\lef}^{c',p'}$ . In $S^{+}$ we want to select
additional rectangles from certain rows in $\W_{\lef}^{c',p'}$. Due to this,
we will oversatisfy the demands of some rays in $\L'$. This will
create some slack which we will exploit later. If $n_{\lef}^{c',p'}<\frac{2K\cdot \log T}{\epsilon}$ there are only few rows from which $S$ selects a rectangle and we will simply guess these rows and rectangles later. Suppose that $n_{\lef}^{c',p'}\geq \frac{2K\cdot \log T}{\epsilon}$.  Let $\W_{\lef,\add}^{c',p'}$
denote the bottom-most $\lfloor \frac{\epsilon}{2K\cdot \log T}\cdot n_{\lef}^{c',p'} \rfloor$
rows in $\W_{\lef}^{c',p'}$ from which $S$ does not contain any
rectangle intersecting $A$;
in case there are less than $\lfloor \frac{\epsilon}{2K\cdot \log T}\cdot n_{\lef}^{c',p'} \rfloor$ such rows then we define that
$\W_{\lef,\add}^{c',p'}$ is the set of all these rows.
We will add to our solution $S^{+}$
\emph{all} rectangles in \emph{all} rows in $\W_{\lef,\add}^{c',p'}$; we denote them by $\R_{\lef,\add}^{c',p'}$. Let
$W_{\lef,\filled}^{c',p'}$ denote the top-most row such that for each row $W \in \W_{\lef}^{c',p'}$ below it, we have that $S \cup \R_{\lef,\add}^{c',p'}$ contains a rectangle in row $W$ intersecting $A$.
In particular, 
for each row in $\W_{\lef}^{c',p'}$ underneath
$W_{\lef,\filled}^{c',p'}$ the solution $S^{+}$ selects the rectangle
intersecting the line $\{\lef(A)\}\times\RR$. This will make it easy
to guess those rectangles since it suffices to guess $W_{\lef,\filled}^{c',p'}$.

In a similar fashion we select additional rectangles for $S^{+}$
from right-sticking-in and spanning rows. The procedure is identical
for both cases, we describe it only for right-sticking-in rows. We
define $\W_{\ri}^{c',p'}$, $n_{\ri}^{c',p'}$, $\R_{\ri}^{c',p'}$, and the set of considered
pairs $(c',p')\in\N_{0}^{2}$ in exactly the same way as $\W_{\lef}^{c',p'}$, $n_{\lef}^{c',p'}$ and  $\R_{\lef}^{c',p'}$ above. We define $\W_{\ri,\add}^{c',p'}$
as the bottom-most $\lfloor \frac{\epsilon}{2K \cdot \log T}\cdot n_{\ri}^{c',p'}\rfloor$
rows in $\W_{\ri}^{c',p'}$ from which $S$ does not contain any rectangle
$R$ with $\lef(R)<\mi(A)\le\ri(R)$,
and we define $\R_{\ri,\add}^{c',p'}$ to be their rectangles.

Overall, we define $S^{+}:=S\cup\bigcup_{c',p'}\R_{\lef,\add}^{c',p'}\cup\R_{\ri,\add}^{c',p'}\cup\R_{\spn,\add}^{c',p'}$.
\begin{lem}\label{lem:propertiesS+}
The set $S^{+}$ has the following properties:
\begin{itemize}
\item $c(S^{+})\le(1+\epsilon/\log T)c(S)$
\item for each pair $(c',p')$ with $n_{\lef}^{c',p'}\geq \frac{2K\cdot \log T}{\epsilon}$ the set $S^{+}$
contains the rectangle intersecting the line $\{\lef(A)\}\times\RR$
from the row $W_{\lef,\filled}^{c',p'}$ and from each row in $\W_{\lef}^{c',p'}$
underneath $W_{\lef,\filled}^{c',p'}$,
\item for each pair $(c',p')$ with $n_{\ri}^{c',p'}\geq \frac{2K\cdot \log T}{\epsilon}$ the set $S^{+}$
contains the rectangle intersecting the line $\{\mi(A)\}\times\RR$
from the row $W_{\ri,\filled}^{c',p'}$ and from each row in $\W_{\ri}^{c',p'}$
underneath $W_{\ri,\filled}^{c',p'}$,
\item for each pair $(c',p')$ with $n_{\spn}^{c',p'}\geq \frac{2K\cdot \log T}{\epsilon}$ the set $S^{+}$
contains the rectangle intersecting the line $\{\mi(A)\}\times\RR$
from the row $W_{\spn,\filled}^{c',p'}$ and from each row in $\W_{\spn}^{c',p'}$
underneath $W_{\spn,\filled}^{c',p'}$,
\item for each $\mathrm{set}\in\{\lef,\ri,\spn\}$ and for each pair $(c',p')\in\N_{0}^{2}$
for which $n_{\mathrm{set}}^{c',p'}\geq \frac{2K\cdot \log T}{\epsilon}$ the following holds:
if a ray $L(s,t)\in\L'$
intersects with a rectangle from a row $W \in \rows_{\mathrm{set}}^{c',p'}$ above $W_{\mathrm{set},\filled}^{c',p'}$ and additionally
\begin{itemize}
\item $\lef(A)\leq t<\mi(A)$ if $\mathrm{set}=\lef$ and
\item $\mi(A)\leq t<\ri(A)$ if $\mathrm{set}\in\{\ri,\spn\}$
\end{itemize}
we have that
$$
	p(\{R\in S^{+}\cap\R_{\mathrm{set}}^{c',p'}:R\cap L(s,t)\neq \emptyset\}) \ge p(\{R\in S\cap\R_{\mathrm{set}}^{c',p'}: R\cap  L(s,t)\neq \emptyset\})
	+(1+\varepsilon)^{p'}\lfloor \frac{\epsilon}{2K \cdot \log T}\cdot n_{\operatorname{set}}^{c',p'}\rfloor.
$$

\end{itemize}
\end{lem}
\begin{proof}
	Consider a group $(c',p')$ and let $\operatorname{set}\in \{\lsub, \rsub, \spn\}$. Then $c(S\cap \R^{c',p'}_{\operatorname{set}})\geq c'\cdot n^{c',p'}_{\operatorname{set}}$ as $S$ contains rectangles in $n^{c',p'}_{\operatorname{set}}$ different rows of $\rows^{c',p'}_{\operatorname{set}}$ and the first rectangle in each such row has a cost of $c'$.
	Recall that $\R^{c',p'}_{\operatorname{set}, \add}$ contains all rectangles from up to $\frac{\varepsilon}{2\cdot K\cdot \log T}n^{c',p'}_{\operatorname{set}}$ rows.
	Furthermore, the total cost of all rectangles in a row is at most a factor $K$ larger than the cost of a single rectangle and the first rectangle in each row has a cost of $c'$. Thus,  we have that \begin{align*}
		c(\R^{c',p'}_{\operatorname{set}, \add})&\leq K\cdot c' \cdot  \frac{\varepsilon}{2\cdot K\cdot \log T}n^{c',p'}_{\operatorname{set}}
		\leq \frac{\varepsilon}{2\cdot \log T}\cdot c(S\cap \R^{c',p'}_{\operatorname{set}}W).
	\end{align*}
	Using this for all $(c', p')$ and each $\operatorname{set}\in \{\lsub, \spn, \rsub\}$, we get  $c(S^+)\leq (1+\frac{\varepsilon}{2\log T})c(S)$.
	
	From the definition of $W^{c', p'}_{\operatorname{set}, \filled}$ we obtain directly that $S^+$ contains the rectangles from row $W^{c', p'}_{\operatorname{set}, \filled}$ and every row in $\rows^{c', p'}_{\operatorname{set}}$ below $W^{c', p'}_{\operatorname{set}, \filled}$ that intersect with
	\begin{itemize}
		\item the line ${\lef(A)}\times \RR$ if $\operatorname{set}=\lsub$ and
		\item  the line ${\mi(A)}\times \RR$ if $\operatorname{set}\in \{\spn, \rsub\}$.
	\end{itemize}
	
	So it remains to show that $S^+$ covers substantially more than $S$. Consider a group $(c',p')$ and suppose that $n^{c', p'}_{\lsub}\geq \frac{2K\cdot \log T}{\varepsilon}$. Let $L(s,t)\in \L'$ be a ray that intersects a rectangle from a row $W\in \rows^{c', p'}_\lsub$ above $W^{c', p'}_{\lsub, \filled}$ and fulfills $\lef(A)\leq t<\mi(A)$. Recall that $\rows^{c', p'}_{\lsub, \add}$ denotes the bottom-most $\lfloor\frac{\varepsilon}{2K\log T}n^{c', p'}_\lsub\rfloor$ rows in $\rows^{c', p'}_\lsub$ from which $S$ does not contain any rectangle intersecting $A$. 
	Each row $W\in \rows^{c', p'}_{\lsub, \add}$ is spanning $A_\lsub$, i.e. there exists a rectangle $R'\in W$ with $\lef(R')\leq \lef(A)$ and a rectangle $R''\in W$ with $\mi(A)\leq \ri(R'')$.
	So there also exists a rectangle $R$ with $\lef(R)\leq t <\ri(R)$.
	As $W$ is below or equal to $W^{c', p'}_{\lsub, \filled}$ and $L(s,t)$ starts above  $W^{c', p'}_{\lsub, \filled}$, the ray $L(s,t)$ intersects with $R$.
	So there are $\lfloor\frac{\varepsilon}{2K\log T}n^{c', p'}_\lsub\rfloor$ rectangles in $\R^{c', p'}_{\lsub, \add}$ that intersect with $L(s,t)$.
	So $p(\{R\in \R^{c', p'}_{\lsub, \add}:R\cap L(s,t)\neq \emptyset\})\geq (1+\varepsilon)^{p'} \lfloor\frac{\varepsilon}{2K\log T}n^{c', p'}_\lsub\rfloor$.
	This yields the last property for $\operatorname{set}=\lsub$. 
	
	For $\operatorname{set}\in \{\rsub, \spn\}$ the proof is very similar. Consider a group $(c',p')$ and suppose that $n^{c', p'}_{\operatorname{set}}\geq \frac{2K\cdot \log T}{\varepsilon}$ and $L(s,t)\in \L'$ is a ray that intersects a rectangle from a row $W\in \rows^{c', p'}_{\operatorname{set}}$ above $W^{c', p'}_{\operatorname{set}, \filled}$ and fulfills $\mi(A)\leq t<\ri(A)$. 
	Recall that $\rows^{c', p'}_{\operatorname{set}, \add}$ denotes the bottom-most $\lfloor\frac{\varepsilon}{2K\log T}n^{c', p'}_{\operatorname{set}}\rfloor$ rows in $\rows^{c', p'}_{\operatorname{set}}$ from which $S$ does not contain any rectangle intersecting $A_\rsub$. 
	In each row $W\in \rows^{c', p'}_{\operatorname{set}, \add}$ there exists a rectangle $R$ with $\lef(R)\leq t <\ri(R)$.
	As $W$ is below or equal to $W^{c', p'}_{\operatorname{set}, \filled}$ and $L(s,t)$ starts above $W^{c', p'}_{\operatorname{set}, \filled}$, the ray $L(s,t)$ intersects with $R$.
	So there are $\lfloor\frac{\varepsilon}{2K\log T}n^{c', p'}_{\operatorname{set}}\rfloor$ rectangles in $\R^{c', p'}_{\operatorname{set}, \add}$ that intersect with $L(s,t)$, which as before, yields the last property for $\operatorname{set}\in \{\rsub, \spn\}$ and completes the proof.
\end{proof}

\subsection{Algorithm}\label{subsec:GSP-algorithm}

Algorithmically, we guess certain properties of $S^{+}$, select some
rectangles from $\R'$ accordingly, and then partition the remaining
problem into two subproblems which we then solve recursively. First,
for each pair $(c',p')$ we guess whether 
$n_{\lef}^{c',p'}<\frac{2K\cdot \log T}{\epsilon}$. If this is the case,
we guess $\W_{\lef}^{c',p'}$ and all rectangles in $\R_{\lef}^{c',p'}$
that are contained in $S^+$ and select those. Since $n_{\lef}^{c',p'}<\frac{2K\cdot \log T}{\epsilon}$ we
can do this in time $n^{O(K \log T/\varepsilon)}$.
Assume now that $n_{\lef}^{c',p'}\ge\frac{2K\cdot \log T}{\epsilon}$.
We guess $W_{\lef,\filled}^{c,p}$ and for
each row $W\in\W_{\lef}^{c,p}$ underneath $W_{\lef,\filled}^{c,p}$
we select the rectangle $R\in W$ intersecting $\{\lef(A)\}\times\RR$
and all rectangles on the left of $R$. Similarly, for each pair $(c',p')$
we guess whether
$n_{\ri}^{c',p'}<\frac{2K\cdot \log T}{\epsilon}$
and if yes, we guess $\W_{\ri}^{c',p'}$ and all rectangles in $\R_{\ri}^{c',p'}$
that are contained in $S^+$ and select them.
Otherwise, we guess $W_{\ri,\filled}^{c,p}$
and for each row $W\in\W_{\ri}^{c,p}$
underneath $W_{\ri,\filled}^{c,p}$ we select the rectangle $R\in W$
intersecting $\{\mi(A)\}\times\RR$ and all rectangles on the left
of $R$. We handle the spanning rows in the same way as the right-sticking-in
rows. Let $\APX_{\mi}$ denote the selected rectangles.

We want to split the remaining problem into a left and a right subproblem
for the areas $A_{\lef}:=[\lef(A),\mi(A))\times[0, \infty)$ and $A_{\ri}:=[\mi(A),\ri(A))\times[0, \infty)$
and for sets of rectangles $\R'_{\lef}$ and $\R'_{\ri}$, respectively,
and for certain sets of rays which we will define in the following.

\paragraph{Centered rows.}

First, we consider the centered rows $\W_{\ce}$. For each row $W\in\W_{\ce}$
we do the following. Let $W_{\lef}\subseteq W$ denote all rectangles
$R\in W$ with $R\subseteq A_{\lef}$ and let $W_{\ri}\subseteq W$
denote all rectangles $R\in W$ with $R\subseteq A_{\ri}$. If $W_{\lef}=W$
then we simply assign all rectangles in $W$ to $\R'_{\lef}$; similarly,
if $W_{\ri}=W$ then we add all rectangles in $W$ to $\R'_{\ri}$.
Suppose that $W_{\lef}\ne W\ne W_{\ri}$. Assume first that there
is a rectangle $R_{\mi}\in W$ with $\lef(R_{\mi})<\mi(A)<\ri(R_{\mi})$;
note that there can be at most one such rectangle. We divide $R_{\mi}$
into a left and a right half defined by $R_{\mi,\lef}:=R_{\mi}\cap A_{\lef}$
and $R_{\mi,\ri}:=R_{\mi}\cap A_{\ri}$. We define $c(R_{\mi,\lef}):=c(R_{\mi})$
and assign all rectangles in $W_{\lef}\cup\{R_{\mi,\lef}\}$ to
$\R'_{\lef}$. Also, we define $c(R_{\mi,\ri}):=\sum_{R\in W_{\lef}}c(R)+c(R_{\mi})$.
The intuition for this is that if the right subproblem selects $R_{\mi,\ri}$
then in our given problem we must also select all rectangles in $W_{\lef}$
and pay $\sum_{R\in W_{\lef}}c(R)$ for them. We assign all rectangles
in $W_{\ri}\cup\{R_{\mi,\ri}\}$ to $\R'_{\ri}$. If there is no
rectangle $R_{\mi}\in W$ with $\lef(R_{\mi})<\mi(A)<\ri(R_{\mi})$
then instead we increase the cost of the leftmost rectangle $R_{\mathrm{leftmost}}\in W_{\ri}$
by $\sum_{R\in W_{\lef}}c(R)$, i.e., we redefine $c(R_{\mathrm{leftmost}}):=c(R_{\mathrm{leftmost}})+\sum_{R\in W_{\lef}}c(R)$.
Finally, we assign all rectangles in $W_{\lef}$ to $\R'_{\lef}$ and
all rectangles in $W_{\ri}$ to $\R'_{\ri}$.

\paragraph{Right-sticking-in and spanning rows.}

Consider now the right-sticking-in rows $\W_{\ri}$ and let $W\in\W_{\ri}$.
A simple case arises if each rectangle $R\in W\setminus\APX_{\mi}$
satisfies that $\mi(A)\le\lef(R)$.  In particular, this happens for rows $W\in\rows_{\ri}^{c',p'}$ below $W_{\ri, \filled}^{c',p'}$
for some pair $(c',p')$ for which $W_{\ri, \filled}^{c',p'}$ is defined. In this case, we assign each rectangle in
$W\setminus \APX_\mi$ to $\R'_{\ri}$. Assume now that there is a rectangle $R\in W\setminus\APX_{\mi}$
intersecting $A_{\lef}$. Then, we assign each rectangle in $W$ to
$\R'_{\lef}$, i.e., to the rectangles for the \emph{left} subproblem.
In particular, this may include rectangles that intersect $A_{\ri}$
or that are even contained in $A_{\ri}$. However, such rectangles
might be needed to satisfy the demands of rays intersecting with $A_{\ri}$
which we will assign to the \emph{right }subproblem. Therefore, when
we define the rays for the left subproblem, those will include certain
additional \emph{artificial rays} $\L_{\ri}^{+}$ intersecting with $A_{\ri}$
which will ensure that the left subproblem selects sufficiently many
rectangles intersecting $A_{\ri}$ and, therefore, help covering the
demand of rays contained in $A_{\ri}$. We treat the spanning rows
in exactly the same way as the right-sticking-in rows.

\paragraph{Left-sticking-in rows.}

Finally, we consider the left-sticking-in rows $\W_{\lef}$. Consider
a row $W\in\W_{\lef}$. If each rectangle $R\in W\setminus\APX_{\mi}$
is contained in the interior of $A$ then, intuitively, we already selected some rectangles in $W$ and thus $W\setminus\APX_{\mi}$ behaves like a centered
row. Therefore, in this case we treat $W\setminus\APX_{\mi}$ in exactly the same way
as we treated the centered rows above. Assume now that there is a
rectangle $R\in W\setminus\APX_{\mi}$ with $\lef(R)\leq \lef(A)$.
If for each rectangle $R'\in W\setminus\APX_{\mi}$ we have that
$R'\cap A_{\ri}=\emptyset$ (i.e., $\ri(R')\le\mi(A)$) then we
assign each rectangle in $W\setminus\APX_{\mi}$ to $\R'_{\lef}$,
i.e., to the left subproblem. On the other hand, if there is a rectangle
$R'\in W\setminus\APX_{\mi}$ with $R'\cap A_{\ri}\ne\emptyset$
then we assign \emph{all} rectangles in $W\setminus\APX_{\mi}$
to $\R'_{\ri}$, i.e., to the right subproblem. Similarly as for
the right-sticking-in and the spanning rows, we will define artificial
rays in $\L_{\lef}^{+}$ for the right subproblem to ensure that from
such rows, the right subproblem selects sufficiently many rectangles
to cover enough demand from rays in $\L'$ that intersect with $A_{\lef}$.

\paragraph{Reference solutions.}
For the left and right subproblem, we define reference solutions $S^+_\lsub$ and $S^+_\rsub$. Intuitively, to define $S^+_\lsub$ we restrict $S^+$ to the rectangles contained in the left subproblem and we define $S^+_\rsub$ similarly. Moreover, whenever we cut a rectangle from $S$ into two pieces, we assign the left piece to $S^+_\lsub$ and the right piece to $S^+_\rsub$.
Formally, we define $S^+_\lsub:=\{R\in \R_{\lef}:\exists R'\in S^+ \text{ with } R\subseteq R'\}$ and $S^+_\rsub:=\{R\in \R_{\ri}:\exists R'\in S^+ \text{ with } R\subseteq R'\}$.

\paragraph{Rays and artificial rays for subproblems.}

It remains to define the sets of rays for the left and right subproblem,
respectively.
As mentioned above, for a ray $L\in\L'$ with $L\cap A_{\lef}\neq \emptyset$
we would like that its demand is partially satisfied by rectangles
from $\R'_{\lef}$, i.e., selected by the left subproblem, and partially
by rectangles from $\R'_{\ri}$, i.e., selected by the right subproblem.
Therefore, for each ray $L\in\L'$ with $L\cap A_{\lef}\neq \emptyset$ we
intuitively reduce its demand by a certain value; formally, we introduce
a ray in a set $\L'_{\lef}$ for the left subproblem corresponding
to $L$ with reduced demand. On the other hand, we introduce artificial
rays in a set $\L_{\lef}^{+}$ for the right subproblem to compensate
for this reduction. We perform a symmetric operation for the rays
in the right subproblem.

Formally, the rays $\L'_{\lef}$ and $\L'_{\ri}$ with their reduced
demands and the artificial rays $\L_{\lef}^{+}$ and $\L_{\ri}^{+}$
are defined by a function $f:A\rightarrow \{0,\dots,\sum_Rp(R)\}$ which
we will guess; this function is a step-function with only polylogarithmically
many steps. Its steps are defined by a partition $\Q=\{Q_{1},...,Q_{k}\}$
of $A$ such that each $Q\in\Q$ is an axis-parallel rectangle of
the form $Q=[x_{Q}^{L},x_{Q}^{R})\times[y_{Q}^{B},y_{Q}^{T})$ for
suitable values $x_{Q}^{L},x_{Q}^{R},y_{Q}^{B}\in\N$ and $y_{Q}^{T}\in \N\cup \{\infty\}$
and
either $Q\subseteq A_{\lef}$ or $Q\subseteq A_{\ri}$. Also, for
any two points $(t,s),(t', s')\in Q$ we have that $f(t,s)=f(t', s')$.
For each ray $L(s,t)\in\L'$ with $L\subseteq A_{\lef}$ we add a
ray $L'(s,t)$ to $\L'_{\lef}$ with a demand of $d(L'(s,t)):=d(L(s,t))-f(t,s)-p(\{R\in \APX^{\mi}:R\cap L(s,t)\neq \emptyset\})$.
Symmetrically, for each ray $L(s,t)\in\L'$ with $L\subseteq A_{\ri}$
we add a ray $L'(s,t)$ to $\L'_{\ri}$ with a demand of $d(L'(s,t)):=d(L(s,t))-f(t,s)-p(\{R\in \APX^{\mi}:R\cap L(s,t)\neq \emptyset\})$.
Additionally, we define a polylogarithmic number of artificial rays
$\L_{\lef}^{+}$ and $\L_{\ri}^{+}$. For each ``step'' $Q\in\Q$
of $f$ we introduce an artificial ray $L_{Q}$ that starts in the
point $(x_{Q}^{R}-1,y_{Q}^{B})$
(i.e., at the bottom-right integer point of $Q$) and is oriented vertically downwards.
We define its demand such that $d(L_{Q}):=f(t,s)$ for each $(t,s)\in Q$.
For each $Q\in\Q$ with $Q\subseteq A_{\lef}$ we denote by $\L_{\lef}^{+}$
the resulting set of rays, i.e., $\L_{\lef}^{+}=\{L_{Q}:Q\subseteq A_{\lef}\}$;
similarly, we define $\L_{\ri}^{+}=\{L_{Q}:Q\subseteq A_{\ri}\}$.
Note that all rays in $\L'_{\lef},\L'_{\ri},\L_{\lef}^{+},$ and $\L_{\ri}^{+}$
are uniquely defined from~$f$. Therefore, we say that they are \emph{induced}
by $f$.

It remains to consider the rays in ${L}'_{\mathrm{out}}$. We introduce sets of rays $\mathcal{L}'_{\lsub,\mathrm{out}}, \mathcal{L}'_{\rsub,\mathrm{out}}$ (independent of $f$) for the left and right subproblems, respectively.
Consider a ray $L(s,t)\in \mathcal{L}'_{\mathrm{out}}$.
We guess the value $p(\{R\in S^+_\lsub:R\cap L(s,t)\neq \emptyset\})$ by which the reference solution for the left subproblem covers $L(s,t)$. Then we add a ray $L'(s,t)$ with a demand of $d(L'(s,t)):=p(\{R\in S^+_\lsub:R\cap L(s,t)\neq \emptyset\})$ to $\mathcal{L}'_{\lsub,\mathrm{out}}$ and a ray $L''(s,t)$ with $d(L''(s,t)):=\max\{0, d(L(s,t))-d(L'(s,t))-p(\{R\in \APX^{\mi}:R\cap L(s,t)\neq \emptyset\})\}$ to $\mathcal{L}'_{\lsub,\mathrm{out}}$.

In the next lemma, we show that there exists a function $f$ and corresponding
induced rays which admit certain properties. Those will allow us to
partition the remaining problem into the left and the right subproblem.
\begin{lem}
\label{lem:guess-rays}There exists a step-function $f:A\rightarrow \{0,\dots,\sum_Rp(R)\}$
with only $(K\cdot M \log(T+p_{\max})/\varepsilon)^{O(1)}$ steps with the following properties. Let $\L'_{\lef},\L'_{\ri},\L_{\lef}^{+},$
and $\L_{\ri}^{+}$ be the rays induced by $f$. It holds that
\begin{itemize}
\item $|\L_{\lef}^{+}\cup\L_{\ri}^{+}|\le(K\cdot M \cdot \log(T\cdot p_{\max})/\varepsilon)^{O(1)}$
\item each ray $L\in\L_{\lef}^{+}$ is contained in $A_{\lef}$ and each
ray $L\in\L_{\ri}^{+}$ is contained in $A_{\ri}$,
\item the solution $S_{\lef}^{+}$ is a feasible solution
to the (left) subproblem $(A_{\lef},\L'_{\lef}\cup\L_{\ri}^{+}\cup \L'_{\lsub,\mathrm{out}} ,\R_{\lef})$,
\item the solution $S_{\ri}^{+}$ is a feasible
solution to the (right) subproblem $(A_{\ri},\L'_{\ri}\cup\L_{\lef}^{+}\cup \mathcal{L}'_{\rsub,\mathrm{out}},\R_{\ri})$.
\end{itemize}
\end{lem}
We will prove this lemma to Section~\ref{sec:ConstructionStepFunction}.
Algorithmically, we guess $f$ which we can do in time $2^{(K\cdot M\cdot  \log(n\cdot T\cdot p_{\max})/\varepsilon)^{O(1)}}$
since we have that $|\Q|\le(K\cdot M \cdot \log(T\cdot p_{\max})/\varepsilon)^{O(1)}$ and for the value of
$f$ corresponding to each $Q\in\Q$ (i.e., the value $f(s,t)$ for
each $(s,t)\in Q$) there are only $O(n\cdot p_{\max})$ options. We recurse
on the left and right subproblems $(A_{\lef},\L_{\lef}\cup\L_{\ri}^{+}\cup \L'_{\lsub,\mathrm{out}},\R_{\lef})$
and $(A_{\ri},\L_{\ri}\cup\L_{\lef}^{+}\cup \L'_{\rsub,\mathrm{out}},\R_{\ri})$. Let $\APX_{\lef}$
and $\APX_{\ri}$ denote the obtained solutions for them. 
Intuitively, we output $\APX_{\mi}\cup \APX_{\lef}\cup \APX_{\ri}$. Formally, we need to include also all rectangles on the left of the rectangles in these sets. Therefore, our
output $\APX$ is the set of all rectangles $R=[\lef(R),\ri(R)) \times [j, j+1) \in \R$ for which there exists a rectangle $R'=[\lef(R'),\ri(R')) \times [j, j+1)
\in \APX_{\mi}\cup \APX_{\lef}\cup \APX_{\ri}$ with $\lef(R)<\ri(R')$.

\subsection{Analysis}\label{sec:AnalysisMain}
For proving Lemma~\ref{lem:approx-RCP}, we need to show that our computed solution is feasible, prove that it has the claimed approximation ratio, and bound the running time of our algorithm. Let $\APX^{\operatorname{root}}$ denote the solution obtained for the root subproblem, i.e., where $A=[0, T)\times[0, \infty)$, $\R'=\R$ and $\L'=\L$. First, we prove feasibility.
\begin{lemma}\label{lem:feasibilityRoot}
	The set $\APX^{\operatorname{root}}$ is a feasible solution for the given instance of RCP.
\end{lemma}

\begin{proof}
	We show by induction, that the set $\APX$ computed as a solution for a given subproblem $(A,\L',\R')$ is a feasible
	solution to the subproblem $(A,\L',\R')$.
	Consider a subproblem $(A,\L',\R')$. 
	If $\lef(A)-\ri(A)=1$ then the claim follows from Lemma~\ref{lem:BaseCase}.
	So suppose this is not the case.
	We have to show that $\APX$ contains a prefix from the rectangles in each row and that it covers the rays $\L'$. 
	Recall that given $\APX_\lsub, \APX_\rsub$ and $\APX_\mi$, for each row $W\in \rows$ we select a rectangle $R=[\lef(R), \ri(R))\times[j, j+1)$ if there exists a rectangle $R'=[\lef(R'), \ri(R'))\times [j, j+1)\in \APX_\lsub \cup \APX_\rsub \cup \APX_\mi$ with $\ri(R')>\lef(R)$.
	This directly implies that $\APX$ contains a prefix in each row.	
	So it remains to show that each ray is covered by $\APX$.
	
	Consider a ray $L(s,t)\in \L'$.
	As a first step, we show that $\APX$ covers at least as much as $\APX_\lsub$, $\APX_\rsub$ and $\APX_\mi$ together.
	Let $W\in \rows$ be a row and let $R=[\lef(R),\ri(R))\times[j, j+1)\in W$ be a rectangle in this row.
	Furthermore let $W':=\{R'=[\lef(R'), \ri(R'))\times[j', j'+1)\in \R'_\lsub \cup \R'_\rsub \cup \APX_\mi:j'=j\}$ denote the rectangles that, intuitively, are in the row $W$, but in one of the subproblems. Recall that for some rows, we might split a rectangle into two parts.
	By construction, for every $t$ there is at most one rectangle $R\in W'$ with $\lef(R)\leq t<\ri(R)$.
	Thus $L(s,t)$ can intersect with at most one rectangle from $W'$.
	So it also intersects with at most one rectangle from  $(\APX_\lsub \cup \APX_\rsub \cup \APX_\mi)\cap W'$.
	Note that all rectangles $R\in W'$ have the same value $p(R)$.
	Suppose that there is a rectangle $R\in(\APX_\lsub \cup \APX_\rsub \cup \APX_\mi)\cap W'$ with $R\cap L(s,t)\neq \emptyset$.
	Then there is also a rectangle $R'\in W$ with $R'\cap L(s,t)\neq \emptyset$.
	Furthermore we have $\lef(R')<t\leq \ri(R)$. This implies $R'\in \APX$ as we add a rectangle $R'\in W$ to $\APX$ if there is a rectangle $R\in W'$ with $\lef(R')<\ri(R)$.
	So we obtain $p(\{R\in \APX\cap W:R\cap L(s,t)\neq \emptyset\})\geq p(\{R\in (\APX_\lsub \cup \APX_\rsub \cup \APX_\mi)\cap W':R\cap L(s,t)\neq \emptyset\})$ and therefore 
	\begin{align*}
		&p(\{R\in \APX:R\cap L(s,t)\neq \emptyset\})\geq p(\{R\in \APX_\lsub:R\cap L(s,t)\neq \emptyset\})\\
		&+p(\{R\in \APX_\rsub:R\cap L(s,t)\neq \emptyset\})
		+p(\{R\in \APX_\mi:R\cap L(s,t)\neq \emptyset\})
	\end{align*}
	
	Now we prove that $\APX$ covers $L(s,t)$.
	First suppose that $L(s,t)\in \mathcal{L}'_{\operatorname{out}}$. Then there is a ray $L'(s,t)$ with demand $d(L'(s,t))$  in the left subproblem and a ray with $L''(s,t)$ with demand $d(L''(s,t))$  in the right subproblem. Note that $\APX_\lsub$ and $\APX_\rsub$ cover the demands of these rays in the respective subproblems, as they are feasible solutions for them. By definition of $d(L''(s,t))$ we have $d(L(s,t))\leq d(L'(s,t))+d(L''(s,t))+p(\{R\in \APX_\mi:R\cap L(s, t)\neq \emptyset\})$. So we obtain the following:
	\begin{align*}
		p(\{R\in \APX:R\cap L(s,t)\neq \emptyset\})
		&\geq p(\{R\in \APX_\lsub:R\cap L(s,t)\neq \emptyset\})
		+p(\{R\in \APX_\rsub:R\cap L(s,t)\neq \emptyset\})\\
		&+p(\{R\in \APX_\mi:R\cap L(s,t)\neq \emptyset\})\\
		&\geq d(L'(s,t))+d(L''(s,t))+p(\{R\in \APX_\mi:R\cap L(s, t)\neq \emptyset\})\\
		&\geq d(L(s,t))
	\end{align*}
	Hence, $L(s,t)$ is covered by $\APX$. The same argument yields that every ray in $\mathcal{L}'_{\operatorname{out}}$ is covered by $\APX$.
	
	Now suppose that $L(s,t)$ intersects $A_\lsub$. Then there is a step $Q\in \Q$ of the function $f$ with $(t,s)\in Q$ and thus $Q\subseteq A_\lsub$.
	Let $\bar{L}(\bar{s},\bar{t})\in\L^+_\lsub$ be the ray at the bottom-right integer point of $Q$.
	Then in the left subproblem, there is a ray $L'(s,t)$ with demand $d(L'(s,t))=\max\{0,d(L(s,t))-f(t,s)-p(\{R\in \APX_\mi:R\cap L(s,t)\neq \emptyset\})\}$ and in the right subproblem, there is a the demand ray $\bar{L}(\bar{s}, \bar{t})$ with demand $d(\bar{L}(\bar{s}, \bar{t}))=f(t,s)$.
	As before, let $W$ be a row, let $R=[\lef(R),\ri(R))\times[j, j+1)\in W$ and let $W':=\{R'=[\lef(R'), \ri(R'))\times[j', j'+1)\in \R'_\lsub \cup \R'_\rsub \cup \APX_\mi:j'=j\}$ denote the rectangles that, intuitively, are in the row $W$, but in one of the subproblems.
	Suppose that there exists a rectangle $R\in W'\cap \APX_\rsub$ with $R\cap \bar{L}(\bar{s}, \bar{t})\neq \emptyset$.
	Note that there is at most one such rectangle in each row.
	As $t<\mi(A)$ and $Q\subseteq A_\lsub$ by construction, we also have $\bar{t}<\mi(A)$ and thus $\lef(R)<\mi(A)$.
	By construction of the rectangles $\R_\rsub$, this is only possible if there also exists a rectangle $R'\in W'$ with $\lef(R')\leq \lef(A)$. 
	So there exists a rectangle $R\in W'$ with $\ri(R)>\bar{t}\geq t$ and a rectangle $R'\in W'$ with $\lef(R')\leq \lef(A)\leq t$.
	As the rectangles in a row are consecutive there also exists a rectangle $R''\in W'$ with $\lef(R'')\leq t<\ri(R'')$.
	As the algorithm selects a prefix in each row and $R\in \APX_\rsub$, we also have $R''\in \APX_\rsub$. This implies $R''\cap L(s,t)\neq \emptyset$ as $s\geq \bar{s}$.
	So $p(\{R\in \APX_\rsub:R\cap L(s,t)\neq \emptyset\})\geq p(\{R\in \APX_\rsub:R\cap \bar{L}(\bar{s}, \bar{t})\neq \emptyset\})\geq d(\bar{L}(\bar{s}, \bar{t}))=f(t,s)$.
	This yields the desired result:
	\begin{align*}
		p(\{R\in \APX:R\cap L(s,t)\neq \emptyset\})
		&\geq p(\{R\in \APX_\lsub:R\cap L(s,t)\neq \emptyset\})
		+p(\{R\in \APX_\rsub:R\cap L(s,t)\neq \emptyset\})\\
		&+p(\{R\in \APX_\mi:R\cap L(s,t)\neq \emptyset\})\\
		&\geq d(L'(s,t))+f(t,s)+p(\{R\in \APX_\mi:R\cap L(s, t)\neq \emptyset\})\\
		&\geq d(L(s,t))
	\end{align*}
	The proof for the case that $L(s,t)$ intersects $A_\rsub$ is the same as the proof for the case that $L(s,t)$ intersects $A_\lsub$ when interchanging $\lsub$ and $\rsub$.
	This shows that every ray is covered and thus $\APX$ is a feasible solution for $(A,\L',\R')$. This completes the induction.
	
	Thus the set $\APX^{\operatorname{root}}$ is a feasible solution to the root subproblem. And as the root subproblem is equivalent to the RCP instance, the set $\APX^{\operatorname{root}}$ is also feasible for the RCP instance.
\end{proof}
As a next step, we bound our approximation ratio.
\begin{lemma}\label{lem:approximationRatioRoot}
	We have $c(\APX^{\operatorname{root}})\leq (2+O(\varepsilon))c(S)$ for any feasible solution $S$.
\end{lemma}
We will prove this lemma in Section~\ref{sec:ProofApproxRatio}. As a final step we bound our running time.
\begin{lemma}\label{lem:runningTime}
	The running time of the algorithm is bounded by $2^{(K\cdot M \cdot \log(n\cdot T\cdot P)/\varepsilon)^{O(1)}}$.
\end{lemma}
\begin{proof}
	As a first step, we show that the size of $\mathcal{L}'_{\operatorname{out}}$ is at most $(K\cdot M \cdot \log( T\cdot p_{\max})/\varepsilon)^{O(1)}$.
	Let $C_1$ be a constant such that the function $f$ from Lemma~\ref{lem:guess-rays} has  at most $(K\cdot M \log(T+p_{\max})/\varepsilon)^{C_1}$ steps.
	We show by induction starting from the root, that for a subproblem with area $A=[\lef(A),\ri(A))\times [0, \infty)$, we have $|\mathcal{L}'_{\operatorname{out}}|\leq (\log T - \log(\ri(A)-\lef(A)))\cdot(K\cdot M\cdot  \log(T+p_{\max})/\varepsilon)^{C_1}$.
	For the root we have $\ri(A)-\lef(A)=T$ and there are $0$ rays outside of $A$, so this is correct.
	So suppose the the hypothesis holds for a call for $A$ and we show that it also holds for the right and left subproblem.
	The rays in $\mathcal{L}'_{\operatorname{out}}$ are passed on to both subproblems (and do not intersect with $A_\lsub$ and $A_\rsub$).
	These are at most $(\log T - \log(\ri(A)-\lef(A)))\cdot (K\cdot M\cdot  \log(T+p_{\max})/\varepsilon)^{C_1}$.	
	The other rays not intersecting $A_\lsub$ in the left subproblem are the rays $\mathcal{L}^+_{\rsub}$.
	By Lemma~\ref{lem:guess-rays}, these are at most $(K\cdot M \cdot \log(T+p_{\max})/\varepsilon)^{C_1}$. This in total yields $(\log T - \log(\lef(A)-\ri(A))+1)\cdot(K\cdot M \cdot \log(T+p_{\max})/\varepsilon)^{C_1}=(\log T - \log((\lef(A)-\ri(A))/2))\cdot(K\cdot M\cdot  \log(T+p_{\max})/\varepsilon)^{C_1}$  not intersecting $A_\lsub$ in the left subproblem.
	The same argument yields the same bound for the right subproblem right subproblem.
	So we always have that the size of $\mathcal{L}'_{\operatorname{out}}$ is at most $\log T \cdot  (K\cdot M \log( T\cdot p_{\max})/\varepsilon)^{C_1}$.

	Using this, we can now prove that the running time of the algorithm is bounded by $2^{(K\cdot M \cdot \log(n\cdot T\cdot P)/\varepsilon)^{O(1)}}$.
	The recursion depth of our algorithm is $O(\log T)$. At each recursion step, we need to guess the numbers for each pair $c', p')$ and $\operatorname{set}\in \{\lsub, \spn, \rsub\}$ we guess whether $n^{c',p'}_{\operatorname{set}}<\frac{2K\cdot \log T}{\varepsilon}$ and up to $\frac{2K\cdot \log T}{\varepsilon}$ rows. For each row, there are at most $n$ options.
	As there are at most $M$ values for $c'$ and at most $O(\log p_{\max}/\varepsilon)$ possible values for $p'$, there are at most $O(M \log p_{\max}/\varepsilon)$ pairs $(c', p')$.
	So this step takes at most $2^{O((M \log T \log p_{\max}/\varepsilon)^2)}$.
	
	Furthermore, we need to guess the function $f$. The function has at most $(K\cdot M \cdot  \log(T\cdot p_{\max})/\varepsilon)^{O(1)}$ steps. 
	For each step $Q\in \Q$, we need to guess the function value in $Q$, which is bounded by $n\cdot p_{\max}$, and the boundaries of $Q$, for which there are at most $n$ options for the top and bottom one and $T$ options for the left and right one.
	So altogether, the guessing of $f$ can be done in $2^{(K\cdot M\cdot \log(n\cdot T\cdot p_{\max})/\varepsilon)^{O(1)}}$.
	
	The last step is to guess the demand for the rays in $\L'_{\operatorname{out}}$. 
	There are at most $(K\cdot M \cdot\log( T\cdot p_{\max})/\varepsilon)^{O(1)}$ such rays as shown above. And for each such ray, the guessed demand can be bounded by $n\cdot p_{\max}$, so this step takes $2^{(K\cdot M \cdot\log(n \cdot T\cdot p_{\max})/\varepsilon)^{O(1)}}$
	Altogether, at each recursion step there are at most $2^{(K\cdot M \log(n+T+p_{\max})/\varepsilon)^{O(1)}}$ recursive calls.
	
	For fixed guesses, the computation of $\APX$ can be done in time $O(n)$.
	The base case can be solved in time $(n \cdot p_{\max})^{O(|\mathcal{L'}_{\operatorname{out}}|)}\leq 2^{(K\cdot M \cdot \log(n\cdot T\cdot p_{\max})/\varepsilon)^{O(1)}}$.
	Altogether, this yields a running time of $ 2^{(K\cdot M \cdot \log(n\cdot T\cdot p_{\max})/\varepsilon)^{O(1)}}$.
\end{proof}

Altogether, we can now prove Lemma~\ref{lem:approx-RCP}.
\begin{proof}[Proof of Lemma~\ref{lem:approx-RCP}]
	By Lemma~\ref{lem:feasibilityRoot} the computed solution $\APX^{\operatorname{root}}$ is feasible.
	By Lemma~\ref{lem:approximationRatioRoot} we have $c(\APX)\leq (2+O(\varepsilon))c(S)$where $S$ denotes the optimal solution to the RCP instance.
	By Lemma~\ref{lem:runningTime}, the running time is bounded by $2^{(K\cdot M \cdot \log(n\cdot T\cdot P)/\varepsilon)^{O(1)}}$. So we can apply Lemma~\ref{lem:preprocessing}, which yields the claimed result by rescaling $\varepsilon$.
\end{proof}

\subsection{Proof of Lemma~\ref{lem:guess-rays}}\label{sec:ConstructionStepFunction}
We construct the step function $f$ separately for the left and right subproblem.
First, we introduce some notation.
Let $\mathcal{G}$ be the set of all pairs $(c', p')\in \N_0^2$, for which there exists $\operatorname{set}\in \{\lsub, \rsub, \spn\}$ such that $\rows^{c', p'}_{\operatorname{set}}\neq \emptyset$.
For each $\operatorname{set}\in \{\lsub, \rsub, \spn\}$ and a set of rows $\rows^{c', p'}_{\operatorname{set}}$ for which $W^{c', p'}_{\operatorname{set}}$ was defined let $\rows^{c', p'}_{\operatorname{set},\operatorname{top}}$ denote all rows in $\rows^{c', p'}_{\operatorname{set}}$ above $W^{c', p'}_{\operatorname{set}, \filled}$.
Also let $\mathcal{H}_{\operatorname{set}}$ denote all pairs $(c', p')$ for which $\rows^{c', p'}_{\operatorname{set},\operatorname{top}}\neq \emptyset$.
For each $\operatorname{set}\in \{\lsub, \rsub, \spn\}$, we apply the following to lemma to $\bar{W}:=\bigcup_{c', p'\in \mathcal{H}_{\operatorname{set} }}\rows^{c', p'}$, to the rectangles $\bar{S}:=S^+\cap  \bar{W}$ and the area $\bar{A}=A_\lsub$ if $\operatorname{set}=\lsub$ and the area $\bar{A}=A_\rsub$ if $\operatorname{set}\in \{\spn, \rsub\}$.
\begin{lemma}\label{lem:rayApproximation}
	Let $\bar{A}=[\bar{a}, \bar{b})\times [0, \infty)$ be an area and let $\bar{\rows} \subseteq \rows$ be rows spanning $\bar{A}$, i.e. for each $W\in \bar{\rows}$ there exist rectangles $R, R'\in W$ with $\lef(R)\leq \bar{a}$ and $\ri(R') \geq \bar{b}$.
	Let $\bar{\rows}=\bigcup_{c', p'} \bar{\rows}^{c',p'}$ be a partition, where a row $W\in \bar{\rows}$ if and only if the leftmost rectangle $R\in W$ has a cost of $c(R)=c'$ and satisfies $(1+\varepsilon)^{p'}\leq p(R)\leq (1+\varepsilon)^{p'+1}$.
	In addition, let $\bar{S}\subseteq \bigcup_{W\in \bar{\rows}}W$ be a set of rectangles that contains a prefix of the rectangles in each row $W\in \bar{\rows}$.
	For each $(c', p')$ let $\bar{n}^{c',p'}:=|\{W\in \bar{\rows}^{c',p'}: W\cap \bar{S} \neq \emptyset\}|$. There exists a step function $\bar{f}:A\rightarrow \{0,\dots,\sum_Rp(R)\}$ with $O((K|\mathcal{G}| \log T/\varepsilon)^2)$ steps such that for each $s, t\in \bar{A}$ we have
	\begin{equation}\label{eq:functionInterpolation}
		\bar{f}(t,s)\leq \sum_{R\in \bar{S}:R\cap L(s, t)\neq \emptyset}p(R)\leq \bar{f}(t,s)+\sum_{(c',p'):\exists R\in \bigcup_{W\in\bar{\rows}^{c',p'}}W \text{ with }R\cap L(s, t)\neq \emptyset}(1+\varepsilon)^{p'}\left\lfloor\frac{\varepsilon}{2\cdot K\log T}\bar{n}^{c',p'}\right\rfloor.
	\end{equation}
\end{lemma}
\begin{proof}
	Let $g(t, s):=p(\{R\in \bar{S}: R\cap L(s,t)\neq \emptyset\})$ be the total amount covered on a ray $L(s,t)$ for each $(t,s)\in \bar{A}$. 
	Note that as we select a prefix in each row and all rows are spanning $\bar{A}$, the function $g(t,s)$ is non-increasing in $t$ (for each fixed $s$). And as the rays are downward oriented, the function is non-decreasing in $s$ (for fixed $t$).
	Now we need to show that we can approximate $g$ by a function $\bar{f}$ with few steps.
	Let $X=\frac{8 \cdot K \cdot \log T}{\varepsilon}$. For each row $W\in \rows$ and $R=[\lef(R), \ri(R))\times [j, j+1)\in W$ let $\projy(W):=j$ be the $y$-coordinate of the rectangles in $W$.
	Consider a pair $(c',p')$.
	Intuitively, we first split $A$ at certain horizontal lines, such that between two consecutive lines, the value of $g$ does not change a lot because there are not a lot of rows between two consecutive lines.
	Formally we show that there exists a set $\mathcal{V}^{c',p'}=\{V_0, \dots V_k\}\subseteq \bar{\rows}^{c',p'}$ with $k\leq X$ such that $V_0$ is the bottom row in $\bar{\rows}^{c',p'}$, $V_k$ is the top row in $\bar{\rows}^{c',p'}$ and for each $k'<k$ we have $|\{W\in \bar{\rows}^{c',p'}: \projy(V_{k'})<\projy(W)<\projy(V_{k'+1})\}|\leq \frac{\bar{n}^{c',p'}}{X}$.
	If $|\bar{\rows}^{c,p}|\leq X$, we chose $\mathcal{V}^{c,p}=\bar{\rows}^{c,p}$. 
	Otherwise choose $V_0$ $V_0$ as the bottom row in $\bar{\rows}^{c',p'}$ and then recursively choose $V_{k'+1}$ with maximal $\projy(V_{k'+1})$ such that $|\{W\in \bar{\rows}^{c',p'}: \projy(V_{k'})<\projy(W)<\projy(V_{k'+1})\}|\leq \frac{\bar{n}^{c',p'}}{X}$. 
	As we chose $V_{k'+1}$ with maximal $\projy(V_{k'+1})$, when we add one more row to this set, we violate the inequality, i.e., we have $|\{W\in \bar{\rows}^{c',p'}: \projy(V_{k'})<\projy(W)\leq \projy(V_{k'+1})\}|\geq \frac{\bar{n}^{c',p'}}{X}$ for all $k'<k-1$, which shows $k\leq X$.
	To simplify notation let $\projy(\infty):=\infty$.
	Let $\mathcal{V}:=\bigcup_{c',p'}\mathcal{V}^{c',p'}\cup\{\infty\}=\{W_1, \dots, W_{\ell'}\}$ be such that $\projy(W_\ell)<\projy(W_{\ell+1})$ for all $\ell<\ell'$.
	The $y$-coordinates of the steps of $f$ are always $y^B_Q=\projy(W_\ell)$ and $y^T_Q=\projy(W_{\ell+1})$ for some $\ell<\ell'$.
	
	Let $W_\ell \in \mathcal{V}$ with $\ell<\ell'$. We will define $f$ such that each step $Q$ of $f$ fulfills $y^B_Q=\projy(W_\ell)$ and $y^T_Q=\projy(W_{\ell+1})$.
	Towards this, let $s':=\projy(W_\ell)$ and let $\mathcal{H}:=\{(c',p'):\exists W\in \bar{\rows}^{c',p'} \text{ with }\projy(W)\leq s'\}$ denote all pairs for which there exists a row below $W_\ell$.
	We show that there exists a set $\{(t_0',s'), \dots ,( t_k',s')\}$ with $k\leq X\cdot |\mathcal{G}|$ such that $t'_0=\lef(\bar{A})$, $t'_k=\ri(\bar{A})$ and for each $k'<k$ we have $g(t'_{k'},s')-g(t'_{k'+1}-1, s')\leq \sum_{(c',p')\in \mathcal{H}}(1+\varepsilon)^{p'+1}\lfloor \frac{\bar{n}^{c',p'}}{X}\rfloor$.
	If $\ri(\bar{A})-\lef(\bar{A})\leq X$, we can just chose $t_k'=\lef(\bar{A})+k$.
	Otherwise choose $t_0'$ as required and then recursively choose $t_{k'+1}$ maximal such that $g(t'_{k'}, s')-g(t'_{k'+1}-1, s')\leq \sum_{(c',p')\in \mathcal{H}}(1+\varepsilon)^{p'+1}\lfloor \frac{\bar{n}^{c',p'}}{X}\rfloor$.
	By doing this, for each $k'\leq k-2$ we have $g(t'_{k'}, s')-g(t'_{k'+1}, s')\geq (1+\varepsilon)^{p'+1}\frac{\bar{n}^{c',p'}}{X}$ for some group $(c',p')\in \mathcal{H}$.
	As $g(t'_0, s')\leq \sum_{(c',p')\in \mathcal{H}}(1+\varepsilon)^{p'+1}\bar{n}^{c',p'}$, this shows $$k\leq \sum_{(c', p')\in \mathcal{H}}\frac{(1+\varepsilon)^{p'+1}\bar{n}^{c',p'}}{(1+\varepsilon)^{p'+1}\cdot \bar{n}^{c',p'}/X}= X\cdot |\mathcal{H}|\leq X\cdot |\mathcal{G}|.$$
	We have a step $Q=[t'_{k'}, t'_{k'+1})\times [s', \projy(W_{\ell+1}))$, i.e. for each $(t,s)$ with $s'\leq s< \projy(W_{\ell+1})$ and $t'_{k'}\leq t < t'_{k'+1}$ let $\bar{f}(t,s):=g( t'_{k'+1}-1, s')$ be the value of $g$ at the bottom-right integer point of $Q$.
	
	It remains to show that $\bar{f}$ fulfills the requirements of the lemma. First note that the number of steps is bounded by $(X\cdot|\mathcal{G}|+1)|V|\leq (X+1)^2|\mathcal{G}|^2\leq O((C|\mathcal{G}|\log T/\varepsilon)^2)$.
	
	Now fix a step $Q=[x^L_Q, x^R_Q)\times [y^B_Q, y^T_Q)$ and let $(t,s)\in Q$. Let $\ell<\ell'$ such that $y^B_Q=\projy(W_\ell)$ and $y^T_Q=\projy(W_{\ell+1})$.
	As $g( t', s')$ is non-increasing in $t'$ and non-decreasing in $s'$, we have that $\min_{(t', s')\in Q\cap \N_0^2}g(t', s')=g(x^T_Q-1, y^B_Q)$.
	So $g(t, s)\geq g(x^R_Q-1, y^B_Q)=\bar{f}(t, s)$, which is the left inequality in \eqref{eq:functionInterpolation}.
	Thus, only the right inequality in \eqref{eq:functionInterpolation} remains to be proven.
	By construction of $\mathcal{V}$, we know that the values $g(t, s)$ and $g(t,y^B_Q)$ do not differ by much, as there are not many rows between $s$ and $y^B_Q$.
	Formally, we have $g(t, s)-g(t,y^B_Q)\leq \sum_{(c',p')\in \mathcal{H}}(1+\varepsilon)^{p'+1}\lfloor \frac{\bar{n}^{c,p}}{X}\rfloor$ as rectangles intersecting with the ray $L(s,t)$, but not the ray $L(y^B_Q, t)$ can only be in rows belonging to groups in $\mathcal{H}$, from each group $(c', p')\in \mathcal{H}$ there can be only $\lfloor \frac{\bar{n}^{c',p'}}{X}\rfloor$ rows and within a row there can be only one rectangle intersecting with the ray.
	
	When we defined $t'_{k'}$, we ensured that $g(t,y^B_Q)-g(x^R_Q-1,y^B_Q)\leq \sum_{(c',p')\in \mathcal{H}}(1+\varepsilon)^{p'+1}\lfloor \frac{\bar{n}^{c',p'}}{X}\rfloor$. 
	Altogether, this implies 
	\begin{align*}
		g(t,s)-\bar{f}(t,s)&=g(t,s)-g(x^R_Q-1,y^B_Q)\\
		&\leq 2\sum_{(c',p')\in \mathcal{H}}(1+\varepsilon)^{p'+1}\lfloor \frac{\bar{n}^{c',p'}}{X}\rfloor\\
		&\leq \sum_{(c',p')\in \mathcal{H}}(1+\varepsilon)^{p'}\lfloor \frac{\varepsilon \cdot \bar{n}^{c',p'}}{2\cdot K \cdot \log T}\rfloor
	\end{align*}
	This completes the proof of the lemma.
\end{proof}
For each $\operatorname{set}\in \{\lsub, \rsub, \spn\}$ let $f_{\operatorname{set}}$ be the function obtained from applying Lemma~\ref{lem:rayApproximation} to $\bar{W}:=\bigcup_{(c', p')\in \mathcal{H}_{\operatorname{set} }}\rows^{c', p'}$, to the rectangles $\bar{S}:=S^+\cap  \bar{W}$ and the area $\bar{A}=A_\lsub$ if $\operatorname{set}=\lsub$ and the area $\bar{A}=A_\rsub$ if $\operatorname{set}\in \{\spn, \rsub\}$.
Let $$f(t,s)=\begin{cases}
	f_\lsub(t,s) &\text{ if } (t,s)\in A_\lsub)\\
	f_\rsub(t,s)+f_\spn(t,s) &\text{ if }(t,s)\in A_\rsub
\end{cases}$$
The function $f$ has at most $O((K|\mathcal{G}|\log T/\varepsilon)^2)$ steps in $A_\lsub$. Each step of $f$ in $A_\rsub$ is the intersection of a step of $f_\rsub$ and a step of $f_\spn$, so $f$ has at most $O((K|\mathcal{G}|\log T/\varepsilon)^4)$ steps in $A_\rsub$.
Therefore, we also only have $O((K|\mathcal{G}|\log T/\varepsilon)^4)$ rays in $\mathcal{L}^+_\lsub$ and $\L^+_\rsub$. For each $(c', p')\in \G$ we have that $c'$ is one of at most $M$ different values and $p'$ is one of at most $O(\log p_{\max}/\varepsilon)$ different values, thus $|\mathcal{G}|=O(M \log p_{\max}/\varepsilon)$. So the number of steps of $f$ (and thus the number of rays in $\L^+_\lsub\cup \L^+_\rsub$) is bounded by $(K\cdot M \log(T+p_{\max})/\varepsilon)^{O(1)}$ steps.

So it remains to show that $S^+_\lsub$ is a feasible solution for the left subproblem and $S^+_\rsub$ is a feasible solution for the right subproblem. 

\begin{lemma}\label{lem:coveringSLR}
	For each ray $L(s,t)\in \L'$ we have
	\begin{align*}
		p(\{R\in S^+:R\cap L(s,t)\neq \emptyset\})&= p(\{R\in S^+_\lsub:R\cap L(s,t)\neq \emptyset\}) \\
		&+p(\{R\in S^+_\rsub:R\cap L(s,t)\neq \emptyset\})\\
		&+p(\{R\in \APX_\mi:R\cap L(s,t)\neq \emptyset\})
	\end{align*}
\end{lemma}
\begin{proof}
	Let $W\in \rows$ be a row, let $R=[\lef(R),\ri(R))\times[j, j+1)\in W$ be a rectangle in this row and let $W':=\{R'=[\lef(R'), \ri(R'))\times[j', j'+1)\in \R_\lsub \cup \R_\rsub \cup \APX_\mi:j'=j\}$ denote, intuitively,  the rectangles in row $W$ in the subproblems. 
	Let $L(s,t)\in \L'$.
	By construction there is at most one rectangle  $R\in (S^+_\rsub\cup S^+_\lsub \cup\APX_\mi)\cap W'$ with $\lef(R)\leq t<\ri(R)$.
	This implies
	\begin{align*}
		p(\{R\in S^+\cap W:R\cap L(s,t)\neq \emptyset\})&= p(\{R\in S^+_\lsub\cap W':R\cap L(s,t)\neq \emptyset\}) \\
		&+p(\{R\in S^+_\rsub\cap W':R\cap L(s,t)\neq \emptyset\})\\
		&+p(\{R\in \APX_\mi\cap W:R\cap L(s,t)\neq \emptyset\})
	\end{align*}
	And a union bound over all rows $W$ yields the lemma.
\end{proof}
Now we show that $S_\lsub$ and $S_\rsub$ are feasible solutions.
\begin{lemma}\label{lem:LRfeasibility}
	The set $S_{\lsub}$ is a feasible solution for the left subproblem and the set $S_\rsub$ is a feasible solution for the right subproblem.
\end{lemma}
\begin{proof}
	We start with $S_\lsub$.
	It follows directly from the construction of $S_\lsub$ that it contains a prefix in each row.
	The other necessary property for feasibility is to show that all demands are satisfied.
	For that, consider a ray $L(s,t)\in \L'_\lsub\cup \L^+_\rsub\cup \L'_{\lsub, \operatorname{out}}$.
	
	If $L(s,t)\in \L'_{\lsub, \operatorname{out}}$, recall that $d(L(s,t))=p(\{R\in S^+_\lsub:R\cap L(s,t)\neq \emptyset\})$. So such a demand ray is covered by $S^+_\lsub$.
	And if $L(t,s)\in \L^+_\rsub$ the demand is $f(s,t)$ and as $f_\lsub$ was chosen according to Lemma~\ref{lem:rayApproximation}, the solution $S_\lsub$ covers this ray.
	
	So suppose that $L(s,t)\in \L'_\lsub$ and let $L'(s,t)\in \L'$ be the ray with the same coordinates as $L(s,t)$.
	Then $d(L(s,t))=d(L'(s,t))-f(t,s)-p(\{R\in \APX_\mi:R\cap L(s,t) \neq \emptyset\})$. So the ray is covered if and only if $p(\{R\in S^+_\lsub :R\cap L(s,t)\neq \emptyset\})\geq d(L'(s,t))-f(t,s)-p(\{R\in \APX_\mi:R\cap L(s,t) \neq \emptyset\})$.
	Using the equality from Lemma~\ref{lem:coveringSLR}, this can be rearranged to
	\begin{align}\label{eq:CoveringReformulation}
		p(\{R\in S^+ :R\cap L(s,t)\neq \emptyset\})-d(L'(s,t))&\geq p(\{R\in S_\rsub:R\cap L(s,t)\neq \emptyset\})-f(s, t)
	\end{align}
	Let $\mathcal{H}:=\{(c',p'):\exists R\in \bigcup_{W\in \rows^{c,p}_{\lsub, \text{top}}}W \text{ with }R\cap L(s,t)\neq \emptyset\}$ denote all $(c', p')$ for which there exists a left-sticking-in row intersecting with the ray. By the second inequality of Lemma~\ref{lem:rayApproximation}, we can upper bound the right hand side of \eqref{eq:CoveringReformulation} by $\sum_{(c',p')\in \mathcal{H}}(1+\varepsilon)^{p'}\left\lfloor\frac{\varepsilon}{2\cdot K\log T}\bar{n}^{c',p'}\right\rfloor$. Se we have to show that the latter value is also a lower bound for the left hand side.
	
	Intuitively, if we would have $S$ instead of $S^+$ on the left hand side of \eqref{eq:CoveringReformulation}, we would have a lower bound of $0$ for the left hand side. But as $S^+$ covers substantially more than $S$ on each ray, we will get the desired lower bound.
	For any group $(c',p')\in \mathcal{H}$, the ray $L(s,t)$ intersects with a rectangle from a row in $\rows^{c',p'}_{\lsub, \text{top}}$.
	So by the last property of Lemma~\ref{lem:propertiesS+} we have \begin{align*}
		p(\{R\in S^{+}\cap \R_{\lsub}^{c',p'}:R\cap L(s,t)\neq \emptyset\})& \geq p(\{R \in S\cap\R_{\lsub}^{c',p'}: R\cap  L(s,t)\neq \emptyset\})\\
		&+(1+\varepsilon)^{p'}\left\lfloor\frac{\varepsilon}{2\cdot K\log T}\bar{n}^{c',p'}\right\rfloor.
	\end{align*}
	So altogether, we obtain
	\begin{align*}
		p(\{R\in S^{+}:R\cap L(s,t)\neq \emptyset\})
		&= p(\{R \in S\setminus\bigcap_{(c', p')\in \mathcal{H}}\R_{\lsub}^{c',p'}: R\cap  L(s,t)\neq \emptyset\})\\
		&+\sum_{(c', p')\in \mathcal{H}}p(\{R\in S^{+}\cap \R_{\lsub}^{c',p'}:R\cap L(s,t)\neq \emptyset\})\\
		&\geq p(\{R \in S\setminus\bigcap_{(c', p')\in \mathcal{H}}\R_{\lsub}^{c',p'}: R\cap  L(s,t)\neq \emptyset\})\\
		&+\sum_{(c', p')\in \mathcal{H}} p(\{R \in S\cap\R_{\lsub}^{c',p'}: R\cap  L(s,t)\neq \emptyset\})\\
		&+\sum_{(c', p')\in \mathcal{H}}(1+\varepsilon)^{p'}\left\lfloor\frac{\varepsilon}{2\cdot K\log T}\bar{n}^{c',p'}\right\rfloor.\\
		&\geq p(\{R\in S:R\cap L(s,t)\neq \emptyset\})\\
		&+\sum_{(c', p')\in \mathcal{H}}(1+\varepsilon)^{p'}\left\lfloor\frac{\varepsilon}{2\cdot K\log T}\bar{n}^{c',p'}\right\rfloor.\\
		&\geq d(L'(s,t)) 
		+\sum_{(c', p')\in \mathcal{H}}(1+\varepsilon)^{p'}\left\lfloor\frac{\varepsilon}{2\cdot K\log T}\bar{n}^{c',p'}\right\rfloor.
	\end{align*}
	This yields \eqref{eq:CoveringReformulation} and thus shows that the ray $L(s,t)$ is covered. So $S^+_\lsub$ is a feasible solution for the left subproblem.
	The proof for $S^+_\rsub$ being a feasible solution for the right subproblem is analogous.
\end{proof}
This also completes the proof of Lemma~\ref{lem:guess-rays}.
\subsection{Proof of Lemma~\ref{lem:approximationRatioRoot}}\label{sec:ProofApproxRatio}
First, we show that the constructed reference solutions $S^+_\lsub$ and $S^+_\rsub$ are not too expensive. For that, we use the introduced cost function $c_\APX$. Whether a row is centered, left-sticking-in or right-sticking-in depends on the area $A$ of the subproblem, so also $c_\APX$ depends on $A$. Therefore, let \begin{align*}
	c_{\APX}(A',S'):=\:&2\cdot c(S'\cap\R_{\ce}(A'))+c(S'\cap\R_{\ri}(A'))+c(S'\cap\R_{\spn}(A'))\\
	&+c(\{R\in S'\cap\R_{\lef}(A'):\lef(R)\le\lef(A')\})
	+2\cdot c(\{R\in S'\cap\R_{\lef}(A'):\lef(R)>\lef(A')\})
\end{align*}
where $\R_{\ce}(A')$, $\R_{\ri}(A')$, $\R_{\spn}(A')$ and $\R_{\lef}(A')$ denotes the centered, right-sticking-in, spanning and left-sticking-in rows w.r.t. the area $A'$.
During the algorithm, we redefine the cost of some rectangles in centered rows (i.e., the rectangle $\R_{\operatorname{leftmost}}$) or left-sticking-in rows $W$ (when we treat $W\setminus \APX_\mi$ as a centered row), therefore we denote by $c'(R)$ the cost of a rectangle $R$ after redefining the costs.
\begin{lemma}\label{lem:costLRsubproblem}
	We have $c_\APX(A, S^+)\geq  c'_\APX(A_\lsub, S^+_\lsub) +c'_\APX(A_\rsub, S^+_\rsub)+c(\APX_\mi)$.
\end{lemma}
\begin{proof}
	Let $W\in \rows$ be a row and let $R=[\lef(R),\ri(R))\times[j, j+1)\in W$ be a rectangle in this row. Furthermore let $W':=\{R'=[\lef(R'), \ri(R'))\times[j', j'+1)\in \R_\lsub \cup \R_\rsub \cup \APX_\mi:j'=j\}$ denote the corresponding rectangles in the left and right subproblem and the selected rectangles.
	We show that
	\begin{equation}\label{eq:costS+perRow}
		c_\APX(A, S^+\cap W)\geq  c'_\APX(A_\lsub, S^+_\lsub\cap W') +c'_\APX(A_\rsub, S^+_\rsub\cap W')+c(\APX_\mi\cap W').
	\end{equation}
	We make a case distinction by the type of row $W$.
	
	First consider a centered row $W$.
	If $W_\lsub=W$ or $W_\rsub=W$ we have $S^+_\lsub\cap W'=W$ or $S^+_\rsub\cap W'=W$, which directly yields \eqref{eq:costS+perRow}.
	So now assume that $W_\lsub\neq W \neq W_\rsub$.
	Then the row $W_\lsub\cup \{R_{\mi, \lsub}\}$ is right-sticking-in in $A_\lsub$ and the row $W_\rsub \cup \{R_{\mi, \lsub}\}$ is left-sticking-in in $A_\rsub$.
	Suppose that $R_\mi$ exists and $\mi(W)\in S^+$.
	Then $S^+_\lsub \cap W'=W_\lsub \cup \{R_{\mi, \lsub}\}$ and $S^+_\rsub\cap W'=S^+\cap W_\rsub \cup \{R_{\mi, \rsub}\}$.
	Recall that $c(R_{\mi, \rsub})=c(W_\lsub\cup \{R_{\mi, \lsub}\})$ and $c(R_{\mi, \lsub})=c(R_\mi)$.
	So \begin{align*}
		&c_\APX(A_\lsub, S^+_\lsub\cap W')+c_\APX(A_\rsub, S^+_\rsub\cap W')\\
		&=c(W_\lsub \cup \{R_{\mi, \lsub}\})+c(R_{\mi, \rsub})+2c( S^+\cap W_\rsub)\\
		&=2c(W_\lsub \cup \{R_\mi\} \cup (S^+\cap W_\rsub))=2c(S^+\cap W)\\
		&=c_\APX(A, S^+\cap W)
	\end{align*}
	And if $\mi(W)\not \in S^+$, we have $S_\lsub \cap W'=S^+\cap W$. This yields $c_\APX(S^+\cap W)=2c(S^+\cap W)\geq c(S^+\cap W)=c_\APX(A_\lsub, S_\lsub\cap W_\rsub)$ and thus yields \eqref{eq:costS+perRow} as well.
	So suppose that $R_\mi $ does not exist and $R_{\mathrm{leftmost}}\in S^+$. Recall that $c'(R_{\mathrm{leftmost}})=c(R_{\mathrm{leftmost}})+c(W_\lsub)$, which yields
	\begin{align*}
		&c'_\APX(A_\lsub, S^+_\lsub\cap W')+c'_\APX(A_\rsub, S^+_\rsub\cap W')\\
		&=c(W_\lsub )+c'(R_{\mathrm{leftmost}})+2c( S^+\cap W_\rsub\setminus R_{\mathrm{leftmost}})\\
		&=2c(W_\lsub )+c(R_{\mathrm{leftmost}})+2c( S^+\cap W_\rsub\setminus R_{\mathrm{leftmost}})\\
		&\leq 2c(W_\lsub  \cup (S^+\cap W_\rsub))=2c(S^+\cap W)\\
		&=c_\APX(A, S^+\cap W)
	\end{align*}
	And if $R_{\mathrm{leftmost}}\in S^+$, we again have $S_\lsub \cap W'=S^+\cap W$ which yields $c_\APX(S^+\cap W)=2c(S^+\cap W)\geq c(S^+\cap W)=c_\APX(A_\lsub, S_\lsub\cap W_\rsub)$ and thus completes the proof of \eqref{eq:costS+perRow} for centered rows.
	
	Next suppose that $W$ is a right-sticking-in row.
	Furthermore, suppose that each rectangle in $R\in W\setminus \APX_\mi$ satisfies $\mi(A)\leq \lef(R)$. Note that $W'\cap R'_\rsub$ is also right-sticking-in.
	Then $S^+_\lsub\cap W'=(S^+\cap W )\setminus \APX_\mi$ and thus
	\begin{align*}
		c_\APX(A, S^+\cap W)&=c(S^+\cap W)=c(\APX_\mi\cap W')+c(S^+_\lsub \cap W')\\
		&=c'_\APX(A_\rsub, S^+_\rsub\cap W')+c(\APX_\mi\cap W')
	\end{align*}
	So suppose that there is a rectangle $R\in W \setminus \APX_\mi$ intersecting $A_\lsub$.
	For such a row $S^+_\lsub\cap W'=S^+\cap W$ and $W_\lsub$ is right-sticking-in in $A_\lsub$.
	Thus $c_\APX(A, S^+\cap W)=c(S^+\cap W)=c(S^+_\lsub\cap W')=c_\APX(A_\lsub, S^+_\lsub\cap W')$, completing the proof of \eqref{eq:costS+perRow} for right-sticking-in rows.
	The proof for spanning rows is exactly the same as for right-sticking-in rows.
	
	Finally suppose that $W$ is a left-sticking-in row.
	Furthermore suppose that each rectangle $R\in W\setminus \APX_\mi$ is contained in the interior of $A$. The same argument as for centered rows shows that $c_\APX'(A_\lsub,  S^+_\lsub\cap W')+c_\APX'(A_\rsub,  S^+_\rsub\cap W')\leq 2c(W\setminus \APX_\mi)$.
	As $W$ is left-sticking-in there must be a rectangle $R\in W$ with $\lef(R)\leq \lef(A)$, which is therefore not contained in the interior of $A$.
	So $\APX_\mi \cap W\neq \emptyset$ and thus $c_\APX(A,S^+\cap W)\geq c(\APX_\mi\cap W)+2c(W\setminus \APX_\mi)$, which yields \eqref{eq:costS+perRow}.
	So suppose that there is a rectangle $R\in W\setminus \APX_\mi$ with $\lef(R)\leq \lef(A)$. In this case, we either assign all rectangles in $ W\setminus \APX_\mi$ to the left subproblem or assign them all to the right subproblem. In both cases, the row is left-sticking-in in the respective subproblem. So in both cases \begin{align*}
		&c'_\APX(A_\lsub, S^+_\lsub\cap W')+c'_\APX(A_\rsub, S^+_\rsub\cap W')+c(\APX_\mi\cap W')\\
		&\leq c(\{R\in W\setminus\APX:\lef(R)\leq \lef(A)\})+2c(\{R\in W\setminus\APX:\lef(R)> \lef(A)\})+c(\APX_\mi\cap W')\\
		&=c_\APX(A, S^+\cap W)
	\end{align*}
	This completes the proof of \eqref{eq:costS+perRow} for left-sticking-in rows.
	And adding up \eqref{eq:costS+perRow} for each row $W$ yields the lemma.
\end{proof}
Now, we show that $\APX$ is not more expensive than the selected rectangles $\APX_\mi$ and the solutions $\APX_\lsub$ and $\APX_\rsub$ combined.
\begin{lemma}\label{lem:costAPX}
	We have $c(\APX)\leq c'(\APX_\lsub) +c'(\APX_\rsub)+c(\APX_\mi)$.
\end{lemma}
\begin{proof}
	Let $W\in \rows$ be a row and let $R=[\lef(R),\ri(R))\times[j, j+1)\in W$ be a rectangle in this row. Furthermore let $W':=\{R'=[\lef(R'), \ri(R'))\times[j', j'+1)\in \R_\lsub \cup \R_\rsub \cup \APX_\mi:j'=j\}$ denote the corresponding rectangles in the left and right subproblem and the selected rectangles in $\APX_\mi$. We show that $c(\APX\cap W)\leq c(\APX_\lsub\cap W') +c(\APX_\rsub\cap W')+c(\APX_\mi\cap W)$.
	Note that $\APX_\mi \cap W$ is always a prefix of row $W$.
	
	First suppose that the rectangle $R_\mi$ exists and the algorithm splits this rectangle into $R_{\mi, \lsub}$ and $R_{\mi, \rsub}$, which happens for some centered or left-sticking-in rows. In this case we have $c(R_{\mi, \lsub})=c(R_\mi)$ and $c'(R_{\mi, \rsub})=c(R'_\lsub\cap W')$. So if $\APX_\rsub \cap W'=\emptyset$ and $R_{\mi, \lsub} \not \in \APX_\lsub$, we have $\APX\cap W=(\APX_\mi\cap W)\cup (\APX_\lsub\cap W')$ and thus $c(\APX\cap W)=c(\APX_\mi\cap W)+ c(\APX_\lsub\cap W')$. And if $\APX_\rsub \cap W'=\emptyset$ and $R_{\mi, \lsub} \in \APX_\lsub$, we have $\APX\cap W=(\APX_\mi\cap W)\cup (\APX_\lsub\cap W')\setminus \{R_{\mi, \lsub}\} \cup \{R_\mi\}$ and thus $c(\APX\cap W)=c(\APX_\mi\cap W)+ c(\APX_\lsub\cap W')-c(R_{\mi, \lsub})+c(R_\mi)=c(\APX_\mi\cap W)+ c(\APX_\lsub\cap W')$.
	And if $\APX_\rsub \cap W'\neq \emptyset$ we have $\APX\cap W=(\APX_\mi\cap W)\cup (\R'_\lsub\cap W') \cup (\APX_\rsub\cap W') \setminus \{R_{\mi, \lsub},R_{\mi, \rsub}\} \cup \{R_\mi\}$ and thus \begin{align*}
		c(\APX\cap W)&=c(\APX_\mi\cap W)+ c(\R'_\rsub\cap W'\setminus\{R_{\mi, \lsub}\})+c(\APX_\rsub\cap W'\setminus\{R_{\mi, \rsub}\})+c(R_\mi)\\
		&=c(\APX_\mi\cap W)+ c'(R_{\mi, \rsub})+c(\APX_\rsub\cap W'\setminus\{R_{\mi, \rsub}\})\\
		&=c(\APX_\mi\cap W)+c'(\APX_\rsub\cap W')
	\end{align*}
	This yields the lemma in the case where  $R_\mi$ exists and the algorithm splits this rectangle into $R_{\mi, \lsub}$ and $R_{\mi, \rsub}$.
	
	So suppose that this is not the case. Note that by construction we have $W=(\APX_\mi\cap W')\cup (\R'_\lsub \cap W') \cup (\R'_\rsub \cap W')$. If $\APX_\rsub \cap W'=\emptyset$ then $\APX\cap W=(\APX_\mi\cap W')\cup (\APX_\lsub \cap W')$, which yields the lemma. So suppose that $\APX_\rsub\cap W'\neq \emptyset$ and let $R$ be the leftmost rectangle of $\APX_\rsub$. Then $R$ is also the leftmost rectangle of $\R'_\rsub \cap W'$ and  by construction we have $c'(R)=c(R)+c(\R'_\lsub\cap W')$. So
	\begin{align*}
		c(\APX\cap W)&=c(\APX_\mi\cap W')+c(\R'_\lsub \cap W') c(\APX_\rsub \cap W')\\
		&=c(\APX_\mi\cap W')+c'(\APX_\rsub \cap W')
	\end{align*}
	which completes the proof.
\end{proof}
Now we can combine Lemmas~\ref{lem:propertiesS+},\ref{lem:costLRsubproblem} and \ref{lem:costAPX} to obtain the following approximation ratio.
\begin{lemma}\label{lem:approximationFactorRecursion}
	Let $S$ be any feasible solution for a subproblem with region $A$. Then
	$$c(\APX)\leq \left(1+\frac{\varepsilon \log(\ri(A)-\lef(A))}{\log T}\right)c_{\APX}(A, S)$$
\end{lemma}
\begin{proof}
	We prove the lemma by induction on $b-a$. If $b-a=1$, the problem can be solved exactly according to Lemma~\ref{lem:BaseCase}. So suppose $b-a>1$. Then $S^+_\lsub$ and $S^+_\rsub$ are feasible solutions for the left and right subproblem. By the induction hypothesis, we have
	\begin{equation*}
		c(\APX_\lsub)\leq \left(1+\frac{\varepsilon \log((b-a)/2)}{\log T}\right)c_{\APX}(A_\lsub, S^+_\lsub)
	\end{equation*}
	\begin{equation*}
		c(\APX_\rsub)\leq \left(1+\frac{\varepsilon \log((b-a)/2)}{\log T}\right)c_{\APX}(A_\rsub, S^+_\rsub)
	\end{equation*}
	Together with Lemmas~\ref{lem:costAPX} and \ref{lem:costLRsubproblem} we obtain
	\begin{align*}
		c(\APX)&\leq c(\APX_\lsub)+c(\APX_\rsub)+c(\APX_\mi)\\
		&\leq \left(1+\frac{\varepsilon \log((b-a)/2)}{\log T}\right)\cdot \left(c_{\APX}(A_\lsub, S^+_\lsub)+c_{\APX}(A_\rsub, S^+_\rsub)+c(\APX_\mi)\right)\\
		&\leq \left(1+\frac{\varepsilon \log((b-a)/2)}{\log T}\right)\cdot c_\APX(A, S^+)\\
		&\leq \left(1+\frac{\varepsilon( \log(b-a)-1)}{\log T}\right)\cdot \left(1+\frac{\varepsilon}{2\log T}\right)c_\APX(A, S)\\
		&\leq \left(1+\frac{\varepsilon \log(b-a)}{\log T}\right)\cdot c_\APX(A, S)
	\end{align*}
	In the second to last step, we used Lemma~\ref{lem:propertiesS+}. This completes the proof.
\end{proof}
This directly the approximation ratio stated in Lemma~\ref{lem:approximationRatioRoot}.
\begin{proof}[Proof of Lemma~\ref{lem:approximationRatioRoot}.]
	Recall that for the root subproblem $A=[0, T)\times[0, \infty)$. By Lemma~\ref{lem:approximationFactorRecursion}, we know that $c(\APX^{\operatorname{root}})\leq \left(1+\frac{\varepsilon \log(\ri(A)-\lef(A))}{\log T}\right)c_{\APX}(A, S)\leq (1+\varepsilon)c_{\APX}(A, S)\leq (2+2\varepsilon)c(S)$.
\end{proof}

\section{Weighted Tardiness}

\label{apx:weighted-tardiness} In this section prove Theorem~\ref{thm:main-tardiness},
i.e., we present a $(1+\epsilon)$-approximation algorithm for the
weighted tardiness objective for any constant $\epsilon>0$ with a
running time of $2^{\mathrm{poly}(\log(n+p_{\max}))}$. Recall
that in this setting each job has a due date $d_{j}\in\N$ and a weight
$w_{j}\in\N$ and its cost function is $\cost_{j}(t)=w_{j}(t-d_{j})$
for each $t\in\RR$. First, we reduce the problem to a special case
of RCP in which each instance has certain structural properties. Then
we give a recursive algorithm which is a $(1+\epsilon)$-approximation
algorithm for this setting of RCP. In particular, due to the additional
structural properties this algorithm is much simpler than the algorithm
for the general case of RCP (see Section~\ref{subsec:GSP-algorithm})
and at the same time achieves an approximation ratio of only $1+\epsilon$,
rather than $2+\epsilon$.

Formally, in our special case of RCP, each instance, given by a set
of rectangles $\R$ and a set of rays $\L$, is \emph{$\delta$-well-structured}
(for some $\delta>0$ that our running time will depend on) which means that for each rectangle $R=[\lef(R),\ri(R))\times[j,j+1)\in\R$
and the (unique) value $h\in\N$ with $2^{h}\le\ri(R)-\lef(R)<2^{h+1}$
we have that $\lef(R)$ and $\ri(R)$ are integral multiples of $\delta\cdot 2^{h}$.
We reduce the weighted tardiness problem to this special case of RCP
in the following lemma,
which we will prove in Section~\ref{sec:TardinessReduction}.

In the following, assume without loss of generality that $\epsilon = 2^{-k}$ for some $k\in\mathbb N$.
\begin{lemma}
\label{lem:tardiness-reduction}Given an $\alpha$-approximation algorithm
for $\delta$-well-structured instances of RCP with a running time of $f(n,K, \max_R\ri(R), p_{\max},\delta)$,
there is a $\alpha(1+\varepsilon)$-approximation algorithm for minimizing
	weighted tardiness in time $f((n p_{\max})^{O(1)}, (1/\epsilon)^{O(1/\epsilon^3)}, (n p_{\max})^{O(1)}, \epsilon/32)$. 
\end{lemma}

Using this reduction,
it suffices to solve $\delta$-well-structured instances of RCP. We
will present an algorithm for this task, corresponding to the following
lemma, in Section~\ref{sec:TardinessRCP}. 
\begin{lemma} \label{lem:tardiness-algorithm}There
is a $(1+\varepsilon)$-approximation algorithm for $\delta$-well-structured
instances of RCP with a running time of $2^{(1/\delta \cdot K \cdot \log (n \cdot p_{\max}\cdot \max_R\ri(R))/\varepsilon)^{O(K)}}$. 
\end{lemma} 
Together,
Lemmas~\ref{lem:tardiness-reduction} and \ref{lem:tardiness-algorithm}
yield Theorem~\ref{thm:main-tardiness}.

\subsection{Reduction}\label{sec:TardinessReduction}
Recall that the $x$-coordinates for the rectangles in RCP are derived from the milestones
from \Cref{lem:milestones}. We modify the construction here, achieving the same properties,
but additionally an equivalent of $\delta$-well-structuredness.
Recall also that $T = O(\max_{j\in J} r_j + \sum_{j\in J} p_j)$ is an upper bound on the time when the last job is finished.
\begin{lemma}\label{lem:milestones-tardiness}
	Consider an instance of weighted tardiness minimization.
	For each job $j$ we can in polynomial time construct a sequence $m_0(j),m_1(j),m_2(j),\dotsc,m_{f_j}(j)\in\mathbb N$
	where
	\begin{enumerate}
		\item $\cost_j(m_{i+1}(j)) \le (1 + \epsilon) \cdot \cost_j(m_{i}(j) + 1)$ for all $0 \le i < f_j$;
		\item $\cost_j(m_{i+1}(j) + 1) > (1 + \epsilon / 4) \cdot \cost_j(m_{i}(j) + 1)$ for all $0 \le i < f_j - 1$;
		\item $m_0(j) = r_j$;
		\item $m_1(j) = d_j$;
		\item $m_{f_j}(j) = T$;
		\item let $h\in \mathbb N$ with $2^{h} \le m_{i+1}(j) - m_{i}(j) < 2^{h+1}$. Then $m_{i}(j), m_{i+1}(j) \in \delta 2^h \cdot \mathbb Z$, where $\delta = \epsilon / 32$.
	\end{enumerate}
\end{lemma}
\begin{proof}
	We set $m_0(j) = r_j$ and $m_1(j) = d_j$.
	Then for $i \ge 1$ suppose we have already defined $m_{i}(j)$.
	Let $h'\in\mathbb N$ with $2^{h'} \le m_i(j) - d_j < 2^{h'+1}$. We define $m_{i+1}(j) \in\{m_{i}(j)+1, m_{i}(j)+2,\dotsc,T\}$ to be
	the minimal number that is at least as large as the \emph{second} smallest $t\in \frac{\epsilon}{2} 2^{h'} \cdot \mathbb Z$ with $t > m_{i}(j)$;
	or $T$ if this does not exists (because it would exceed $T$). Let $f_j \in\mathbb N$ be the smallest index $i$ where $m_i(j)$ equals $T$.

	It is easy to check that Properties~1 and~2 hold for $i=0$.
	Let $1 \le i < f_j$.
	Then
	\begin{multline*}
		\cost_j(m_{i+1}(j)) = w_j (m_{i+1}(j) - d_j) \le w_j (m_{i}(j) + 1 + \epsilon 2^{h'} - d_j) \\
		\le (1 + \epsilon)(m_{i}(j) + 1 - d_j) = (1 + \epsilon)\cost_j(m_{i}(j) + 1) .
	\end{multline*}
	Furthermore, for $1 \le i < f_j - 1$,
	\begin{multline*}
		\cost_j(m_{i+1}(j) + 1) = w_j (m_{i+1}(j) + 1 - d_j) > w_j (m_{i}(j) + 1 + \epsilon 2^{h'-1} - d_j) \\
		\ge (1 + \epsilon/4)(m_{i}(j) + 1- d_j) = (1 + \epsilon/4)\cost_j(m_{i}(j) + 1) .
	\end{multline*}
	Consider now Property~6.
	If $m_{i+1}(j) - m_i(j) \le 32/\epsilon$ then the claim is trivial, since $\epsilon / 32 \cdot 2^h \le 1$,
	and $m_i(j), m_{i+1}(j)\in\mathbb Z$.
	Hence, assume otherwise.
	Let $h\in\mathbb N$ with $2^{h} \le m_{i+1}(j) - m_i(j) < 2^{h+1}$.
	It holds that that $m_{i}(j) - d_j < 2^{h}$,
	since otherwise $m_{i+1}(j) < m_{i}(j) + \epsilon 2^h$, a contradiction to the definition of $h$.
	It follows that $m_{i+1}(j) \in \frac{\epsilon}{2} 2^{h} \cdot \mathbb Z \subseteq \frac{\epsilon}{32} 2^{h} \cdot \mathbb Z$.

	For the discreteness of $m_i(j)$, we will argue that the differences $m_{i+1}(j) - m_{i}(j)$ for $i > 1/\epsilon$ do not grow too quickly. Note that we have
	\begin{equation*}
		m_{i+1}(j) - m_i(j) > \epsilon (m_i(j) - d_j) / 4 .
	\end{equation*}
	Furthermore, for all $i > 2/\epsilon$ we have
	\begin{equation*}
		m_{i+1}(j) - m_i(j) \le 2 \epsilon (m_i(j) - d_j) .
	\end{equation*}
	Assuming that $\epsilon$ is sufficiently small it follows that
	\begin{equation*}
		m_{i+1}(j) - d_j \le 2(m_i(j) - d_j) .
	\end{equation*}
	Thus,
	\begin{equation*}
		m_{i+1}(j) - m_i(j) \le 2 \epsilon (m_i(j) - d_j) \le 4 \epsilon (m_{i-1} - d_j) < 16 (m_{i}(j) - m_{i-1}(j)) .
	\end{equation*}
	It follows that $2^{h''} \le m_{i}(j) - m_{i-1}(j)$, where $2^{h''} \ge 2^h / 16$.
	By the argument used to show $m_{i+1}(j) \in \frac{\epsilon}{2} 2^h \cdot \mathbb Z$
	we obtain $m_{i}(j) \in \frac{\epsilon}{2} 2^{h''} \cdot \mathbb Z \subseteq \frac{\epsilon}{32} 2^{h} \cdot \mathbb Z$.
\end{proof}
Since this construction satisfies all properties of \Cref{lem:milestones}, we can use it
with the reduction in \Cref{sec:Reduction}.
Recall that the rectangles $[a, b) \times [j', j'+1)\in\R$ all had the form that there
is some $j\in J$ and $i$ with $a = m_i(j)$ and $b = m_{i+1}(j)$.
Assuming that $i \ge 1$, the previous lemma ensures that for $h\in\mathbb N$ with $2^h \le b - a < 2^{h+1}$ we have
$a, b\in \frac{\epsilon}{32} 2^h \cdot \mathbb Z$.
The first rectangle of each job with $x$-coordinates $[m_0(j), m_1(j)) = [r_j, d_j)$ may be very
wide and $r_j, d_j$ cannot be rounded. This rectangle, however, has cost zero, so the reduction in
\Cref{sec:Reduction} anyway removes it in the preprocessing.
The instance size and $K$ are bounded as in \Cref{lem:reduction-to-RCP}. The maximum right hand side of
a rectangle, i.e., $\max_{R} \ri(R)$, is bounded by $T$ and $T$ is bounded by $(n + p_{\max})^{O(1)}$
via the preprocessing in \Cref{sec:Reduction}.
Hence, this implies \Cref{lem:tardiness-reduction}.

\subsection{Algorithm}\label{sec:TardinessRCP}

In this section we prove Lemma~\ref{lem:tardiness-algorithm}. Suppose
that we are given a $\delta$-well-structured instance of RCP. Our
approach to solve it is similar to our algorithm for general case of RCP.
As before, let $p_{\max}=\max_Rp(R)$. We start with a preprocessing step. This time, we argue that it suffices to construct an algorithm for which $C:=\frac{\max_{R\in \R}c(R)}{\min_{R\in \R}c(R)}=O(K\cdot n/\varepsilon)$.
\begin{lem}\label{lem:tardinessPreprocessing}
	Assume that there is an $\alpha$-approximation algorithm for $\delta$-well-structured instances of RCP
	with a running time of $2^{( 1/\delta \cdot\log C \cdot K \cdot \log (n \cdot p_{\max}\cdot \max_{R}\ri(R)))^{O(K)}}$. Then there
	is an $(1+\epsilon)\alpha$-approximation algorithm for $\delta$-well-structured
	instances of RCP with a running time of $2^{(1/\delta \cdot K \cdot \log (n \cdot p_{\max}\cdot \max_{R}\ri(R))/\varepsilon)^{O(K)}}$.
\end{lem}
\begin{proof}
	Let $\R$, $\L$ denote the rectangles and rays of a  $\delta$-well-structured instance $I$ and let $S$ be an optimal solution to this instance. 
	Then we guess the most expensive rectangle $R_{\max}\in S$ in the optimal solution. 
	Let $\W^{\operatorname{cheap}}$ denote all rows in which there exists a rectangle of cost at most $\varepsilon\cdot c(R_{\max})/(n \cdot K)$. 
	Let $\R^{\operatorname{cheap}}$ denote all rectangles in the rows $\W^{\operatorname{cheap}}$. 
	We select $\R^{\operatorname{cheap}}$. 
	As for each row $W\in \W^{\operatorname{cheap}}$ we have $\sum_{R\in W}c(R)\leq K\cdot \varepsilon\cdot c(R_{\max})/(n \cdot K)\leq \varepsilon c(R_{\max})/n$, this implies $c(\R^{\operatorname{cheap}})\leq \varepsilon c(R_{\max})$.
	
	Let $\R^{\operatorname{discard}}$ denote all rectangles $R$ with $c(R)>c(R_{\max})$ and all rectangles in the same row on the right of such a rectangle. Note that the optimal solution cannot select any rectangle in $\R^{\operatorname{discard}}$ as it contains a prefix from each row. 
	Next, we define a set of rays $\L'$. For each ray $L(s,t)\in \L$ we add a ray $L'(s,t)\in \L'$ with a demand of $d(L'(s,t))=\max\{0,d(L(s,t))-p(\{R\in \R^{\operatorname{cheap}}: R\cap L(s,t)\neq \emptyset\})\}$. Then we use the given algorithm for the instance $I'$ with rectangles $\R\setminus(\R^{\operatorname{cheap}}\cup \R^{\operatorname{discard}})$ and the rays $\L'$. 
	Given a solution $\APX$ to this instance, we output $\APX\cup \R^{\operatorname{cheap}}$.
	The most expensive rectangle in $I'$ has a cost of $c(R_{\max})$ and the cheapest rectangle has a cost of at least  $\varepsilon\cdot c(R_{\max})/(n \cdot K)$, so we have $C\leq K\cdot n/\varepsilon$ for the instance $I'$.
	
	Clearly $\APX\cup \R^{\operatorname{cheap}}$ is a feasible solution for the instance $I$. Also $S\setminus \R^{\operatorname{cheap}}$ is a feasible solution for $I'$. So $c(\APX)\leq \alpha \cdot c(S\setminus \R^{\operatorname{cheap}})$, which implies $c(\APX\cup \R^{\operatorname{cheap}})\leq \alpha \cdot c(S\setminus \R^{\operatorname{cheap}}) +  \varepsilon c(R_{\max})\leq c(S\setminus \R^{\operatorname{cheap}}) +  \varepsilon c(S)\leq (1+\varepsilon)\alpha c(S)$.
	This shows that we obtain a $(1+\varepsilon)\alpha$-approximation. 
	As we need to call the given algorithm at most $n$ times (once for each choice of $R_{\max}$) and $C\leq K\cdot n/\varepsilon$ for the instance $I'$, we obtain the claimed running time.
\end{proof}

Our algorithm is based
on a recursion in which the input of each recursive call consists
of
\begin{itemize}
\item \label{tardinessPropertyMaintained1}an instance of RCP defined by
a set of rays $\L'$ and a set of rectangles $\R'$; we denote by $\rows'$
be the partition of $\R'$ into rows,
\item an area $A$ of the form $A=[\lef(A),\ri(A))\times[0,\infty)$ for
two values $\lef(A),\ri(A)\in\N_{0}$ such that $\ri(A)-\lef(A)=2{}^{k}$
for some $k\in\mathbb{N}_{0}\cup\{-1\}$ and $\lef(A)$ and $\ri(A)$
are integral multiples of $2^{k}$,
\item each rectangle $R\in\R'$ is contained in $A$, i.e., $R\subseteq A$,
\item each ray $L\in\L'$ is contained in $[\lef(A),\ri(A))\times\R$.
\end{itemize}
The recursive call returns a solution to the RCP instance defined
by $\R'$ and $\L'$. At the end of our algorithm, we output the solution
returned by the (main) recursive call in which $\R':=\R$, $\L':=\L$,
$\lef(A)=0$, and $\ri(A):=T$ where $T$ is the smallest power of
2 that is at least $\max_{R\in\R}\ri(R)$. Any solution to this subproblem
is a solution to our given instance.

Assume we are given a recursive call as defined above. The base case
of our recursion arises when $\ri(A)-\lef(A)=1/2$. Then $\R'=\emptyset$
since by assumption each rectangle in $\R'$ has integer coordinates
and is contained in $A$. Therefore, the subproblem is feasible if
an only if for each ray $L\in\L'$ we have that $d(L)=0$. Assume
now that $\ri(A)-\lef(A)\geq1$. As for the general case of RCP, we
split the problem into a left and a right subproblem at the vertical
line $\mi(A)\times\RR$ where $\mi(A):=(\ri(A)-\lef(A))/2$. Formally,
the areas for the left and right subproblem are defined by $A_{\lsub}:=[\lef(A),\mi(A))\times[0, \infty)$
and $A_{\rsub}:=[\mi(A),\ri(A))\times[0, \infty)$, respectively. In
the following, we will define the rectangles $\R'_{\lsub}$ and $\R'_{\rsub}$
for the left and right subproblem, respectively, and select certain
rectangles from $\R'\setminus(\R'_{\lsub}\cup\R'_{\rsub})$. Finally
define the rays $\L'_{\lsub}$ and $\L'_{\rsub}$ for the left and
right subproblem, respectively,

First, we define that $\R'_{\lsub}$ contains all rectangles from each row
$W\in\rows'$ in which \emph{each }rectangle $R\in W$ is contained
in $A_{\lsub}$, i.e., $R\subseteq A_{\lsub}$. Symmetrically, we
define that $\R'_{\rsub}$ contains all rectangles from each row $W\in\rows'$
in which each rectangle $R\in W$ is contained in $A_{\rsub}$, i.e.,
$R\subseteq A_{\rsub}$. Consider now all remaining rows in $\rows'$,
i.e., let $\rows^{\cross}\subseteq\rows'$ be the set of all rows
$W\in\rows'$ which contain a rectangle $R\in W$ with $R\cap A_{\lsub}\ne\emptyset$
and a rectangle $R'\in W$ with $R\cap A_{\rsub}\ne\emptyset$.

We partition the rows in $\rows^{\cross}$ into groups. Intuitively,
the rectangles of rows in the same group are vertical translates of
each other, have approximately the same cost, and the same value. Formally,
a group $\rows_{g}^{\cross}\subseteq\rows^{\cross}$ is specified by a tuple
$g=(c_{1},\dots c_{k},p,t_{0},\dots,t_{k})$ of integers where $k\leq K$
(recall that $K$ is an upper bound on the number of rectangles in
a row). Consider a row $W\in\rows^{\cross}$ and assume that $R_{1},\dots R_{k'}$
are its rectangles such that $\lef(R_{i})<\lef(R_{i+1})$ for each
$i\in[k'-1]$. We define that the row $W$ is in group $\rows_{g}^{\cross}$
if and only if
\begin{itemize}
\item $k=k'$,
\item $(1+\varepsilon)^{p}\leq p(R_{1})<(1+\varepsilon)^{p+1}$, and
\item for $i\leq k$ we have that $(1+\varepsilon)^{c_{i}}\leq c(R_{i})\leq(1+\varepsilon)^{c_{i}+1}$,
$\lef(R_{i})=t_{i-1}$ and $\ri(R_{i})=t_{i}$.
\end{itemize}
Let $\mathcal{G}:=\{g\in\Z^{4}\cup\Z^{6}\cup...\cup\Z^{2K+2}:\rows_{g}^{\cross}\neq\emptyset\}$
denote all vectors $g$ specifying a non-empty group $\rows_{g}^{\cross}$.
Next, we show that there is only a quasi-polynomial number of groups
in $\mathcal{G}$. 
\begin{lemma} \label{lem:tardinessGroupsNumber}
	We have that $|\mathcal{G}|\le(1/\delta\cdot \log C \cdot \log n \cdot \log p_{\max} \cdot \log T/\varepsilon)^{O(K)}$.
\end{lemma}
\begin{proof}
	Let $g\in \G$ and consider a row $W\in \rows^{\cross}_g$.  Let $R_\mi\in W$ be the rectangle intersecting $\mi(A)\times \RR$, i.e.\ with $\lef(R)\le\mi(A)$ and $\ri(R)> \mi(A)$.
	We can limit the number of options for $g$ as follows:
	\begin{itemize}
		\item There are $K$ options for the number $k$ of rectangles in row $W$.
		\item Recall that $C=\frac{\max_{R\in \R}c(R)}{\min_{R\in \R}c(R)}$. So there are only only $O((\log C)/\varepsilon)$ different integers $c_i$ for which there exists a rectangle $R\in \R$ with $(1+\varepsilon)^{c_i}\leq c(R)\leq (1+\varepsilon)^{c_i+1}$. Thus there are $O(((\log C)/\varepsilon)^K)$ options for $c_1, \dots, c_k$.
		\item There are $O((\log p_{\max})/\varepsilon)$ options for $p$ as $p(R)\leq p_{\max}$ for all $R\in \R$.
		\item The length of any rectangle $R\in W$ is at most $T$. So there are $\log T$ possible values of $h$ such that $2^h\leq \ri(R)-\lef(R)\leq 2^{h+1}$. As the instance is $\delta$-well-structured, there are $O((\log T)/\delta)$ possible lengths for any rectangle. So there are $O((\log T)/\delta)$ options for the length of $R_\mi$.
		\item The rectangle $R_\mi$ fulfills $\lef(R)\le\mi(A)$ and $\ri(R)> \mi(A)$ and as the instance is $\delta$-well-structured, there are only $1/\delta$ options for $\lef(R_\mi)$ when the length of $R_\mi$ is fixed.
		\item There are $k$ options which of the rectangles $R_1, \dots, R_k$ is $R_\mi$.
		\item The rectangles in a row are adjacent. So given the exact position for one rectangle (in this case $R_\mi$), it suffices to enumerate all possible lengths of the other rectangles, as this determines their position. As there are $O(\log T/\delta)$ options for the length of a single rectangle, there are $O((\log T/\delta)^{K-1})$ options for the rectangles $W\setminus \{R_\mi\}$.
	\end{itemize}
	So it suffices to know $k, c_1, \dots, c_k, p$, the lengths of the rectangles $R_1, \dots, R_k$, the position $\lef(R_\mi)$ and the information which of the rectangles $R_1, \dots, R_k$ is $R_\mi$ to determine the group. So $|\G|=O(K\cdot ((\log C)/\varepsilon)^K \cdot (\log p_{\max})/\varepsilon \cdot  ((\log T)/\delta)^K \cdot 1/\delta \cdot K)=(1/\delta\cdot \log C \cdot \log p_{\max} \cdot (\log T)/\varepsilon)^{O(K)}$.
\end{proof}
Let $S$ be an optimal solution to the given subproblem. Intuitively,
we want to guess the rectangles in $S$ in each row $\rows_{g}^{\cross}$
for some $g\in\mathcal{G}$. Consider a group $\rows_{g}^{\cross}$
for some $g\in\mathcal{G}$. Essentially, the main difference of the
rows within a group is their $y$-coordinate, since for the rectangles
in such a row the costs and the values are almost the same by the
definition of the groups. When two rectangles differ only by their
$y$-coordinate, every ray $L\in\L'$ intersecting with the top rectangle
also intersects with the bottom one. Therefore, intuitively it is
sufficient to select the same number of rectangles from the group $g$
as $S$ and select them greedily from bottom to top, if we select
by a factor $1+\epsilon$ more rectangles to compensate for the fact
that their values may differ by a factor $1+\epsilon$.

Formally for the group $g$, we do the following: For each $0\leq i\leq k$
let $\rows_{g,i}^{\cross}$ denote the rows in $\rows_{g}^{\cross}$
from which $S$ contains exactly $i$ rectangles. Let $n_{g,i}:=|\rows_{g,i}^{\cross}|$
denote the number of such rows. First, we guess the number $n_{g,i}$
for each $i\in[K]$. For each $i\in[K]$ with $n_{g,i}\leq1/\varepsilon$
we simply guess the corresponding rows $\rows_{g,i}^{\cross}$ and
for each row $W\in\rows_{g,i}^{\cross}$ we select its leftmost $i$
rectangles. Let $\APX_{g,i}$ denote the resulting rectangles. Let
$\rows_{g}^{\done}=\bigcup_{i\in [k]:n_{g,i}\leq1/\varepsilon}\rows_{g,i}^{\cross}$
denote all rows, from which we already guessed the rectangles in $S$.

We go through the indices $i$ with $n_{g,i}\geq1/\varepsilon$ in
decreasing order, i.e., from $K$ to 1. Consider an index $i\in[K]$.
Intuitively, in the $\lfloor(1+2\varepsilon)n_{g,i}\rfloor$ bottom
rows where no rectangle is selected yet, we add the first $i$ rectangles
to our solution. Formally, let $\rows_{g,i}^{*}$ denote the bottom-most
$\lfloor(1+2\varepsilon)n_{g,i}\rfloor$ rows in $\rows_{g}\setminus(\rows_{g}^{\done}\cup\bigcup_{i':i<i'\leq K,n_{g,i'}>1/\varepsilon}\rows_{g,i'}^{*})$.
Now let $\APX_{g,i}$ denote all rectangles that are within the first
$i$ rectangles in some row $W\in\rows_{g,i}^{*}$. We also select
$\APX_{g,i}$. We define $\APX_{g}:=\bigcup_{1\leq i\leq k}\APX_{g,i}$
and $S_{g}:=\bigcup_{W\in\rows_{g}}S\cap W$. In the following lemma
we prove that $\APX_{g}$ is not much more expensive than $S_{g}$. 

\begin{lemma}\label{lem:tardinessAPXgCost}
We have that $c(\APX_{g})\leq(1+5\varepsilon)c(S_{g})$.
\end{lemma}
\begin{proof}
	Let $g=(c_{1},\dots c_{k},p,t_{0},\dots,t_{k})\in \G$. By construction we have $\APX_g\cap \bigcup_{W\in \rows^{\done}_g}=S_g\cap \bigcup_{W\in \rows^{\done}_g}W$. So we can intuitively ignore the rows $\rows^{\done}_g$.
	Furthermore for each $i\leq k$ with $n_{g, i}>1/\varepsilon$ we have that $\APX_{g, i}$ contains the first $i$ rectangles from $\lfloor(1+2\varepsilon)n_{g,i}\rfloor$ rows. Due to the definition of $g$, we have that the cost of the first $i$ rectangles in $g$ is at least $\sum_{i'\leq i}(1+\varepsilon)^{c_{i'}}$ and at most $\sum_{i'\leq i}(1+\varepsilon)^{c_{i'}+1}$.
	So $c(\APX_{g, i})\leq (1+2\varepsilon)n_{g,i}\sum_{i'\leq i}(1+\varepsilon)^{c_{i'}+1}$ and  $c(S\cap \bigcup_{W\in\rows^{\cross}_{g, i}}W)\geq n_{g,i}\sum_{i'\leq i}(1+\varepsilon)^{c_{i'}}$. 
	This implies $c(\APX_{g, i})\leq (1+\varepsilon)(1+2\varepsilon)c(S\cap \bigcup_{W\in\rows^{\cross}_{g, i}}W)\leq (1+5\varepsilon)c(S\cap \bigcup_{W\in\rows^{\cross}_{g, i}}W)$.
	So we obtain
	\begin{align*}
		c(\APX_{g})&\leq c(\APX_g\cap \bigcup_{W\in \rows^{\done}_g})+\sum_{i\in [k]:n_{g, i}>1/\varepsilon}c(\APX_{g, i})\\
		&\leq c(S_g\cap \bigcup_{W\in \rows^{\done}_g})+\sum_{i\in [k]:n_{g, i}>1/\varepsilon}(1+5\varepsilon)c(S\cap \bigcup_{W\in\rows^{\cross}_{g, i}}W)\\
		&\leq (1+5\varepsilon)c(\sum_{i\leq k}\rows^{\cross}_{g, i})\\
		&=(1+5\varepsilon)c(S_g)
	\end{align*}
This completes the proof.
\end{proof}
Next, we also show that it covers the same amount on any ray as $S_g$, so it is a good substitute for $S_g$.
\begin{lemma}\label{lem:tardinessAPXgRays}
	For any ray $L(s,t)\in \L'$, we have that $\sum_{R\in\APX_{g}:R\cap L(s,t)\neq\emptyset}p(R)\geq\sum_{R\in S_{g}:R\cap L(s,t)\neq\emptyset}p(R)$.
\end{lemma}
\begin{proof}
Let $g=(c_{1},\dots c_{k},p,t_{0},\dots,t_{k})\in \G$. By construction we have $\APX_g\cap \bigcup_{W\in \rows^{\done}_g}=S_g\cap \bigcup_{W\in \rows^{\done}_g}W$. So we can intuitively ignore the rows $\rows^{\done}_g$.
Let $L(s,t)\in \L'$ be a ray. 
In each row $W\in\rows^{\done}_g$ the sets $\APX_g$ and $S_g$ contain the same rectangles, so they also cover the same amount on $L(s,t)$. 
Suppose that there is no row $W\in \rows_g \setminus \rows_g^{\done}$ with a rectangle $R\in W\setminus \APX_g$ and $R\cap L(s,t)\neq \emptyset$. 
Then every rectangle $\{R\in S_g:R\cap L(s,t)\neq \emptyset\}$ is also in $\APX_g$, which yields the lemma.
So suppose now that there is a row $W'\in \rows_g \setminus \rows_g^{\done}$ such that there exists a rectangle $R'\in W'\setminus \APX_g$ with $R\cap L(s,t)\neq \emptyset$.
Let $i'\leq k$ such that $t_{i'-1}\leq t <t_{i'}$. 
Then $\lef(R')=t_{i'-1}$ and $\ri(R')=t_{i'}$.
For each $i\geq i'$, we have $W'\not \in \rows_{g, i}^*$ as $\APX_g$ contains the first $i$ rectangles from each row in $\rows_{g, i}^*$, but not the $i$-th rectangle $R'$ in row $W'$.
Recall that $\rows_{g,i}^{*}$ denotes the bottom-most
	$\lfloor(1+2\varepsilon)n_{g,i}\rfloor$ rows in $\rows_{g}\setminus(\rows_{g}^{\done}\cup\bigcup_{i':i<i'\leq K,n_{g,i'}>1/\varepsilon}\rows_{g,i'}^{*})$. This implies that all rows in $\rows_{g, i}^*$ are below $W'$.
As the $i'$-th rectangle in each such row is in $\APX_g$ and intersects $L(s,t)$, we have
$|\{R\in \APX_g \setminus \bigcup_{W\in \rows^{\done}_g }W: R\cap L(s,t)\neq \emptyset\}|\geq \sum_{i\in \{i', \dots, k\}: n_{g, i}>1/\varepsilon}\lfloor(1+2\varepsilon)n_{g,i}\rfloor$. 
By definition of $n_{g,i}$ we have $|\{R\in S_g \setminus \bigcup_{W\in \rows^{\done}_g }W:R\cap L(s,t)\neq \emptyset\}|\leq \sum_{i\in \{i', \dots, k\}: n_{g,i}>1/\varepsilon}n_{g,i}$. Recall that each rectangle $R$ in some row $W\in \rows_g$ fulfills $(1+\varepsilon)^p\leq p(R)<(1+\varepsilon)^{p+1}$. This implies
\begin{align*}
	p(\{R\in \APX_g \setminus \bigcup_{W\in \rows^{\done}_g }W: R\cap L(s,t)\neq \emptyset\}|)
	&\geq (1+\varepsilon)^{p}\sum_{i\in \{i', \dots, k\}: n_{g, i}>1/\varepsilon}\lfloor(1+2\varepsilon)n_{g,i}\rfloor\\
	&\geq (1+\varepsilon)^{p+1}\sum_{i\in \{i', \dots, k\}: n_{g, i}>1/\varepsilon}n_{g,i}\\
	&\geq p(\{R\in S_g \setminus \bigcup_{W\in \rows^{\done}_g }W:R\cap L(s,t)\neq \emptyset\})
\end{align*}
And as $\{R\in \APX_g \cap \bigcup_{W\in \rows^{\done}_g }W\}=\{R\in S_g \cap \bigcup_{W\in \rows^{\done}_g }W\}$, this completes the proof of the lemma.
\end{proof}
Finally, we define the rays $\mathcal{L}'_{\lsub}$ and $\mathcal{L}'_{\rsub}$
for the left and right subproblem, respectively. Intuitively, we define
that $\mathcal{L}'_{\lsub}$ contain all rays in $\L'$ that intersect
$A_{\lsub}$, but we adjust their demands if they intersect some rectangles
in $\APX_{g}$ for some groups $g\in\mathcal{G}$. Formally, for each
ray $L(s,t)\in\L'$ with $(t,s)\in A_{\lsub}$ we
add a
ray $L'(s,t)$ to $\L'_\lsub$
with a demand of $d(L'(s,t)):=\max\{0,d(L(s,t))-\sum_{g\in\G}\sum_{R\in\APX_{g}:R\cap L(s,t)\neq\emptyset}p(R)\}$.
We define the rays $\mathcal{L}'_{\rsub}$ analogously. We recursively
compute solutions $\APX_{\lsub}$ and $\APX_{\rsub}$ for the left
and right subproblem. If these solutions exist, we obtain the candidate solution  $\bigcup_{g\in\mathcal{G}}\APX_{g}\cup\APX_{\lsub}\cup\APX_{\rsub}$. We output the minimal cost candidate solution $\APX$ among all guesses.
This completes the description of the algorithm.

\subsection{Analysis}

We show that the computed solution is feasible, has small cost, and
that our algorithm has a running time of $2^{(1/\delta \cdot\log C \cdot  K \cdot \log (n \cdot p_{\max}\cdot T)/\varepsilon)^{O(K)}}$. 
First, we show that $S$ restricted to the rectangles
in $\R'_{\lsub}$ is a feasible solution for the left subproblem and
$S$ restricted to the rectangles in $\R'_{\rsub}$ is a feasible
solution for the right subproblem, as the following lemma shows.
\begin{lemma}\label{lem:left-right-subproblems} Let $S_{\lsub}=\R'_{\lsub}\cap S$
	and $S_{\rsub}=\R'_{\rsub}\cap S$. Then $S_{\lsub}$ is a feasible
	solution for the left subproblem and $S_{\rsub}$ is a feasible solution
	for the right subproblem. 
\end{lemma}
\begin{proof}
	We prove that $S_{\lsub}$ is a feasible solution for the left subproblem, the proof that $S_{\rsub}$ is a feasible solution for the right subproblem is analogous to the proof for $S_\lsub$.
	From the definition of $S_{\lsub}$ and the fact that $S$ contains a prefix in each row, it follows that $S_\lsub$ contains a prefix in each row as well.
	So it suffices to show that $S_\lsub$ covers all rays an $\L'_\lsub$. Let $L'(s,t)\in \L'_\lsub$ be a ray. Then there is a ray $L(s,t)\in \L'$ and $d(L'(s,t)):=\max\{0,d(L(s,t))-\sum_{g\in\G}\sum_{R\in\APX_{g}:R\cap L(s,t)\neq\emptyset}p(R)\}$. Note that there is no rectangle $R\in \R'_\rsub$ with $R\cap L(s,t)\neq \emptyset$. So
	\begin{align*}
		p(\{R\in S_\lsub:R\cap L(s,t)\neq \emptyset\})
		&=p(\{R\in S:R\cap L(s,t)\neq \emptyset\})-\sum_{g\in \G}p(\{R\in S_g:R\cap L(s,t)\neq \emptyset\})\\
		&\geq d(L(s,t))-\sum_{g\in \G}p(\{R\in \APX_g:R\cap L(s,t)\neq \emptyset\})\\
		&= d(L'(s,t))\\
	\end{align*}
	where the inequality follows from Lemma~\ref{lem:tardinessAPXgRays}. This show that $S_\lsub$ is a feasible solution and thus completes the proof.
\end{proof}
So we actually obtain a solution $\APX$.
Now we prove that $\APX$ is feasible.
\begin{lemma}\label{lem:tardinessFeasibility} 
	The set $\APX$ is a feasible solution. 
\end{lemma}
\begin{proof}
Suppose that we recursively obtained solutions $\APX_\lsub$ and $\APX_\rsub$ for the left and right subproblem. Note that by Lemma~\ref{lem:left-right-subproblems} this is at least the case for the correct guesses. Then $\APX_\lsub$ and $\APX_\rsub$ contain a prefix of the rectangles in each row, as they are feasible solutions to the respective subproblems. For each $g\in \G$ the set $\APX_g$ contains a prefix from each row by construction.
Thus, also $\APX$ contains a prefix from each row.
So it remains to show that every ray $L(s,t)\in \L'$ is covered. Suppose that $L(s,t)$ intersects $A_\lsub$, as the proof for the case when $L(s,t)$ intersects $A_\rsub$ is analogue. Then there is a ray $L'(s,t)\in \L'_\lsub$ with $d(L'(s,t)):=\max\{0,d(L(s,t))-\sum_{g\in\G}\sum_{R\in\APX_{g}:R\cap L(s,t)\neq\emptyset}p(R)\}$.
As $\APX_\lsub$ is a feasible solution for the left subproblem, we have
$p(\{R\in \APX_\lsub:R\cap L(s,t)\neq \emptyset\})\geq d(L'(s,t))$. And as $\APX_\lsub$ and the sets $\APX_g$ for $g\in \G$ are disjoint, we have
\begin{align*}
	p(\{R\in \APX:R\cap L(s,t)\neq \emptyset\})&=p(\{R\in \APX_\lsub:R\cap L(s,t)\neq \emptyset\})+\sum_{g\in \G}p(\{R\in \APX_g:R\cap L(s,t)\neq \emptyset\})\\
	&\geq d(L'(s,t))+\sum_{g\in \G}p(\{R\in \APX_g:R\cap L(s,t)\neq \emptyset\})\\
	&\geq d(L(s,t))
\end{align*}
This shows that $L(s,t)$ is covered and thus completes the proof.
\end{proof}
Next, we bound the cost of $\APX$. Lemma~\ref{lem:tardinessAPXgCost}
already bounded the cost of $\APX_{g}$ for each $g\in \G$. Thus, it remains only to
bound the cost of $\APX_{\lsub}$ and $\APX_{\rsub}$. 
By combining Lemmas~\ref{lem:tardinessAPXgCost} and \ref{lem:left-right-subproblems}
we obtain an approximation factor of $1+5\varepsilon$.
\begin{lemma}\label{lem:tardinessApproximationfactor}
We have that $c(\APX)\leq(1+5\varepsilon)c(S)$. 
\end{lemma}
\begin{proof}
We prove the lemma by induction on $\lef(A)-\ri(A)$. In the base case where $\lef(A)-\ri(A)=1/2$, we solve the problem optimally as then $\APX=S_g=\emptyset$.
So suppose $\lef(A)-\ri(A)>1/2$. Then by induction and Lemma~\ref{lem:left-right-subproblems}, we have that $c(\APX_\lsub)\leq (1+5\varepsilon)c(S_\lsub)$ and $c(\APX_\rsub)\leq (1+5\varepsilon)c(S_\rsub)$. By Lemma~\ref{lem:tardinessAPXgCost} we have $c(\APX_g)\leq (1+5\varepsilon)c(S_g)$. Altogether, we obtain 
\begin{align*}
	c(\APX)&=c(\APX_\lsub)+c(\APX_\rsub)+\sum_{g\in \G}c(\APX_g)\\
	&\leq (1+5\varepsilon)c(S_\lsub)+(1+5\varepsilon)c(S_\rsub)+\sum_{g\in \G}(1+5\varepsilon)c(S_g)\\
	&= (1+5\varepsilon)c(S)
\end{align*}
\end{proof}
Thus it only remains to bound the running time.
\begin{lemma} \label{lem:tardinessRunningtime}
	The running time of the algorithm can be bounded by $2^{(1/\delta \cdot \log C \cdot  K \cdot \log (n \cdot p_{\max}\cdot T)/\varepsilon)^{O(K)}}$. 
\end{lemma}
\begin{proof}
	The recursion depth of our algorithm is $O(\log T)$. In each recursive call we have $|\G|\leq (1/\delta\cdot \log C\cdot \log p_{\max} \cdot \log T/\varepsilon)^{O(K)}$. 
	For each $g\in \G$ and each $i\in[K]$, we need to guess $n_{g, i}$ and if $n_{g, i}\leq 1/\varepsilon$ we also guess the rows $\rows^{\cross}_{g, i}$. 
	As $n_{g, i}\leq n$ and there are at most $n$ rows, this can be done in $n^{O(K/\varepsilon)}$ per group $g$. 
	For fixed guesses the remaining computation can be done in time $O(n)$. 
	Thus, the total running time is $((n^{O(K/\varepsilon)})^{|\G|})^{O(\log T)}$. By Lemma~\ref{lem:tardinessGroupsNumber}, this can be bounded by $2^{(1/\delta \cdot\log C \cdot  K \cdot \log (n \cdot p_{\max}\cdot T)/\varepsilon)^{O(K)}}$.
\end{proof}
Altogether, this implies Lemma~\ref{lem:tardiness-algorithm}.
\begin{proof}[Proof of Lemma~\ref{lem:tardiness-algorithm}]
	By Lemma~\ref{lem:tardinessFeasibility} the algorithm computes a feasible solution to every subproblem, so also to the root subproblem which is the output of the algorithm. 
	By Lemma~\ref{lem:tardinessApproximationfactor} we have $c(\APX)\leq (1+5\varepsilon)c(S)$, we obtain a $(1+5\varepsilon)$-approximation algorithm. 
	By Lemma~\ref{lem:tardinessRunningtime}, the running time is bounded by $2^{(1/\delta \cdot\log C \cdot  K \cdot \log (n \cdot p_{\max}\cdot T)/\varepsilon)^{O(K)}}$. So we can apply Lemma~\ref{lem:tardinessPreprocessing}, which yields the claimed result by rescaling~$\varepsilon$.
\end{proof}
\newpage

\bibliographystyle{plain}
\bibliography{references}
\end{document}